\documentclass[letterpaper,twocolumn]{article} 
\usepackage{times}  
\usepackage{helvet}  
\usepackage{courier}  
\usepackage[hyphens]{url}  
\usepackage{graphicx} 
\usepackage{natbib}  
\usepackage{caption} 
\newcommand{\longversion}[1]{}
\newcommand{\shortversion}[1]{#1}
\newcommand{\futuresketch}[1]{}
\usepackage{times}
\usepackage{soul}
\usepackage[hidelinks]{hyperref}
\usepackage{url}
\usepackage[utf8]{inputenc}
\usepackage{multirow}
\usepackage{amssymb,amsfonts,amsmath,amsthm}%
\newcommand{\footnoteitext}[1]{\stepcounter{footnote}
  \footnotetext[\thefootnote]{#1}}
\usepackage{enumitem}
\usepackage{stackengine}
\usepackage{calc}
\newlength\shlength
\newcommand\xshlongvec[2][0]{\setlength\shlength{#1pt}%
  \stackengine{-5.6pt}{$#2$}{\smash{$\kern\shlength%
    \stackengine{7.55pt}{$\mathchar"017E$}%
      {\rule{\widthof{$#2$}}{.57pt}\kern.4pt}{O}{r}{F}{F}{L}\kern-\shlength$}}%
      {O}{c}{F}{T}{S}}

\usepackage{xspace}
\usepackage[ruled,vlined,linesnumbered]{algorithm2e}
\usepackage[table,x11names]{xcolor}
\usepackage{graphicx}  
\usepackage{caption}

\SetKwInput{KwData}{In}
\SetKwInput{KwResult}{Out}
\SetAlFnt{\small}
\SetAlCapFnt{\small}
\SetAlCapNameFnt{\small}
\SetAlCapHSkip{0pt}
\SetEndCharOfAlgoLine{}

\makeatother

\makeatletter
\newcommand*{\inlineequation}[2][]{%
  \begingroup
    \refstepcounter{equation}%
    \ifx\\#1\\%
    \else
      \label{#1}%
    \fi
    \relpenalty=10000 %
    \binoppenalty=10000 %
    \ensuremath{%
      #2%
    }%
    ~\@eqnnum
  \endgroup
}

\newcommand{\problemFont}[1]{\textsc{#1}}

\newcommand{\mtext}[1]{\ensuremath{\mathcal{#1}}}
\newcommand{\cntc}[0]{\ensuremath{\#\cdot}}

\DeclareMathOperator{\poly}{poly}
\DeclareMathOperator{\rank}{rank}
\DeclareMathOperator{\ord}{ord}
\DeclareMathOperator{\block}{bn}

\DeclareMathOperator{\SP}{ProofStates}
\DeclareMathOperator{\checkord}{isLvl}
\DeclareMathOperator{\checkmod}{CheckMod}
\DeclareMathOperator{\gatherproof}{proven}

\DeclareMathOperator{\possord}{lvl}
\newcommand{\MAI}[2]{\ensuremath{#1^+_{#2}}}%

\newcommand{\QBFSAT}{\textsc{QSat}\xspace}

\DeclareMathOperator{\tower}{\ensuremath{\mathsf{tow}}}

\usepackage{microtype}
\usepackage{tikz}
\usetikzlibrary{shapes}
\usetikzlibrary{calc}

\newcommand{\citex}[1]{\citeauthor{#1}~\shortcite{#1}}
\newcommand{\citey}[1]{\citeauthor{#1},~\citeyear{#1}}

\renewcommand{\P}{\ensuremath{\textsc{P}}\xspace}
\newcommand{\NP}{\ensuremath{\textsc{NP}}\xspace}
\newcommand{\SIGMA}[2]{\ensuremath{\Sigma_{\textrm{#1}}^{\textrm{#2}}}}

\newcommand{\HCF}{\text{HCF}\xspace}

\usetikzlibrary{positioning}

\tikzstyle{tdnode} = [draw,rounded corners,top color=vertexTopColor,bottom color=vertexBottomColor,minimum size=1.5em]
\tikzstyle{stdnode} = [tdnode, font=\scriptsize]
\tikzstyle{stdnodecompact} = [stdnode, inner sep = 1.5pt, outer sep = 0.1pt]
\tikzstyle{stdnodetable} = [stdnode, inner sep = 0.5pt, outer sep = 0]
\tikzstyle{stdnodenum} = [minimum size=1.5em, font=\scriptsize]
\tikzstyle{tdedge} = [-,draw,thick]
\tikzstyle{tdlabel} = [draw=none, rectangle, fill=none, inner sep=0pt, font=\scriptsize]
\colorlet{vertexTopColor}{white}
\colorlet{vertexBottomColor}{black!10}

\usepackage{ctable}
\usepackage{ifthen}
\newif\iflong
\newcommand{\restrict}[2]{\ensuremath{#1\cap #2}}

\usepackage{float}

\newcommand{\SB}{\{}%
\newcommand{\SM}{\mid}%
\newcommand{\SE}{\}}%
\def\hy{\hbox{-}\nobreak\hskip0pt}

\newcommand{\mdpa}[1]{\ensuremath{\mathtt{PCNT}_{#1}}}
\newcommand{\ta}[1]{\ensuremath{2^{#1}}}
\newcommand{\Card}[1]{\left|#1\right|}
\newcommand{\CCard}[1]{\|#1\|}
\newcommand{\algo}[1]{\ensuremath{\mathbb{#1}}}
\DeclareMathOperator{\width}{width}
\DeclareMathOperator{\children}{chldr}

\DeclareMathOperator{\pred}{prev}
\DeclareMathOperator{\prev}{prev}

\DeclareMathOperator{\dom}{dom}
\DeclareMathOperator{\rootOf}{root}

\DeclareMathOperator{\scc}{scc}

\usepackage{mathtools}
\DeclarePairedDelimiter\ceil{\lceil}{\rceil}

\newcommand{\bvali}[3]{\ensuremath{[\![#1]\!]_{#2,#3}}}

\makeatletter
\newcommand{\algorithmfootnote}[2][\footnotesize]{
  \let\old@algocf@finish\@algocf@finish
  \def\@algocf@finish{\old@algocf@finish
    \leavevmode\rlap{\begin{minipage}{\linewidth}
    #1#2
    \end{minipage}}
  }
}
\makeatother

\DeclareMathOperator{\orig}{\algo{A}\hy origins}
\DeclareMathOperator{\origs}{\algo{A}\hy origins}
\newcommand{\origa}[1]{\operatorname{#1\hy origins}}
\newcommand{\origse}[1]{\operatorname{#1\hy origins}}
\DeclareMathOperator{\Ext}{Ext}
\DeclareMathOperator{\Exts}{Exts}
\DeclareMathOperator{\PExt}{SatExt}
\DeclareMathOperator{\pmc}{pasc}
\DeclareMathOperator{\ipmc}{ipasc}
\DeclareMathOperator{\bucket}{=_P}%
\DeclareMathOperator{\buckets}{buckets}
\DeclareMathOperator{\subbuckets}{sub\hy buckets}

\newcommand{\TTT}{\ensuremath{\mathcal{T}}}%
\newcommand{\WWW}{\ensuremath{\mathcal{W}}}%
\newcommand{\por}{\vee}

\newcommand{\eqdef}{\ensuremath{\,\mathrel{\mathop:}=}}
\newcommand{\hsep}{\leftarrow\,}
\newcommand{\SAT}{\textsc{SAT}\xspace}
\newcommand{\ASP}{\textsc{ASP}\xspace}

\newcommand{\PASP}{\textsc{\#PAs}\xspace}%
\newcommand{\PDASP}{\textsc{\#PDAs}\xspace}%
\newcommand{\AlgA}{\algo{A}}%
\newcommand{\AlgS}{\AlgA}%
\newcommand{\PROJ}{\algo{PROJ}\xspace}
\newcommand{\at}{\text{\normalfont at}}
\newcommand{\var}{\at}
\newcommand{\bigO}[1]{\ensuremath{{\mathcal O}(#1)}}
\newcommand{\CCC}{\ensuremath{\mathcal{C}}}%

\newcommand{\tuplecolor}[1]{\textcolor{#1}}
\newcommand{\inputPredColor}{orange!55!red}
\newcommand{\outputPredColor}{blue!45!black}
\newcommand{\statePredColor}{green!62!black}
\newcommand{\specialPredColor}{red!62!black}

\newcommand{\tabval}{\ensuremath{u}}
\newcommand{\tab}[1]{\ensuremath{\tau_{#1}}}

\newcommand{\att}[1]{\ensuremath{\at_{\hspace{-0.05em}\leq\hspace{-0.05em}#1}}}
\newcommand{\attneq}[1]{\ensuremath{\at_{\hspace{-0.05em}<\hspace{-0.05em}#1}}}

\newcommand{\prog}{\ensuremath{\Pi}}
\newcommand{\progt}[1]{\ensuremath{\prog_{\hspace{-0.05em}\leq\hspace{-0.05em}#1}}}
\newcommand{\progtneq}[1]{\ensuremath{\prog_{\hspace{-0.05em}<\hspace{-0.05em}#1}}}

\newcommand{\dpa}{\ensuremath{\mathtt{DP}}}
\newcommand{\Tab}[1]{\ensuremath{\text{C-Tabs}}}

\def\thyph{\text{-}\penalty0\hskip0pt\relax}

\newcommand{\ATab}[1]{\ensuremath{#1\thyph\text{Comp}}}

\newcommand{\tw}[1]{\mathit{tw}(#1)}

\newcommand{\Nat}{\mathbb{N}} %
\DeclareMathOperator{\type}{type}
\newcommand{\intr}{\textit{int}}
\newcommand{\leaf}{\textit{leaf}}
\newcommand{\inner}{\textit{inner}}
\newcommand{\rem}{\textit{forget}}
\newcommand{\join}{\textit{join}}

\newtheorem{example}{Example}
\newtheorem{conjecture}{Conjecture}
\newtheorem{proposition}{Proposition}
\newtheorem{hypothesis}{Hypothesis}

\newtheorem{observation}{Observation}
\newtheorem{theorem}{Theorem}
\newtheorem{lemma}{Lemma}
\newtheorem{definition}{Definition}
\newtheorem{corollary}{Corollary}

\newenvironment{restateobservation}[1][\unskip]{%
  \begingroup
  \def\theobservation{\ref{#1}} 
}%
{%
  \addtocounter{observation}{-1}
  \endgroup
}%

\newenvironment{restatecorollary}[1][\unskip]{%
  \begingroup
  \def\thecorollary{\ref{#1}} 
}%
{%
  \addtocounter{corollary}{-1}
  \endgroup
}%

\newenvironment{restateproposition}[1][\unskip]{%
  \begingroup
  \def\theproposition{\ref{#1}} 
}%
{%
  \addtocounter{proposition}{-1}
  \endgroup
}%

\newenvironment{restatetheorem}[1][\unskip]{%
  \begingroup
  \def\thetheorem{\ref{#1}} 
}%
{%
  \addtocounter{theorem}{-1}
  \endgroup
}%

\DeclareMathOperator{\pcnt}{pasc}
\DeclareMathOperator{\local}{local}
\DeclareMathOperator{\sipmc}{s-ipasc}
\DeclareMathOperator{\icnt}{ipasc}

\newcommand{\PRIM}{\ensuremath{{\algo{PHC}}}\xspace}
\newcommand{\INC}{\ensuremath{{\algo{INC}}}\xspace}

 \title{Characterizing Structural Hardness of Logic Programs: What makes  Cycles and Reachability Hard for Treewidth?} 
\author{%
  Markus Hecher
}

\begin{document}
%
%
\maketitle
\begin{abstract}
Answer Set Programming (ASP) is a problem modeling and solving framework for several problems in KR with growing industrial applications. 
Also for studies of computational complexity and deeper insights into the hardness and its sources, ASP 
has been attracting 
researchers for many years.
%
%
These studies resulted in fruitful characterizations in terms of complexity classes, 
fine-grained insights 
in form of dichotomy-style results, 
as well as detailed parameterized complexity landscapes.
%
%
%
%
Recently, this lead to a novel result establishing that for the measure treewidth, which captures structural density of a program, the evaluation of the well-known class of normal programs is expected to be slightly harder than deciding satisfiability (SAT).
%
However, it is unclear how to utilize this structural power of ASP.
This paper deals with a novel reduction from SAT to normal ASP that goes beyond well-known encodings: We explicitly utilize the structural power of ASP, whereby we sublinearly decrease the treewidth, which probably cannot be significantly improved.
%
%
%
Then, compared to existing results, this characterizes hardness in a fine-grained way by establishing the required functional dependency of the dependency graph's cycle length (SCC size) 
on the treewidth. 
%
%

\end{abstract}

\section{Introduction}

Answer Set Programming (ASP)~\cite{BrewkaEiterTruszczynski11,GebserKaminskiKaufmannSchaub12} is a declarative problem modeling and solving framework for knowledge representation and reasoning and artificial intelligence in general. 
This makes ASP a key formalism and suitable target language for solving problems in that area effectively, e.g., \cite{BalducciniGelfondNogueira06a,NiemelaSimonsSoininen99,NogueiraBalducciniGelfond01a,GuziolowskiEtAl13a,SchaubWoltran18,AbelsEtAl19}.
Such problems are thereby encoded in a logic program, which is a set of rules describing its solutions by means of so-called answer sets --
an approach that goes beyond satisfying a set of clauses (rules) as in \SAT, but additionally requires justifications for variables (atoms) that are claimed to hold.
Considerable effort has been put into providing extensions and a rich modeling language that can be efficiently evaluated by solvers like clasp~\cite{GebserKaufmannSchaub09a} or wasp~\cite{AlvianoEtAl19}.
\begin{example}[Encoding with \ASP]\label{ex:encoding}
The classical way of encoding satisfiability (\SAT) 
of a formula~$F$ is to guess for each variable~$v\in \var(F)$ its truth value via the two rules $v \leftarrow \neg \hat v$ and~$\hat v \leftarrow \neg v$. Then, for every clause~$l_1 \vee l_2 \vee \ldots \vee l_n$ in~$F$,
an additional constraint ensures the clause: $\bot \leftarrow \bar{l_1}, \bar{l_2}, \ldots, \bar{l_n}$ with~$\bar{l_i}$ being~$v_i$, if~$l_i=\neg v_i$, and~$\hat l_i$ otherwise.
\end{example}

The computational complexity of ASP is fairly well studied, where 
for the \emph{consistency problem} of deciding whether a so-called \emph{normal logic program} admits an answer set is \NP-complete~\cite{BidoitFroidevaux91,MarekTruszczynski91}.
This result also extends to
the class of \emph{head-cycle-free (HCF)}
programs~\cite{Ben-EliyahuDechter94},
but if full disjunctions are allowed in the heads of a rule,
the complexity increases to $\Sigma_2^P$-completeness~\cite{EiterGottlob95}.
Over the time, studying the complexity of ASP raised further attention.
There is a wide range of more fine-grained studies~\cite{Truszczynski11} for ASP, also in the context of parameterized
complexity~\cite{CyganEtAl15,Niedermeier06,DowneyFellows13,
FlumGrohe06}, 
where certain  parameters~\cite{
LacknerPfandler12} 
are taken into account.
In parameterized complexity, the ``hardness'' of a problem 
is classified according to the effort required 
to solve the problem, e.g., runtime dependency, in terms of a certain
\emph{parameter}. 
For \ASP there is growing research on the well-studied and prominent structural parameter treewidth~\cite{JaklPichlerWoltran09,BichlerMorakWoltran18,BliemEtAl20}. 
Intuitively, treewidth yields a \emph{tree decomposition},
which is a structural representation
that can be used
for solving numerous combinatorially hard problems in parts; the treewidth indicates the maximum number of variables of these parts one has to investigate 
during problem solving.

Recently, it has been shown that 
when assuming the \emph{Exponential Time Hypothesis (ETH)}~\cite{ImpagliazzoPaturiZane01}, which implies that \SAT \emph{cannot} be solved in time better than single-exponential in the number of variables in the worst case,
normal \ASP seems to be slightly harder~\cite{Hecher22} for treewidth than \SAT. 
More precisely, (i) normal \ASP can be solved in time~$2^{{\mathcal{O}(k\cdot\log(k))}} \cdot \poly(n)$ for any logic program of treewidth~$k$ with~$n$ variables (atoms)
 and under ETH this dependency on the treewidth can not be significantly improved.
 The reason for the hardness lies in very large cycles (\emph{strongly connected components, SCCs}) of the program's dependency graph; the hardness proof requires cycle lengths that are \emph{unbounded in the treewidth}, i.e., cycles involve instance-size many atoms. 
Interestingly, this is in stark contrast to (ii) \SAT, which can be decided in time~$2^{{\mathcal{O}(k)}} \cdot \poly(n)$.
The classical reduction of \SAT to \ASP, while preserving treewidth, does not introduce any cycles in the encoding (see Example~\ref{ex:encoding}). Thus, the question arises if one can construct cyclic programs to reduce \SAT while decreasing treewidth.
In more details, this paper asks:
%
%
%
%
\begin{itemize}
	\item How can we encode \SAT in (normal) \ASP, thereby decreasing the treewidth by the amount that reflects the runtime difference between (i) and (ii)?
	\item Given the gap between unbounded cycle lenghts in (i) and no cycles in (ii), what is the  difference in cycle length (SCC size) of the complexity between normal \ASP and \SAT? 
	Can we bound the cycle length in the treewidth that still makes
	normal \ASP hard? 
	\item Can we draw further runtime consequences and lower bounds for other fragments or related extensions of \ASP?
\end{itemize}

\smallskip\noindent\textbf{Contributions.} We address these questions via a novel reduction that closes the gap to existing complexity results and lower bounds. Concretely, we provide the following results.
\begin{itemize}
%
	\item First, we establish a novel reduction from \SAT to normal \ASP that in contrast to existing transformations 
fully utilizes the power of reachability along cycles, thereby \emph{decreasing treewidth} from~$k$ to~$\mathcal{O}(\frac{k}{\log(k)})$.
Unless ETH fails, it is not expected that this reduction can be significantly improved,
i.e., further major treewidth decreases are unexpected. 
To the best of our knowledge, this is the first reduction fully utilizing the structural power of normal \ASP.
Then, we also study the largest cycles (SCC sizes) of the dependency graph of the constructed program.
\item 
Interestingly, the constructed cycles (SCC sizes) of the dependency graph are of size at most~$2^{\mathcal{O}(k\cdot\log(k))}$.
This is a major improvement compared to the largest SCC sizes of the recent hardness result, which is unbounded in the treewidth. 
Then, we show that for the class of $\iota$-tight programs,  the SCC sizes can be even decreased to~$2^{\mathcal{O}(k\cdot\log(\iota))}$, while still preserving hardness for treewidth.
%
	\item Finally, we show that our reduction has immediate further implications in terms of computational complexity. 
%
 We establish 
  for the class of $\iota$-tight programs a corresponding \emph{ETH-tight lower bound}. Further, 
%
counting answer sets of a normal program with respect to a projection of interest is expected to be slightly harder
than counting answer sets of disjunctive programs. Notably, both problems are complete
for the same (classical) complexity class, but are surprisingly of different hardness for treewidth.
%
\end{itemize}

\smallskip\noindent\textbf{Related Work.} 
Programs of bounded even or odd cycles have been analyzed~\cite{LinZhao04}. 
%
%
%
%
Further, the 
feedback width has been studied, which depends on the atoms required to break large SCCs~\cite{GottlobScarcelloSideri02}.
There have been improvements for so-called~$\iota$-tight programs~\cite{FandinnoHecher21} with~$\iota$ being smaller than treewidth~$k$, 
which allow for runtimes of~$2^{{\mathcal{O}(k\cdot\log(\iota))}} \cdot \poly(n)$.
For normal and HCF programs, slightly superexponential algorithms in the treewidth~\cite{FichteHecher19} for solving consistency are known.
For disjunctive ASP
algorithms have been proposed~\cite{JaklPichlerWoltran09,PichlerEtAl14} 
running in time linear in the instance size, but double exponential
in the treewidth. 
Hardness of further problems has been studied by means of runtime dependency in the treewidth, e.g., levels of 
exponentiality, where triple-exponential algorithms are known~\cite{
MarxMitsou16,FichteHecherPfandler20}.
%
%
%
%
%
%
\futuresketch{The proposed algorithm 
was used for 
counting answer sets involving projection~\cite{GebserKaufmannSchaub09a}, 
which is at least double exponential~\cite{FichteEtAl18} in the treewidth. 
However, for plain counting (single exponential), it can overcount due to lacking unique level mappings (orderings). 
}

Numerous reductions from \ASP to \SAT are known~\cite{Clark77,Ben-EliyahuDechter94,LinZhao03,Janhunen06,AlvianoDodaro16,BomansonJanhunen13,Bomanson17}.
%
%
%
%
These reductions focus on
the resulting formula size and number of auxiliary variables,
where a sub-quadratic blow-up is unavoidable~\cite{LifschitzRazborov06}.
Unless ETH fails, a sub-quadratic blow-up in the treewidth cannot be circumvented as well ~\cite{Hecher22}.
%
%
%
%
%
For \SAT, empirical results~\cite{AtseriasFichteThurley11}
involving resolution-width 
and treewidth yield efficient \SAT solver runs on instances of small treewidth.

\section{Preliminaries}

We assume familiarity with graph terminology, cf.,~\cite{Diestel12}.
Let~$G=(V,E)$ be a directed graph. Then, a set~$C\subseteq V$ of vertices of~$G$
is a \emph{strongly-connected component (SCC)} of~$G$ if
$C$ is a~$\subseteq$-largest set such that 
for every two distinct vertices~$u,v$ in~$C$ there is a directed path from~$u$ to~$v$ in~$G$.
%

\futuresketch{
\paragraph{Basics and Combinatorics.}
For a set~$X$, let $\ta{X}$ be the \emph{power set of~$X$}
consisting of all subsets~$Y$ with $\emptyset \subseteq Y \subseteq X$.
Let $\vec s$ be a sequence of elements of~$X$. When we address the
$i$-th element of the sequence~$\vec s$ for a given positive
integer~$i$, we simply write $\vec s_{(i)}$. The sequence~$\vec s$
\emph{induces} an \emph{ordering~$<_{\vec s}$} on the elements in~$X$
by defining the
relation~$<_{\vec s} \eqdef \SB (\vec s_{(i)},\vec s_{(j)}) \SM 1 \leq
i < j \leq \Card{\vec s}\SE$.
Given some integer~$n$ and a family of finite subsets~$X_1$, $X_2$,
$\ldots$, $X_n$. Then, the generalized combinatorial
inclusion-exclusion principle~\cite{GrahamGrotschelLovasz95a} states
that the number of elements in the union over all subsets is
$\Card{\bigcup^n_{j = 1} X_j} = \sum_{I \subseteq \{1, \ldots, n\}, I
  \neq \emptyset} (-1)^{\Card{I}-1} \Card{\bigcap_{i \in I} X_i}$.

\paragraph{Computational Complexity.}
We assume familiarity with standard notions in computational
complexity~\cite{Papadimitriou94}
and use counting complexity classes as defined
by~\citex{DurandHermannKolaitis05}.
%
%
%
For parameterized complexity, we refer to standard
texts~\cite{CyganEtAl15}. 
%
%
We recall some basic notions.
Let $\Sigma$ and $\Sigma'$ be some finite alphabets.  We call
$I \in \Sigma^*$ an \emph{instance} and $\CCard{I}$ denotes the size
of~$I$.  
%
Let $L \subseteq \Sigma^* \times \Nat$ and
$L' \subseteq {\Sigma'}^*\times \Nat$ be two parameterized problems. An
\emph{fpt-reduction} $r$ from $L$ to $L'$ is a many-to-one reduction
from $\Sigma^*\times \Nat$ to ${\Sigma'}^*\times \Nat$ such that for all
$I \in \Sigma^*$ we have $(I,k) \in L$ if and only if
$r(I,k)=(I',k')\in L'$ such that $k' \leq g(k)$ for a fixed computable
function $g: \Nat \rightarrow \Nat$, and there is a computable function
$f$ and a constant $c$ such that $r$ is computable in time
$O(f(k)\CCard{I}^c)$. If additionally~$g$ is a linear function,
then~$r$ is referred to as~\emph{fpl-reduction}.
%
%
%
A \emph{witness function} is a
function~$\mathcal{W}\colon \Sigma^* \rightarrow 2^{{\Sigma'}^*}$ that
maps an instance~$I \in \Sigma^*$ to a finite subset
of~${\Sigma'}^*$. We call the set~$\WWW(I)$ the \emph{witnesses}. A
\emph{parameterized counting
  problem}~$L: \Sigma^* \times \Nat \rightarrow \Nat_0$ is a
function that maps a given instance~$I \in \Sigma^*$ and an
integer~$k \in \Nat$ to the cardinality of its
witnesses~$\Card{\WWW(I)}$.
Let $\mtext{C}$ be a decision complexity class,~e.g., \P. Then,
$\cntc\mtext{C}$ denotes the class of all counting problems whose
witness function~$\WWW$ satisfies (i)~there is a
function~$f: \Nat_0 \rightarrow \Nat_0$ such that for every
instance~$I \in \Sigma^*$ and every $W \in \WWW(I)$ we have
$\Card{W} \leq f(\CCard{I})$ and $f$ is computable in
time~$\bigO{\CCard{I}^c}$ for some constant~$c$ and (ii)~for every
instance~$I \in \Sigma^*$ the decision problem~$\WWW(I)$ belongs to
the complexity class~$\mtext{C}$.
Then, $\cntc\P$ is the complexity class consisting of all counting
problems associated with decision problems in \NP.
Let $L$ and $L'$ be counting problems with witness functions~$\WWW$
and $\WWW'$. A \emph{parsimonious reduction} from~$L$ to $L'$ is a
polynomial-time reduction~$r: \Sigma^* \rightarrow \Sigma'^*$ such
that for all~$I \in \Sigma^*$, we
have~$\Card{\WWW(I)}=\Card{\WWW'(r(I))}$. It is easy to see that the
counting complexity classes~$\cntc\mtext{C}$ defined above are closed
under parsimonious reductions. It is clear for counting problems~$L$
and $L'$ that if $L \in \cntc\mtext{C}$ and there is a parsimonious
reduction from~$L'$ to $L$, then $L' \in \cntc\mtext{C}$.
%
%

}

\smallskip
\noindent\textbf{Tree Decompositions (TDs).} %
A \emph{tree decomposition (TD)}~\cite{RobertsonSeymour86} 
of a given graph~$G{=}(V,E)$ is a pair
$\TTT{=}(T,\chi)$ where $T$ is a tree rooted at~$\rootOf(T)$ and $\chi$ 
assigns to each node $t$ of~$T$ a set~$\chi(t)\subseteq V$,
called \emph{bag}, such that (i) $V=\bigcup_{t\text{ of }T}\chi(t)$, (ii)
$E\subseteq\SB \{u,v\} \SM t\text{ in } T, \{u,v\}\subseteq \chi(t)\SE$,
and (iii) for each $r, s, t\text{ of } T$, such that $s$ lies on the path
from~$r$ to $t$, we have $\chi(r) \cap \chi(t) \subseteq \chi(s)$.
For every node~$t$ of~$T$, we denote by $\children(t)$ the \emph{set of child nodes of~$t$} in~$T$.
%
We
let $\width(\TTT) {\eqdef} \max_{t\text{ of } T}\Card{\chi(t)}-1$.
The
\emph{treewidth} $\tw{G}$ of $G$ is the minimum $\width({\TTT})$ over
all TDs $\TTT$ of $G$. \futuresketch{TDs can be 5-approximated in \emph{single exponential time}~\cite{BodlaenderEtAl13} in the treewidth.} 
For a node~$t \text{ of } T$, we say that $\type(t)$ is $\leaf$ if 
$t$ has
no children
; $\join$ if $t$ has exactly two children~$t'$ and $t''$ with
$t'\neq t''$; 
$\inner$ if~$t$ has a single child.
%
If for
every node $t\text{ of } T$, %
$\type(t) \in \{ \leaf, \join, \inner\}$, 
the TD is called \emph{nice}.
%
A TD can be turned into a nice TD~\cite{Kloks94a}[Lem.\ 13.1.3] \emph{without width-increase} in linear~time.
Without loss of generality, we assume that \emph{bags of nice TDs are distinct}. 

%

\smallskip\noindent\textbf{Answer Set Programming (ASP).}
%
We assume familiarity with propositional satisfiability (\SAT)~\cite{BiereHeuleMaarenWalsh09,KleineBuningLettman99},
where we use clauses, formulas, and assignments in the usual meaning.
Two assignments~$I: X \rightarrow \{0,1\}$, $I': X'\rightarrow \{0,1\}$ are 
 \emph{compatible},
whenever for every~$x\in X\cap X'$ we have that~$I(x)=I'(x)$.
 
We follow standard definitions of propositional ASP~\cite{BrewkaEiterTruszczynski11,JanhunenNiemela16a}.
%
Let $\ell$, $m$, $n$ be non-negative integers such that
$\ell \leq m \leq n$, $a_1$, $\ldots$, $a_n$ be distinct propositional
atoms. Moreover, we refer by \emph{literal} to an atom or the negation
thereof.
%
A \emph{program}~$\prog$ is a set of \emph{rules} of the form
%
\(
a_1\por \cdots \por a_\ell \hsep a_{\ell+1}, \ldots, a_{m}, \neg
a_{m+1}, \ldots, \neg a_n.
\)
%
%
%
%
%
%
%
%
%
%
For a rule~$r$, we let $H_r \eqdef \{a_1, \ldots, a_\ell\}$,
$B^+_r \eqdef \{a_{\ell+1}, \ldots, a_{m}\}$, and
$B^-_r \eqdef \{a_{m+1}, \ldots, a_n\}$.
%
%
%
We denote the sets of \emph{atoms} occurring in a rule~$r$ or in a
program~$\prog$ by $\at(r) \eqdef H_r \cup B^+_r \cup B^-_r$ and
$\at(\prog)\eqdef \bigcup_{r\in\prog} \at(r)$.
%
%
%
Program~$\prog$ is \emph{normal} if $\Card{H_r} \leq 1$ for
every~$r \in \prog$.
The \emph{dependency graph}~$D_\prog$ of $\prog$ is the
directed graph defined on the atoms
from~$\bigcup_{r\in \prog}H_r \cup B^+_r$, where for every
rule~$r \in \prog$ two atoms $a\in B^+_r$ and~$b\in H_r$ are joined by
an edge~$(a,b)$.
\futuresketch{A head-cycle of~$D_\prog$ is an $\{a, b\}$-cycle\footnote{Let
  $G=(V,E)$ be a digraph and $W \subseteq V$. Then, a cycle in~$G$ is
  a $W$-cycle if it contains all vertices from~$W$.} for two distinct
atoms~$a$, $b \in H_r$ for some rule $r \in \prog$. 
Program~$\prog$ is
\emph{head-cycle-free} if $D_\prog$ contains no
head-cycle~\cite{Ben-EliyahuDechter94}.}
%
%

An \emph{interpretation} $I$ is a set of atoms. $I$ \emph{satisfies} a
rule~$r$ if $(H_r\,\cup\, B^-_r) \,\cap\, I \neq \emptyset$ or
$B^+_r \setminus I \neq \emptyset$.  $I$ is a \emph{model} of $\prog$
if it satisfies all rules of~$\prog$, in symbols $I \models \prog$. 
For brevity, we view propositional formulas 
as sets of formulas (e.g., clauses) that need to be satisfied, and
use the notion of interpretations, models, and satisfiability analogously. 
%
The \emph{Gelfond-Lifschitz
  (GL) reduct} of~$\prog$ under~$I$ is the program~$\prog^I$ obtained
from $\prog$ by first removing all rules~$r$ with
$B^-_r\cap I\neq \emptyset$ and then removing all~$\neg z$ where
$z \in B^-_r$ from the remaining
rules~$r$~\cite{GelfondLifschitz91}. %
$I$ is an \emph{answer set} of a program~$\prog$, denoted~$I\models \prog$, if $I$ is a minimal
model of~$\prog^I$. %
%
The problem of deciding whether an \ASP program has an answer set is called
\emph{consistency}, which is \SIGMA{2}{P}-complete~\cite{EiterGottlob95}. 
If the input is restricted to normal programs, the complexity drops to
\NP-complete~
\cite{MarekTruszczynski91}.
\futuresketch{A head-cycle-free program~$\prog$ 
can be translated into a normal program in polynomial
time~\cite{Ben-EliyahuDechter94}.}
%
%
%
%
%
%
%
%
%
%
%
%
%
%
%

The following characterization of answer sets is often
invoked for  normal programs~\cite{LinZhao03}.
Let~$A\subseteq\at(\Pi)$ be a set of atoms. Then, a function~$\varphi: A \rightarrow \{0,\ldots,\Card{A}-1\}$ is an \emph{ordering} over~$\dom(\varphi)\eqdef A$.
Let~$I$ be a model of a normal program~$\prog$ and~$\varphi$ be an ordering over~$I$. An atom~$a\in I$ is \emph{proven}
if there is a rule~$r\in\prog$ \emph{proving~$a$}, where $a\in H_r$ with (i)~$B^+_r\subseteq I$,
(ii)~$I \cap B^-_r = \emptyset$ and
$I \cap (H_r \setminus \{a\}) = \emptyset$,
and (iii)~$\varphi(b) < \varphi(a)$ for every~$b\in B_r^+$. Then, $I$ is an
\emph{answer set} of~$\prog$ if (i)~$I$ is a model of~$\prog$, and
(ii) \emph{$I$ is proven}, i.e., every~$a \in I$ is proven.
For an ordering~$\varphi$ and two atoms~$a,b\in\at(\Pi)$, we
write~$a \prec_\varphi b$ whenever \emph{$b$ directly succeeds~$a$},
i.e.,~$\varphi(b)=\varphi(a)+1$.
The empty ordering~$\varphi$ with~$\dom(\varphi)=\emptyset$ is abbreviated by~$\varnothing$.
\futuresketch{This characterization vacuously extends to head-cycle-free
programs by results of~\citex{Ben-EliyahuDechter94}.}
%
%
%
%
%
\futuresketch{Given a program~$\prog$. It is folklore that an atom~$a$ of any answer
set of $\prog$ has to occur in some head of a rule
of~$\prog$~\cite[Ch~2]{GebserKaminskiKaufmannSchaub12}, which we 
hence assume in the following. 
}
%
%
%
%
%
%
%
\longversion{%
} 
%

%
%
%
%
%
%
%
%
%
%
%
%
%
%
%
%
%
%
\longversion{%
}%
\longversion{%
}%

\smallskip
\noindent\textbf{Primal Graph.} We need graph representations to use treewidth for ASP~\cite{JaklPichlerWoltran09}.
The \emph{primal graph 
}$\mathcal{G}_\prog$
of program~$\prog$ has the atoms of~$\prog$ as vertices and an
edge~$\{a,b\}$ if there exists a rule~$r \in \prog$ and $a,b \in \at(r)$.
The primal graph~$\mathcal{G}_F$ of a Boolean Formula~$F$ (in CNF) %
uses variables of~$F$ as vertices and adjoins two vertices~$a,b$ by an edge, if there is a clause in~$F$ %
containing~$a,b$.
Let~$\mathcal{T}=(T,\chi)$ be a TD of~$\mathcal{G}_F$. Then, 
for every node~$t$ of~$T$, 
we define the \emph{bag clauses}~$F_t\eqdef \{c\in F\mid \var(c)\subseteq\chi(t)\}$.

\begin{example}\label{ex:run}
Consider formula~$F\eqdef \{c_1, c_2, c_3\}$,
where $c_1\eqdef (a \vee \neg b)$, $c_2\eqdef (\neg a \vee c \vee d)$, $c_3\eqdef (\neg c \vee \neg d)$.
Figure~\ref{fig:graph-td} (left) depicts the primal graph~$\mathcal{G}_F$ and Figure~\ref{fig:graph-td} (right) shows a TD of~$\mathcal{G}_F$.
Then, observe that~$F_{t_1}=\{c_1\}$, $F_{t_2}=\{c_2,c_3\}$, and~$F_{t_3}=\emptyset$.
%
\end{example}

\smallskip\noindent\textbf{$\iota$-Tightness.}
For a program~$\Pi$ and an atom~$a\in\at(\Pi)$ we denote the \emph{SCC of atom~$a$} in~$D_\Pi$  by~$\scc(a)$.
Then, given a TD~$\mathcal{T}=(T,\chi)$ of~$\mathcal{G}_\Pi$, the
\emph{tightness width} is $\max_{t\text{ of }T}\max_{x\in\chi(t)}\Card{\chi(t) \cap \scc(x)}$.
The \emph{tightness treewidth~$\iota$} of~$\Pi$ is the smallest tightness width among every TD of width in~$\mathcal{O}(\tw{\mathcal{G}_\Pi})$;
 in this case we say~$\Pi$ is~\emph{$\iota$-tight}.

\begin{proposition}[\cite{FandinnoHecher21}]\label{thm:almosttightness}\label{prop:iota}
Assume a normal, $\iota$-tight program~$\Pi$; 
the treewidth of~$\mathcal{G}_\Pi$ is~$k$.
Then, consistency of~$\Pi$ can be decided
in time~$2^{\mathcal{O}(k\cdot \log(\iota))} \cdot \poly(\Card{\at(\Pi)})$. 
\end{proposition}

\begin{figure}[t]%
  \centering
  \shortversion{ %
    \begin{tikzpicture}[node distance=7mm,every node/.style={fill,circle,inner sep=2pt}]
\node (a) [label={[text height=1.5ex,yshift=0.0cm,xshift=0.05cm]left:$e$}] {};
\node (b) [right of=a,label={[text height=.85ex,xshift=0.25em]left:$a$}] {};
\node (e) [below of=a, label={[text height=1.5ex,xshift=-.34em,yshift=.42em]right:$d$}] {};
\node (d) [below of=b,label={[text height=1.5ex,yshift=0.09cm,xshift=-0.07cm]right:$b$}] {};
\node (c) [left of=e,label={[text height=1.5ex,yshift=0.09cm,xshift=0.05cm]left:$c$}] {};
\draw (a) to (c);
\draw (b) to (d);
\draw (c) to (e);
\draw (d) to (e);
\draw (e) to (a);
\draw (d) to (a);
\end{tikzpicture}%
    \includegraphics{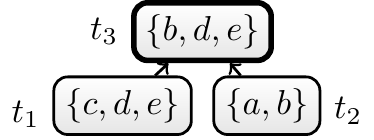}
    \vspace{-.4em}
    \caption{Graph~$G$ (left) and a TD~$\mathcal{T}$ of~$G$ (right).}
  }%
  \longversion{%
    \begin{subfigure}[c]{0.47\textwidth}
      \centering%
      \begin{tikzpicture}[node distance=7mm,every node/.style={fill,circle,inner sep=2pt}]
\node (a) [label={[text height=1.5ex,yshift=0.0cm,xshift=0.05cm]left:$e$}] {};
\node (b) [right of=a,label={[text height=.85ex,xshift=0.25em]left:$a$}] {};
\node (e) [below of=a, label={[text height=1.5ex,xshift=-.34em,yshift=.42em]right:$d$}] {};
\node (d) [below of=b,label={[text height=1.5ex,yshift=0.09cm,xshift=-0.07cm]right:$b$}] {};
\node (c) [left of=e,label={[text height=1.5ex,yshift=0.09cm,xshift=0.05cm]left:$c$}] {};
\draw (a) to (c);
\draw (b) to (d);
\draw (c) to (e);
\draw (d) to (e);
\draw (e) to (a);
\draw (d) to (a);
\end{tikzpicture}%
      \input{graph0/td}%
      \caption{Graph~$G_1$ and a tree decomposition of~$G_1$.}
      \label{fig:graph-td}
    \end{subfigure}
    \begin{subfigure}[c]{0.5\textwidth}
      \centering \begin{tikzpicture}[node distance=7mm,every node/.style={fill,circle,inner sep=2pt}]
\node (a) [label={[text height=1.5ex,yshift=0.0cm,xshift=0.12cm]left:$d$}] {};
\node (b) [right = 0.5cm of a,label={[text height=1.5ex,xshift=0.12cm]left:$a$}] {};
\node (c) [right = 0.5cm of b,label={[text height=1.5ex,xshift=0.05cm]left:$b$}] {};
\node (d) [right = 0.5cm of c,label={[text height=1.5ex,xshift=-0.1cm]right:$c$}] {};
\node (e) [right = 0.5cm of d,label={[text height=1.5ex,xshift=-0.1cm]right:$e$}] {};
\node (r3) [below = 0.5cm of a,label={[text height=1.5ex,xshift=-0.1cm]right:${r_3}$}] {};
\node (r1) [below = 0.5cm of b,label={[text height=1.5ex,xshift=-0.05cm]right:${r_1}$}] {};
\node (r2) [below = 0.5cm of c,label={[text height=1.5ex,xshift=-0.12cm]right:${r_2}$}] {};
\node (r4) [below = 0.5cm of d,label={[text height=1.5ex,xshift=-0.05cm]right:${r_4}$}] {};
\draw (a) to (r3);
\draw (c) to (r3);
\draw (b) to (r1);
\draw (c) to (r1);
\draw (b) to (r2);
\draw (c) to (r2);
\draw (d) to (r2);
\draw (d) to (r4);
\draw (e) to (r4);
\end{tikzpicture}%
      \begin{tikzpicture}[node distance=0.5mm]
\tikzset{every path/.style=thick}

\node (leaf0) [tdnode,label={[]left:$t_1$}] {$\{a,b, {r_1}, {r_2}\}$};
\node (leaf1) [tdnode,label={[xshift=0em]left:$t_2$}, above = 0.2cm of leaf0] {$\{b, d,{r_2}, {r_3}\}$};
\node (leaf2) [tdnode,label={[xshift=0em]right:$t_3$}, right = 0.2cm of leaf0]  {$\{c, e, {r_4}\}$};
\coordinate (middle) at ($ (leaf1.north east)!.5!(leaf2.north west) $);
\node (join) [tdnode,ultra thick,label={[xshift=-0.25em]right:$t_4$}, right = 0.25cm of leaf1] {$\{c, r_2\}$}; 

\draw [->] (leaf0) to (leaf1);
\draw [<-] (join) to (leaf1);
\draw [<-] (join) to (leaf2);
\end{tikzpicture}%
      \caption{Graph~$G_2$ and a tree decomposition of~$G_2$.}
      \label{fig:graph-td2}%
    \end{subfqigure}
    \caption{Graphs~$G_1, G_2$ and two corresponding tree
      decompositions.}
  }%
  \label{fig:graph-td}%
\end{figure}

\futuresketch{
\begin{example}
\label{ex:running1}\label{ex:running}
Consider the following program
\vspace{-0.1em}
$\prog\eqdef$
%
$\SB\overbrace{ a \lor b \hsep}^{r_1};\, %
\overbrace{c \lor e \hsep d}^{r_2};\, %
\overbrace{d \lor e \hsep b}^{r_3};\, %
\overbrace{b \hsep e, \neg d}^{r_4};\, %
\overbrace{d \hsep \neg b}^{r_5} %
\SE$.
%
%
Observe that $\prog$ is head-cycle-free.
Then, $I\eqdef\{b, c, d\}$ is an answer set of~$\prog$,
since~$I\models\Pi$, and we can prove with ordering
$\varphi \eqdef\{b\mapsto 0, d\mapsto 1, c\mapsto 2\}$
atom~$b$ by rule~$r_1$, 
atom~$d$ by rule~$r_3$, and
atom~$c$ by rule~$r_2$.
Further answer sets are $\{b,e\}$,
$\{a,c,d\}$, and~$\{a,d,e\}$.


\end{example}}%
%


\futuresketch{
The following result for QBFs is known, where
it turns out that deciding $\QBFSAT$ remains $\ell$-fold exponential in the treewidth
of the primal graph (even when restricting the graph 
to the variables of the inner-most quantifier block).

\begin{proposition}[\cite{FichteHecherPfandler20}]\label{prop:lb}
Given a QBF~$Q=\exists V_1. \forall V_2, \ldots, \forall V_\ell. F$ of quantifier depth~$\ell$.
Then, unless ETH fails, the validity of~$Q$ cannot be decided
in time~$\tower(\ell, o(k))\cdot\poly(\var(Q))$, where~$k$
is the treewidth of the primal graph (even when restricted to vertices in~$V_\ell$).
\end{proposition}}

\futuresketch{
\subsection{Reduction of SAT to normal ASP}

Note that~$w!/w^w=e^{-\mathcal{O}(k)}$ via Sterling's formula~\cite{LokshtanovMarxSaurabh11}.

Given a SAT formula~$F$ and a tree decomposition~$\mathcal{T}=(T,\chi)$ of the primal graph of~$F$ such that~$T=(N,A,n)$.
We reduce~$F$ to a normal program~$\Pi$ such that
$F$ is satisifiable if and only if~$\Pi$ has an answer set.
The reduction itself is guided by~$\mathcal{T}$ and the width of~$\mathcal{T}'$ is 
$\bigO{w/\log(w)}=\bigO{w/\log(w\cdot\log(w))}=\bigO{w/[\log(w)+\log(\log(w))]}$, 
where~$w$ is the width of~$\mathcal{T}$.
In total, with ASP, we can assign up to~$2^{\bigO{s\cdot\log(s)}}$ many states
for each bag~$\chi(t')$ of~$\mathcal{T}'$ of cardinality at most~$s$.
Concretely, we have exactly~$\Sigma_{I \subseteq \chi(t')} \Card{I}!$ many states.
We will use these states to simulate SAT, but thereby ``save'' in the treewidth.

Given any two bags~$\chi(t_1),\chi(t_2)$ of~$\mathcal{T}$ such that~$t_1$ is a children of~$t_2$, we construct the following ASP rules in~$\Pi$.
We assume an arbitrary total ordering~$\prec_{t_1}$ of the assignments~$I\in 2^{\chi(t_1)}$
in bag~$t_1$.
Further, we assume an arbitrary total ordering~$\prec_{t'_1}$ among $2^{\chi(t'_1)} \times ord(\chi(t'_1))$.
We use a mapping~$m: 2^{\chi(t_1)} \rightarrow [2^{\chi(t'_1)} \times ord(\chi(t'_1))]$,
which is a bijection between an assignment~$I$ in~$\chi(t_1)$ and an assignment~$I'$ over~$\chi(t'_1)$ together with
an ordering among those atoms in~$I'$.
The mapping~$m$ is naturally defined by assigning the elements in~$\prec_{t_1}$ to~$\prec_{t'_1}$ according to
the ordinal of the element.
Then, we add for each element~$I\in2^{\chi(t_1)}$ the rules~$0 { a_I } 1.$, but ensure an ``at most one'' behavior.
This can be done later on top of the TD. Further, we add~$v_j \leftarrow v_i, a_I.$
For each element~$I\in2^{\chi(t_1)}$ we add a parent~$t'_{1.i}$ to~$t'_1$.

\subsection{Reduction of normal ASP to PASP using trick of exponential compression as used in the hierarchy}

TBD: Just reduce normal ASP to InvPASP problem, using the same ideas as in the hierarchy, exponential compression and such, then easily reduce InvPASP to PASP.
We have now a lower bound for PASP using normal ASP.
}

\section{Decreasing Treewidth of \SAT  via \ASP}\label{sec:main}
In this section
we show how to translate a Boolean formula into a logic program,
thereby decreasing the treewidth and explicitly utilizing the structural power of \ASP.
Thereby, 
in contrast to the 
 standard translation as sketched in Example~\ref{ex:encoding}, we explicitly utilize cycles and the power of reachability the \ASP formalism provides. 


\futuresketch{
We proceed in two steps. First, we translate the Boolean formula into a formula,
where it is guaranteed that there exists a nice tree decomposition of the formula's primal graph
such that every variable of the formula appears only in constantly many bags.

\subsection{Elimination of Decomposition-Global Variables}
In the first step we show how one can eliminate decomposition-global variables
of a Boolean formula~$F$, i.e., we reduce to a new formula~$F'$,
where variables are only in at most three decomposition bags.

Let therefore
$\mathcal{T}=(T,\chi)$ be a nice tree decomposition of~$\mathcal{G}_{F}$. 
The idea is as follows: Every variable~$v$ in a bag~$\chi(t)$ for a node~$t$ of~$T$
is bound to~$t$, i.e., we require a new variable~$v_t$.
So the new formula~$F'$ has as variables the fresh set~$\{v_t \mid t\text{ in }T, v\in \chi(t)\}$
of variables. Then, for every node~$t$ of~$T$ with~$t'\in \children(t)$ and~$x\in \chi(t)\cap\chi(t')$,
we ensure equivalence: \inlineequation[botup:prop]{x_t \longleftrightarrow x_{t'}}. 
Finally, for every node~$t$ of~$T$ and clause~$(l_1 \vee \ldots \vee l_o)\in \{c \mid c\in F, \var(c)\subseteq \chi(t) \}$, we take care that the clause is satisfied in~$t$: \inlineequation[botup:clause]{l_1^t \vee \ldots \vee l_o^t}, 
where function~$\cdot^t$ takes a literal over some variable~$x$ and replaces the occurrence of~$x$ by~$x_t$.

Observe that indeed the resulting formula~$F'$ is satisfiable if and only if~$F$ is satisfiable and there is even a bijective correspondence between models of~$F$ and models of~$F'$. Moreover, we can easily construct a tree decomposition~$\mathcal{T}'\eqdef (T,\chi')$ of~$\mathcal{G}_{F'}$ such that every variable of~$F'$ occurs in at most three bags of~$\mathcal{T}'$. Thereby, the tree~$T$ of~$\mathcal{T}'$  is used as before, we only need to define~$\chi'$, which we give for every node~$t$ of~$T$ as follows: $\chi'(t)\eqdef \{x_t \mid x\in\chi(t)\} \cup \{x_{t^*} \mid \chi(t)\neq\emptyset,x\in \chi(t^*), t^* \text{ is the parent of }t\text{ in }T\}$.
The construction of~$\mathcal{T}'$ ensures that~$\mathcal{T}'$ is a TD of~$\mathcal{G}_{F'}$. Further, the width of~$\mathcal{T}'$ is at most doubled and every variable of~$F'$ appears in at most three different bags of~$\mathcal{T}'$. 

\smallskip
In the following, we \emph{only} assume such a formula~$F'$ and a nice TD~$\mathcal{T}''$ of~$\mathcal{G}_{F'}$, which is obtained from~$\mathcal{T}'$ above, by adding an auxiliary node~$t'$ for every node~$t$ that has two child nodes and making~$t$ the parent of~$t'$, i.e., $t'$ serves as a fresh join node that is placed between~$t$ and their child nodes.  
As a result, every variable of~$F'$ then appears in at most four bags; we refer to~$F'$ as \emph{local formula}
and~$\mathcal{T}''$ as \emph{join-local tree decomposition (of~$\mathcal{G}_F$)}.
This first step is a preprocessing step, providing a normal form of formulas and TDs, making it easier to present and grasp the second part of our~reduction. 
}


\smallskip
The concrete decrease of treewidth of our reduction of the next subsection will be tightly linked to the following observation,
which expresses that the factorial~$k!$ of a number~$k\in \Nat$ is bounded from below by~$k^{{\Omega}(k)}$.

\begin{observation}\label{obs:klogk}
Let~$k\in \Nat$. Then, $k!$ is in~$2^{\Omega(k\cdot\log(k))}$.
\end{observation}
\begin{proof}
We have~$\frac{k!}{k^k}=e^{-\mathcal{O}(k)}$ by using Stirling's formula, see, e.g.,~\cite{LokshtanovMarxSaurabh11}.
As a result, we derive that~$k!$ corresponds to~$\frac{2^{k\cdot\log(k)}}{e^{\mathcal{O}(k)}}$ = $\frac{2^{k\cdot\log(k)}}{2^{\log(e)\cdot\mathcal{O}(k)}}$= ${2^{k\cdot\log(k)-\log(e)\cdot\mathcal{O}(k)}}$= ${2^{\Omega(k\cdot\log(k))}}$.
\end{proof}

This observation immediately implies that~$k!$ is of the same order of magnitude as~$k^{{\Theta}(k)}$,
as obviously~$k!$ is in~$k^{\mathcal{O}(k)}$.

\subsection{Decreasing Treewidth by the Power of Reachability}

The idea of our main reduction~$\mathcal{R}$ is as follows.
We take an instance~$F$ of \SAT, i.e., a Boolean formula and a nice tree decomposition~$\mathcal{T}=(T,\chi)$ of~$\mathcal{G}_F$ of width~$k$.
Then, we simulate for each node of~$T$, the up to~$2^k$ many assignments via~$k'!$ many orderings, where~$k'$ shall be sufficiently smaller than~$k$.
More precisely, we decrease the treewidth from~$k$ to~$k'$ such that~$k'! \geq 2^k$.
Then, since $k'!$ is in~$2^{\Omega(k'\cdot\log(k'))}$ (see Observation~\ref{obs:klogk}), we have that~$2^{\Omega(k'\cdot\log(k'))}$ is at least~$2^k$ and therefore~$k'=\mathcal{O}(\frac{k}{\log(k')})=\mathcal{O}(\frac{k}{\log(k)})$. 
As a result, our approach allows us to 
slightly reduce treewidth, thereby efficiently utilizing
the power of \ASP and positive cycles in order to solve \SAT with less structural overhead.
While this seems surprising, it is in line with the known hardness result of \ASP, cf., Proposition~\ref{thm:almosttightness} for~$\iota=k$.

Formally, we determine~$k'$ by taking the smallest integer~$k'$ such that~$k'! \geq 2^k$.
We define such a value for every node~$t$ of~$T$, where~$k'_t$ is the \emph{smallest integer} such that~$k'_t! \geq 2^{\Card{\chi(t)}}$.
Then, we define a set~$V_t$ of \emph{ordering vertices} consisting of~$k'_t$ many fresh vertices that are uniquely
determined by the bag~$\chi(t)$, i.e., for any TD nodes~$t,t'$ we have~$\chi(t)=\chi(t')$ if and only if~$V_t=V_{t'}$. 
Using this set~$V_t$, we refer to  the resulting \emph{set of at least~$2^{\Card{\chi(t)}}$ 
many orderings} among elements in~$V_t$ by~$\ord(t)$. 
Further, for every node~$t$ of~$T$, we refer to the bijective \emph{mapping from a subset~$X\subseteq\ord(t)$ of orderings to assignments} by~$\mathcal{I}_t: X \rightarrow 2^{\chi(t)}$.
More precisely, for every node~$t$ and ordering~$\varphi\in\ord(t)$, the corresponding \emph{unique assignment of~$\varphi$ over atoms in~$\chi(t)$} is given by~$\mathcal{I}_t(\varphi)$ (if exists).
Note that~$\mathcal{I}_t$ is any arbitrary, but fixed bijection, i.e., it might be undefined for some unused orderings in~$\ord(t)$. 
%
%
%
%
%
%
\begin{example}\label{ex:v}
Recall formula~$F$ and TD~$\mathcal{T}=(T,\chi)$ of Example~\ref{ex:run}.
By definition, we require that~$\Card{V_{t_1}}!\geq 2^2$, 
~$\Card{V_{t_2}}!\geq 2^3$, and~$\Card{V_{t_3}}!\geq 2$.
As a result, we need to choose~$\Card{V_{t_1}}=3$, $\Card{V_{t_2}}=4$, and $\Card{V_{t_3}}=2$.
Consequently, there are orderings~$\alpha\in\ord(t_1)$, $\beta\in\ord(t_2)$, where~$\mathcal{I}_{t_1}(\alpha)$ and $\mathcal{I}_{t_2}(\beta)$ is not defined.
By convention, we refer to the elements in~$V_{t_1}$ by~$v_1^j$, to those in~$V_{t_2}$ by~$v_2^j$ and to those in~$V_{t_3}$ by~$v_3^j$.
\end{example}

\paragraph{Ordering-Augmented Tree Decompositions.}
Let~$\mathcal{T}=(T,\chi)$ be a nice tree decomposition of~$\mathcal{G}_\Pi$.
In order to decouple the~$k'!$ many assignments, we need to get access to any of the simulated orderings one-by-one.
To this end, we define an \emph{ordering-augmented tree decomposition}. 

\begin{definition}
\label{def:ord}
Let~$F$ be a Boolean formula and $\mathcal{T}=(T,\chi)$ be a nice TD of~$\mathcal{G}_F$. 
Then, we construct an ordering-augmented TD~$\mathcal{T}'=(T',\chi', \varphi,\psi)$ of~$\mathcal{G}_F$ from~$\mathcal{T}$ as follows, where~$(T',\chi')$ is a TD and $\varphi,\psi$ are mappings from nodes to orderings. 
%
For every~$t$ of~$T$, let $\chi'(t){\eqdef} \chi(t)$, $\varphi_t{\eqdef} \varnothing$, $\psi_t {\eqdef} \varnothing$.
For every two neighboring nodes~$t$, $t'$ of~$T$ with~$t'\in\children(t)$ and $\ord(t)\times\ord(t')=\{(\alpha_1, \beta_1), \ldots, (\alpha_\ell, \beta_\ell)\}$, 
%
we add a sequence of fresh nodes~$t_1, \ldots, t_\ell$ between~$t$ and~$t'$, such that for every~$1\leq i\leq \ell$ we define~$\chi'(t_i){\eqdef}\chi(t)$, $\varphi_{t_i}{\eqdef}\alpha_i$, and~$\psi_{t_i}{\eqdef}\beta_i$. 
For every leaf node~$t$ of~$T$ with~$\ord(t)=\{\alpha_1, \ldots, \alpha_{\ell}\}$, 
in~$T'$ we chain copy nodes~$t_1,\ldots,t_\ell$ below~$t$, where for every~$1\leq i\leq \ell$, $\chi'(t_i){\eqdef} \chi(t)$, $\varphi_{t_i}{\eqdef}\alpha_i$, $\psi_{t_i}{\eqdef}\varnothing$.
%
%
\end{definition}
%
%

Observe that this definition provides the basis to analyze assignments in the form of different orderings, individually and one-by-one. 
Further, by comparing every pair of orderings of neighboring TD nodes, our reduction will later synchronize neighboring orderings and ensure compatibility. 

\begin{example}\label{ex:ord}
Recall formula~$F$ and TD~$\mathcal{T}=(T,\chi)$ from the previous example.
Then, TD~$\mathcal{T}$ can be turned into an ordering-augmented TD~$\mathcal{T}'=(T',\chi',\varphi,\psi)$ according to Definition~\ref{def:ord}.
Thereby we add a sequence (path) of child nodes to~$t_1$;
each of these nodes~$t_i$ handles one ordering~$\varphi_{t_i}$ over
$V_{t_1}$.
Similarly, this is carried out for node~$t_2$.
Between~$t_3$ and~$t_1$ we also add a path of fresh nodes,
where each node covers a combination (pair) of orderings~$\varphi_{t_3}, \psi_{t_1}$,
which will be essentially used for synchronization.
Analogously, this is done between~$t_3$ and~$t_2$.
\end{example}

\paragraph{Involved Atoms.}
In order to guess among those~$\Card{\ord(t)}$ many orderings per node~$t$ of~$T$, 
we characterize each ordering~$\alpha\in \ord(t)$ by means of \emph{atoms modeling ordering edges} of the form~$e_{x,y}$  (and its negation~$\hat e_{x,y}$) to indicate whether for every two different vertices~$x,y\in V_t$, we have~$x \prec_\alpha y$. 
Further, we require \emph{atoms of the form~$r_x$} for every~$x\in V_t$, which stores whether~$x$ is reached or not. For~$V_t$ itself (which might span over several TD nodes) and~$y\in V_t$ we require 
an additional (source) reachability atom of the form~$r_{s_{V_t}}$ and edge atom~$e_{s_{V_t},y}$.
Then, we  also use (destination) reachability atom~$r_{d_{V_t}}$ and edge atoms~$e_{y,d_{V_t}}$ for every~$y\in V_{t}$.

We will check, whether for every~$x\in V_t$ of every node~$t$ of~$T$, there is at \emph{most one outgoing edge}, i.e., an atom~$e_{x,y}$ contained in an answer set.
To ensure this, we guide information of at most one outgoing edge along the tree decomposition, whereby for every node~$t$ of~$T$ and~$x\in V_t$, we use an \emph{auxiliary atom~$o_t^x$}.
This is also required for the source~$s_{V_t}$, i.e., we also use \emph{auxiliary atom}~$o_t^{s_{V_t}}$ for every node~$t$.
Observe that this includes atoms~$o_{t'}^{s_{V_{t}}}$ for nodes~$t'$, if~$V_t=V_{t'}$.

In order to compare orderings, for every tree decomposition node~$t$ of~$T$
and element~$x$ contained in ordering~$\alpha$ with~$\alpha\in \{\varphi_t, \psi_t\}$,
we use an additional \emph{ordering query atom~$q_t^x$}. If this query atom holds,
the ordering~$\alpha$ is at least fulfilled starting from the first atom contained in~$\alpha$ up to the atom~$x$.
Intuitively, if~$q_t^y$ holds for the last atom~$y$ contained in~$\alpha$, we will ultimately be able to determine whether~$\alpha$ holds.
These query atoms are supported by means of additional \emph{testing points~$p_t^x$} as well as \emph{initial testing points~$p_\epsilon^x$},
which we use for every atom~$x$ contained in every~$\alpha\in\{\varphi_t,\psi_t\}$ for every node~$t$ of~$T$. 
The testing points allow us to check orderings with queries in every node.
For every node~$t$ of~$T$ and atom~$x\in V_t$, we refer by
$\pred(x,t)$ to the \emph{testing points of~$x$ preceding~$t$}. Formally, $\pred(x,t)\eqdef \children(t) \text{ if } x\in V_{t'} \text{ for some }t'\in \children(t)\text{, otherwise }\{\epsilon\}$.
%
%

\paragraph{The Reduction.}
The reduction~$\mathcal{R}$ takes a Boolean formula~$F$ and an ordering-augmented TD~$\mathcal{T}=(T,\chi,\varphi,\psi)$.
Before we commence with the description of~$\mathcal{R}$, we require the following definition.
For every node~$t$ of~$T$, we let the \emph{set~$E_t$ of ordering edges} be any arbitrary fixed subset of
direct successors of~$\varphi_t$.  More precisely, for the first element~$a\in\dom(\varphi_t)$ and the last element~$b\in\dom(\varphi_t)$, i.e., $a$ has no~$\prec_{\varphi_t}$ predecessor and $b$ has no~$\prec_{\varphi_t}$ successor, 
$E_t$ is the largest subset with~$E_t\subseteq \{e_{x,y} \mid x \prec_{\varphi_t} y\} \cup \{e_{s_{V_t},a}, e_{b,d_{V_t}}\}$ 
such that for any node~$t'$ of~$T$ with~$t'\neq t$ we have~$E_t\cap E_{t'}{=}\emptyset$.
%
%
%
%
%

\begin{example}\label{ex:edges}
Consider again formula~$F$ and ordering-augmented TD~$\mathcal{T}'=(T',\chi',\varphi,\psi)$ from above.
Observe that the definition of~$E_t$ is rather open,
but it essentially requires that every ordering edge is encountered in exactly
one node of~$T'$.
For example, the ordering~$\varphi_{t^*}=\{v_1^1 \mapsto 0, v_1^2 \mapsto 1, v_1^2 \mapsto 2\}$
is handled in a node~$t^*$ below~$t_1$.
We could set $E_{t^*}=\{e_{s_{V_{t_1}},v_1^1}, e_{v_1^1,v_1^2}, e_{v_1^2,v_1^3}, e_{v_1^3,d_{V_{t_1}}}\}$.
Assume a different ordering~$\varphi_{t'}=\{v_1^1 \mapsto 0, v_1^3 \mapsto 1, v_1^2 \mapsto 2\}$ over~$V_{t_1}$
in the child node~$t'$ of~$t^*$.
Then, considering~$E_{t^*}$, i.e., avoiding overlapping edges, we could set $E_{t'}=\{e_{v_1^1,v_1^3}, e_{v_1^3,v_1^2}, e_{v_1^2,d_{V_{t_1}}}\}$.
Similarly every remaining edge is covered uniquely in a node.
\end{example}

%
%

Overall, the reduction~$\mathcal{R}$ as given in Figure~\ref{fig:red} consists of five blocks, where the first block of Formulas~(\ref{red:start})--(\ref{red:linkroot}) concerns the \emph{choice of orderings} (via edge atoms~$e_{y,x}$) and enforces that every reachability atom~$r_x$ holds. Then, the second block of Formulas~(\ref{red:fintrans})--(\ref{red:fin}) takes care that for every node~$y\in V_t$ of every node~$t$ of~$T$, there is \emph{at most one outgoing edge} of the form~$e_{y,x}$ in an answer set.
The third block of Formulas~(\ref{red:copy})--(\ref{red:linkbypass2}) ensures that \emph{testing points} are maintained. 
This then allows us to properly \emph{define ordering queries} in the forth block of Formulas~(\ref{red:linknopred})--(\ref{red:linkpq}). 
Finally, the last block of Formulas~(\ref{red:probecons})--(\ref{red:checkunused}) takes care that chosen \emph{orderings are compatible} and that every clause in~$F$ is satisfied.

\begin{figure*}[t]
{
\begin{flalign}
	&\textbf{Block 1: Orderings \& Reachability}\hspace{-10em}\notag\\[-.4em]
	\label{red:start}&r_{s_{V_t}}\leftarrow & \text{for every }t\text{ in }T\\
	\label{red:reached}&\leftarrow \neg r_y & \text{ for every }t\text{ in }T, y\in V_t \cup \{d_{V_t}\}\\
	\label{red:choose}&e_{y,x} \vee \hat e_{y,x} \leftarrow r_y & \text{ for every }t\text{ in }T, 
e_{y,x}\in E_t\\
	\label{red:reach}&{p_{\epsilon}^x} \leftarrow e_{y,x}  & \text{ for every }t\text{ in }T, e_{y,x}\in E_t, \{x,y\}\subseteq V_t\\
	\label{red:reachfin}&{r_{d_{V_t}}} \leftarrow e_{y,d_{V_t}}  & \text{ for every }t\text{ in }T, e_{y,d_{V_t}}\in E_t\\
%
%
	\label{red:linkplast}&r_x \leftarrow p_{t'}^x &\text{for every }t\text{ in }T, \varphi_t=\varnothing, t'\in\children(t), x\in \dom(\psi_{t'})\setminus \dom(\varphi_{t'})\\ 
	\label{red:linkroot}&r_x \leftarrow p_{\rootOf(T)}^x &\text{for every }x\in V_{\rootOf(T)}\\[.05em] 
%
%
	&\textbf{Block 2: ${\leq}1$ Outgoing Edge}\hspace{-10em}\notag\\[-.4em]
	\label{red:fintrans}&o_t^y\leftarrow o_{t'}^y &\text{ for every }t\text{ in }T, t'\in\children(t), y\in (V_t\cap V_{t'}) \cup \{s_{V_t}\mid V_t=V_{t'}\}\\
	\label{red:finnew}&o_t^y\leftarrow e_{y,x} &\text{ for every }t\text{ in }T, e_{y,x}\in E_t, y\in V_t\\
	\label{red:fin}&\leftarrow o_{t'}^y, e_{y,x} &\text{ for every }t\text{ in }T, e_{y,x}\in E_t, t'\in\children(t),y\in V_{t'}\cup \{s_{V_t}\mid V_t=V_{t'}\}\\[.05em]
	&\textbf{Block 3: Testing Points}\hspace{-10em}\notag\\[-.35em]
	\label{red:copy}&p_{t}^x \leftarrow p_{t_1}^x, \ldots, p_{t_o}^x & \text{for every }t\text{ in }T, x\in V_t, \varphi_t=\varnothing, 
	\pred(x,t)=\{t_1,\ldots,t_o\},\\
\label{red:linkbypass}&p_{t}^x \leftarrow p_{t'}^x, \neg q_{t}^y & \text{for every }t\text{ in }T, \alpha\in\{\varphi_t,\psi_t\}, \pred(x,t)=\{t'\}, \{x,y\}\subseteq\dom(\alpha), y\prec_{\alpha} x\\  
	\label{red:linkbypass2}&p_{t}^x \leftarrow p_{t'}^x, \neg q_{t}^x & \text{for every }t\text{ in }T, \alpha\in\{\varphi_t,\psi_t\}, \pred(x,t)=\{t'\}, x\text{ has no}\prec_{\alpha}\text{ successor }\\[.05em]  
	&\textbf{Block 4: Ordering Queries}\hspace{-10em}\notag\\[-.4em]
	\label{red:linknopred}&q_{t}^x \leftarrow p_{t'}^x & \text{ for every }t\text{ in }T, \pred(x,t)=\{t'\},\alpha\in\{\varphi_t,\psi_t\}, x\in\dom(\alpha), x \text{ has no }\prec_{\alpha}\text{ predecessor} \\
	\label{red:linkprevious}&q_{t}^x \leftarrow p_{t'}^x, q_{t}^y & \text{for every }t\text{ in }T, \alpha\in\{\varphi_t,\psi_t\}, \pred(x,t)=\{t'\}, \{x,y\}\subseteq\dom(\alpha), y\prec_{\alpha} x \\
	\label{red:linkpq}&p_{t}^x \leftarrow q_{t}^x & \text{for every }t\text{ in }T, \alpha\in\{\varphi_t,\psi_t\}, x\in\dom(\alpha)\\[.05em]
	&\textbf{Block 5: Compatibility \& \SAT}\hspace{-10em}\notag\\[-.4em]
	\label{red:probecons}&\leftarrow q_{t}^x, q_{t}^y &\text{for every }t\text{ in }T, \alpha\in\{\varphi_t,\psi_t\}, x\in\dom(\alpha) \text{ has no}\prec_{\alpha}\text{successor}, \\
&&\notag y\in\dom(\psi_{t}) \text{ has no}\prec_{\psi_{t}}\text{successor}, \mathcal{I}_t(\varphi_t)\text{ and }\mathcal{I}_t(\psi_{t})\text{ incompatible}\\ 
	\label{red:check}&\leftarrow q_{t}^x &\text{for every }t\text{ in }T, x\in\dom(\varphi_t) \text{ has no }\prec_{\varphi_t}\text{successor}, \mathcal{I}_t(\varphi_t)\not\models F_t\\ 
	%
	\label{red:checkunused}&\leftarrow q_{t}^x &\text{for every }t\text{ in }T, x\in\dom(\varphi_t) \text{ has no}\prec_{\varphi_t}\text{successor}, \mathcal{I}_t(\varphi_t)\text{ not defined}
\end{flalign}
\vspace{-1.75em}
}\caption{The reduction~$\mathcal{R}$ that takes a formula~$F$ and a corresponding ordering-augmented TD~$\mathcal{T}=(T,\chi,\varphi,\psi)$ of~$\mathcal{G}_F$.}\vspace{-.35em}\label{fig:red}
\end{figure*}

\paragraph{Block 1: Choice of Orderings, Formulas~(\ref{red:start})--(\ref{red:linkroot}).} 
The first block concerns choosing orderings. Note that the disjunction of Formulas~(\ref{red:choose})
is head-cycle-free and can be simply converted to normal rules by shifting~\cite{Ben-EliyahuDechter94}.
Then, Formulas~(\ref{red:start}) set reachability of source vertices and Formulas~(\ref{red:reached}) ensure reachability of all the vertices in~$V_t$ as well as the destination vertex for~$V_t$, for every node~$t$ of~$T$. Formulas~(\ref{red:choose}) require to choose outgoing edges (at least one by Formulas~(\ref{red:reached})) from every reachable vertex~$y$ to some vertex~$x$. This then yields initial testing points for~$x$ by Formulas~(\ref{red:reach}).
For the destination vertices of sets~$V_t$, such testing points are not needed, so we immediately obtain reachability by Formulas~(\ref{red:reachfin}).
The connection and propagation between testing points will be achieved by Block 3.
In the end, the last testing point for a vertex~$x$ yields reachability of~$x$. This is ensured by Formulas~(\ref{red:linkplast}), whenever~$x$ does not appear in an ordering for a successor node of~$t'$, or by Formulas~(\ref{red:linkroot}), if~$x$ appears in an ordering of the root node.

\smallskip
\noindent\textbf{Block 2: $\leq 1$ Outgoing Edge, Formulas~(\ref{red:fintrans})--(\ref{red:fin}).}
This block ensures at most one outgoing edge per vertex~$y$,
where~$y$ can be also the source vertex~$s_{V_t}$ for a set~$V_t$ of ordering vertices.
The information of whether~$y$ has decided an outgoing
edge up to a node is propagated from a node~$t'$ to its parent node~$t$ by Formulas~(\ref{red:fintrans}).
Then, whenever in a node~$t$ an outgoing edge for~$y$ is chosen,
$o_t^y$ has to hold by Formulas~(\ref{red:finnew}).
Finally, Formulas~(\ref{red:fin}) prevent choosing outgoing edges for an atom~$y$ in a node~$t$, if already chosen in a child node~$t'$. 

\smallskip
\noindent\textbf{Block 3: Propagate Testing Points, Formulas~(\ref{red:copy})--(\ref{red:linkbypass2}).}
The third block concerns about propagation of testing points,
if certain queries do not hold.
For the case of the empty ordering, i.e., $\varphi_t=\varnothing$, 
in a node~$t$, Formulas~(\ref{red:copy}) directly propagate testing points for every atom~$x\in V_t$ from the evidence of testing points for~$x$ in every child node of~$t$ (or from~$p_\epsilon^x$ if~$\prev(x,t)=\{\epsilon\}$).
Further, whenever a certain ordering relation~$y\prec_\alpha x$
for either~$\alpha=\varphi_t$ or predecessor~$\alpha=\psi_t$
does not hold, we still need to derive the corresponding testing point,
see Formulas~(\ref{red:linkbypass}), as this testing point is required for further queries or for deriving reachability in the end, cf., Formulas~(\ref{red:linkplast}), (\ref{red:linkroot}). This also holds for the very last element of~$\alpha$, see Formulas~(\ref{red:linkbypass2}).
The reason why we need to cover both orderings~$\varphi_t$ as
well as~$\psi_t$, is that neighboring orderings require compatibility, which will be discussed below Block~5.

\smallskip
\noindent\textbf{Block 4: Define Ordering Queries, Formulas~(\ref{red:linknopred})--(\ref{red:linkpq}).}
This block focuses on deriving query atoms, which
ensure that certain orderings hold.
The first element of any ordering~$\alpha\in\{\varphi_t,\psi_t\}$ 
is derived from the previous testing point, as given by Formulas~(\ref{red:linknopred}).
Then, whenever~$y\prec_\alpha x$ is met, Formulas~(\ref{red:linkprevious})  enable to derive
query atom~$q_t^x$, which depends on the previous testing point for~$x$ as well as on~$q_t^y$. This thereby ensures that the order~$y\prec_\alpha x$ is indeed preserved, which is in contrast to Formulas~(\ref{red:linkbypass}) and~(\ref{red:linkbypass2}) above.
Finally, Formulas~(\ref{red:linkpq}) immediately yield the corresponding testing point~$p_t^x$ in case query atom~$q_t^x$ holds.

\smallskip
\noindent\textbf{Block 5: Compatibility of Orderings \& Satisfiability, Formulas~(\ref{red:probecons})--(\ref{red:checkunused}).}
The last block takes care of compatibility and
satisfiability of every clause of the formula~$F$
by excluding orderings, whose corresponding 
assignments do not satisfy some clause.
To this end, Formulas~(\ref{red:probecons}) excludes those cases of incompatible~$\varphi_t$ and~$\psi_t$, i.e., to prevent inconsistencies, it is prohibited that query atoms~$q_t^x, q_t^y$ for the last element~$x$ of~$\varphi_t$ and the last element~$y$ of~$\psi_t$ hold.
Most importantly, Formulas~(\ref{red:check}) ensure that the corresponding assignment of ordering~$\varphi_t$ satisfy clauses in~$F_t$, in case the query atom~$q_t^x$ for the last element~$x$ of~$\varphi_t$ holds.
Finally, Formulas~(\ref{red:checkunused}) avoids corner cases, where unused orderings could be taken, which would 
enable bypassing satisfiability. 

\begin{example}
Recall formula~$F$, ordering-augmented TD $\mathcal{T}'{=}(T',\chi',\varphi,\psi)$, 
as well as~$\varphi_{t^*}$ and $E_{t^*}$ 
from Example~\ref{ex:edges}.
We briefly sketch the rules generated for node~$t^*$.
\smallskip
\noindent\begin{tabular}{@{\hspace{0.15em}}l@{\hspace{0.15em}}@{\hspace{0.15em}}l@{\hspace{0.0em}}}
\toprule
(\ref{red:start})& $r_{s_{V_{t^*}}}\leftarrow$\\
(\ref{red:reached})& $\leftarrow r_{v_1^1}$;\quad $\leftarrow r_{v_1^2}$;\quad $\leftarrow r_{v_1^3}$;\quad $\leftarrow \neg r_{V_{d_{V_{t^*}}}}$\\
%
(\ref{red:choose})& $e_{s_{V_{t^*}}, v_1^1} \vee \hat e_{s_{V_{t^*}}, v_1^1} \leftarrow r_{s_{V_{t^*}}}$;\quad 
$e_{v_1^1, v_1^2} \vee \hat e_{v_1^1, v_1^2} \leftarrow r_{v_1^1}$;\\
& $e_{v_1^2, v_1^3} \vee \hat e_{v_1^2, v_1^3} \leftarrow r_{v_1^2}$;\quad $e_{v_1^3, d_{V_{t^*}}} \vee \hat e_{v_1^3, d_{V_{t^*}}} \leftarrow r_{v_1^3} $\\
(\ref{red:reach}) & $p_\epsilon^{v_1^1} \leftarrow e_{s_{V_{t^*}},v_1^1}$;\quad $p_\epsilon^{v_1^2} \leftarrow e_{v_1^1,v_1^2}$;\quad  $p_\epsilon^{v_1^3} \leftarrow e_{v_1^2,v_1^3}$ \\%
(\ref{red:reachfin}) & $r_{v_1^3} \leftarrow e_{v_1^3,d_{V_{t^*}}}$\\
\midrule
(\ref{red:fintrans}) & $o_{t^*}^{v_1^1} \leftarrow o_{t'}^{v_1^1}$;\; $o_{t^*}^{v_1^2} \leftarrow o_{t'}^{v_1^2}$;\; $o_{t^*}^{v_1^3} \leftarrow o_{t'}^{v_1^3}$;\; $o_{t^*}^{d_{V_{t^*}}} \leftarrow o_{t'}^{d_{V_{t^*}}}$ \\
(\ref{red:finnew}) & $o_{t^*}^{s_{V_{t^*}}} \leftarrow e_{s_{V_{t^*}}, v_1^1}$;\quad $o_{t^*}^{v_1^1} \leftarrow e_{v_1^1, v_1^2}$;\quad $o_{t^*}^{v_1^2} \leftarrow e_{v_1^2, v_1^3}$;\\
%
%
%
&$o_{t^*}^{v_1^3} \leftarrow e_{v_1^3, d_{V_{t^*}}}$\\
%
%
%
(\ref{red:fin}) & $\leftarrow o_{t'}^{s_{V_{t^*}}}, e_{s_{V_{t^*}},v_1^1}$;\quad $\leftarrow o_{t'}^{v_1^1}, e_{v_1^1, v_1^2}$;\quad $\leftarrow o_{t'}^{v_1^2}, e_{v_1^2, v_1^3}$;\\%
& $\leftarrow o_{t'}^{v_1^3},e_{v_1^3, d_{V_{t^*}}}$\\
\midrule
(\ref{red:copy}) & $p_{t^*}^{v_1^1} \leftarrow p_{t'}^{v_1^1}$;\quad $p_{t^*}^{v_1^2} \leftarrow p_{t'}^{v_1^2}$;\quad $p_{t^*}^{v_1^3} \leftarrow p_{t'}^{v_1^3}$ \\
(\ref{red:linkbypass}) & $p_{t^*}^{v_1^2} \leftarrow p_{t'}^{v_1^2}, \neg q_{t^*}^{v_1^1}$;\quad $p_{t^*}^{v_1^3} \leftarrow p_{t'}^{v_1^3}, \neg q_{t^*}^{v_1^2}$  \\
(\ref{red:linkbypass2}) & $p_{t^*}^{v_1^3} \leftarrow p_{t'}^{v_1^3}, \neg q_{t^*}^{v_1^3}$\\

\midrule

(\ref{red:linknopred}) & $q_{t^*}^{v_1^1} \leftarrow p_{t'}^{v_1^1} $ \\
(\ref{red:linkprevious}) & $q_{t^*}^{v_1^2} \leftarrow p_{t'}^{v_1^2}, q_{t^*}^{v_1^1} $;\quad  $q_{t^*}^{v_1^3} \leftarrow p_{t'}^{v_1^3}, q_{t^*}^{v_1^2} $ \\
(\ref{red:linkpq}) & $p_{t^*}^{v_1^1} \leftarrow q_{t^*}^{v_1^1}$;\quad $p_{t^*}^{v_1^2} \leftarrow q_{t^*}^{v_1^2}$;\quad $p_{t^*}^{v_1^3} \leftarrow q_{t^*}^{v_1^3}$\\
%
%
\bottomrule
\end{tabular}%

\noindent Note that if~$\mathcal{I}_{t^*}(\varphi_{t^*})\not\models F_{t^*}$ or~$\mathcal{I}_{t^*}(\varphi_{t^*})$ is undefined (unused ordering), Formulas~(\ref{red:check}) and~(\ref{red:checkunused}) 
generate~$\leftarrow q_{t^*}^{v_1^3}$.
\end{example}

%
%
\def\rddots#1{\cdot^{\cdot^{\cdot^{#1}}}}
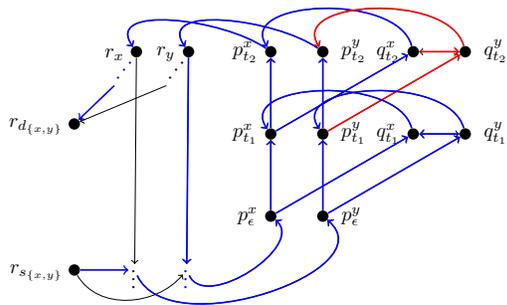
\begin{figure}[t]
\centering%
    \resizebox{.85\linewidth}{!}{%
	\begin{tikzpicture}[node distance=7mm,every node/.style={fill,circle,inner sep=2pt}]%
		\node (s1) [label={[text height=1.5ex,yshift=0.0cm,xshift=0.05cm]left:$q_{t_1}^x$}] {};
		\node (s2) [right=of s1,label={[text height=1.5ex,yshift=0.0cm,xshift=0.05cm]right:$q_{t_1}^y$}] {};
%
%
	%
		\node (ps1) [left=6.5em of s1,label={[text height=1.5ex,yshift=0.0cm,xshift=0.05cm]left:$p_{t_1}^x$}] {};
		\node (ps2) [right=of ps1,label={[text height=1.5ex,yshift=0.0cm,xshift=0.05cm]right:$p_{t_1}^y$}] {};
		%
		%
		\node (d1) [above=3.5em of s1,label={[text height=1.5ex,yshift=0.0cm,xshift=0.05cm]left:$q_{t_2}^x$}] {};
		\node (d2) [right=of d1,label={[text height=1.5ex,yshift=0.0cm,xshift=0.05cm]right:$q_{t_2}^y$}] {};
%
%
		\node (pd1) [left=6.5em of d1,label={[text height=1.5ex,yshift=0.0cm,xshift=0.05cm]left:$p_{t_2}^x$}] {};
		\node (pd2) [right=of pd1,label={[text height=1.5ex,yshift=0.0cm,xshift=0.05cm]right:$p_{t_2}^y$}] {};
%
		\node (e1) [below=3.5em of ps1,label={[text height=1.5ex,yshift=0.0cm,xshift=0.05cm]left:$p_\epsilon^x$}] {};
		\node (e2) [right=of e1,label={[text height=1.5ex,yshift=0.0cm,xshift=0.05cm]right:$p_\epsilon^y$}] {};

%
%
		\node (f2) [left=3.5em of pd1,label={[text height=1.5ex,yshift=0.0cm,xshift=0.05cm]left:$r_y$}] {};
		\node (f1) [left= of f2,label={[text height=1.5ex,yshift=0.0cm,xshift=0.05cm]left:$r_x$}] {};
		\node (fg1) [white,below left=.25em of f1,yshift=1.35em,xshift=-.5em,label={[text height=1.5ex,yshift=-0.2cm,xshift=0.9em,rotate=90]left:\textcolor{blue}{$\ddots{}$}}] {};
		\node (fg2) [white,below left=.25em of f2,yshift=1.35em,xshift=-.5em,label={[text height=1.5ex,yshift=-0.2cm,xshift=0.9em,rotate=90]left:{$\ddots{}$}}] {};
		\node (rd2) [below=3em of f1,xshift=-3.1em,label={[text height=1.5ex,yshift=0.0cm,xshift=0.05cm]left:$r_{d_{\{x,y\}}}$}] {};
		\node (g1) [white,below left=9em of ps1,label={[text height=1.5ex,yshift=-0.2cm,xshift=0.9em]left:$\vdots$}] {};
		\node (g2) [white,right= of g1,label={[text height=1.5ex,yshift=-0.2cm,xshift=0.9em]left:\textcolor{blue}{$\vdots$}}] {};
		\node (x2) [left=3em of g2,xshift=-2em,label={[text height=1.5ex,yshift=0.0cm,xshift=0.05cm]left:$r_{s_{\{x,y\}}}$}] {};
		
		\draw [thick,blue,->] (x2) to (g1);
		\draw [->,out=-60,in=230] (x2) to (g2);

		\draw [->] (e1) to (s1);
		\draw [->] (e2) to (s2);
		\draw [thick,blue,->] (e2) to (s2);
		
		\draw [->] (e1) to (ps1);
		\draw [->] (e2) to (ps2);

		\draw [thick,blue,->] (e1) to (ps1);
		\draw [thick,blue,->] (e2) to (ps2);

		\draw [out=145,in=125,->] (pd1) to (f1);
		\draw [out=145,in=125,->] (pd2) to (f2);
		\draw [thick,blue,out=145,in=125,->] (pd1) to (f1);
		\draw [thick,blue,out=145,in=125,->] (pd2) to (f2);
		\draw [->] (f1) to (g1);
		\draw [->] (f2) to (g2);
		\draw [thick,blue,->] (f2) to (g2);
		\draw [thick,blue,->] (e1) to (s1);
		\draw [->,out=-65,in=-40] (g2) to (e1);
		\draw [->,out=-65,in=-40] (g1) to (e2);
		\draw [thick,blue,->,out=-65,in=-40] (g1) to (e2);
		\draw [thick,blue,->,out=-65,in=-40] (g2) to (e1);
%
		\draw [thick,blue,->] ($(fg1.south) + (0,-.75)$) to (rd2);
		\draw [->] ($(fg2.south) + (0,-.75)$) to (rd2);
		\draw [<->] (d1) to (d2);
		\draw [red,thick,->] (d1) to (d2);
		%
		\draw [<->](s1) -- (s2);
		\draw [thick,blue,<-](s1) -- (s2);
%
%
		\draw [->] (ps1) to (pd1);
		\draw [thick,blue,->] (ps1) to (pd1);
	%
		\draw [->] (ps2) to (pd2);
		\draw [thick,blue,->] (ps2) to (pd2);
%

		\draw [->] (ps1) to (d1);
		\draw [->] (ps2) to (d2);
		\draw [thick,red,->] (ps2) to (d2);

		\draw [thick,blue,->] (ps1) to (d1);
		
		\draw [out=105,in=125,->] (d1) to (pd1);
		\draw [thick,blue,out=105,in=125,->] (d1) to (pd1);
		\draw [out=105,in=125,->] (d2) to (pd2);
		\draw [red,thick,out=105,in=125,->] (d2) to (pd2);
		%
		%
		\draw [out=105,in=125,->] (s1) to (ps1);
		\draw [thick,blue,out=105,in=125,->] (s1) to (ps1);
		\draw [out=105,in=125,->] (s2) to (ps2);
		\draw [thick,blue,out=105,in=125,->] (s2) to (ps2);
	\end{tikzpicture}}\vspace{-2.5em}
\caption{Positive dependency graph of the reduction for two TD nodes~$t_1$, $t_2$ over ordering vertices~$V_{t_1}{=}V_{t_2}{=}\{x,y\}$. These blocks potentially cause cycles (larger SCCs). The parts highlighted in blue are the atoms in their order of derivation, assuming that~$y$ precedes~$x$, i.e., $y \prec_{\varphi_{t_1}} x$. Edges highlighted in red cannot be taken, i.e., $q_2^y$ cannot be derived, since this implies~$x \prec_{\varphi_{t_2}} y$, contradicting $y \prec_{\varphi_{t_1}} x$; causing a cycle of unproven atoms involving~$r_y$.}\label{fig:sketch}
\end{figure}

\futuresketch{
\clearpage
\begin{example}
Recall program~$\Pi$ from Example~\ref{ex:running1}, and TD~$\mathcal{T}$ of $\mathcal{G}_\Pi$ given in Figure~\ref{fig:graph-td}.
We briefly show Formula~$F$ for node~$t_3$. 

\smallskip
\noindent\begin{tabular}{@{\hspace{0.15em}}l@{\hspace{0.15em}}|@{\hspace{0.15em}}l@{\hspace{0.0em}}}
Formulas & Formula $F$\\
\hline
(\ref{red:checkrules})& $\neg b \vee d \vee e$; $\neg e \vee d \vee b$; $d \vee b$\\
(\ref{red:prop})& $(d \prec_{t_1} e) \leftrightarrow (d \prec_{t_3} e)$; $(e \prec_{t_1} d) \leftrightarrow (e \prec_{t_3} d)$\\
%
(\ref{red:checkremove})& $c \rightarrow p^c_{<t_1}$; $a \rightarrow p^a_{<t_2}$\\
(\ref{red:checkremove2}) & $b \rightarrow p^b_{<t_3}$; $d \rightarrow p^d_{<t_3}$; $e \rightarrow p^e_{<t_3}$\\ 
(\ref{red:check}) & $p^b_{<t_3} \leftrightarrow (p^b_{t_3} \vee p^b_{<t_2})$;\\
 		  & $p^d_{<t_3} \leftrightarrow (p^d_{t_3} \vee p^d_{<t_1})$; $p^e_{<t_3} \leftrightarrow (p^e_{t_3} \vee p^e_{<t_1})$\\ 
(\ref{red:checkfirst})& $p^b_{t_3} \leftrightarrow [e \wedge b \wedge (e \prec_{t_3} b) \wedge \neg d]$;\\ 
		      & $p^d_{t_3} \leftrightarrow [(b \wedge d \wedge (b \prec_{t_3} d) \wedge \neg e) \vee (d \wedge \neg b)]$;\\
		      & $p^e_{t_3} \leftrightarrow [b \wedge e \wedge (b \prec_{t_3} e) \wedge \neg d]$
\end{tabular}%
\vspace{.25em}
\end{example}
}

\vspace{-.1em}
\subsection{Properties and Consequences of the Reduction}

Next, we show that the reduction indeed utilizes the structural parameter treewidth,
i.e., the treewidth is decreased.

\begin{theorem}[Treewidth-Awareness]\label{thm:tw}
The reduction from a Boolean formula~$F$ and an ordering-augmented TD~$\mathcal{T}{=}(T,$ $\chi,\varphi,\psi)$ of~$\mathcal{G}_F$ to normal program~$\Pi$ consisting of Formulas~(\ref{red:start}) to~(\ref{red:checkunused}) slightly decreases treewidth.
Precisely, if~$k$ is the width of $\mathcal{T}$,
the treewidth of~$\mathcal{G}_\Pi$ is in~$\mathcal{O}(\frac{k}{\log(k)})$.
\end{theorem}

\begin{proof}[Proof (Sketch)]
\vspace{-.55em}
We construct a TD~$\mathcal{T}'=(T,\chi')$ of~$\mathcal{G}_\Pi$ to show that the width of~$\mathcal{T}'$ increases only slightly (compared to~$k$).
To this end, let~$t$ be a node of~$T$ with~$\children(t)=\langle t_1, \ldots, t_\ell \rangle$
and let~$\hat t$ be the parent of~$t$ (if exists).
We define (i)~$R(t)\eqdef \{r_x \mid x\in V_t\}\cup \{r_x \mid x\in V_{t'}, t'\in\children(t)\}\cup\{r_{s_{V_t}}, r_{d_{V_t}}\}$, 
(ii)~$E(t)\eqdef \{e_{x,y}, \hat e_{x,y}\mid e_{x,y}\in E_t\}$, 
(iii)~$P(t)\eqdef \{q_t^x, p_t^x, p_{t'}^x, p_\epsilon^x \mid x\in \dom(\varphi_t)\cup\dom(\psi_t), t'\in\prev(x,t)\}$, and
(iv)~$O(t)\eqdef \{o_t^x, o_t^{s_{V_t}} \mid x\in V_t\}\cup\{ o_{t'}^x, o_{t'}^{s_{V_{t'}}} \mid t'\in\children(t), x\in V_{t'}\cap V_t\}$.
Then, we let
$\chi'(t) \eqdef R(t) \cup E(t) \cup P(t) \cup O(t)$.
%
%
Observe that~$\mathcal{T}'$ is a TD of~$\mathcal{G}_\Pi$ and by construction~$\Card{\chi'(t)}$ is in $\mathcal{O}(\Card{V_t})$.
%
By definition of~$V_t$, $\Card{V_t}\leq k'$, where~$k'! \geq 2^k$.
Then, since $k'!$ is in~$2^{\Omega(k'\cdot\log(k'))}$ (see Observation~\ref{obs:klogk}), we have that~$2^{\Omega(k'\cdot\log(k'))}$ is at least~$2^k$ and therefore~$k'=\mathcal{O}(\frac{k}{\log(k')})=\mathcal{O}(\frac{k}{\log(k)-\log(\log(k'))})=\mathcal{O}(\frac{k}{\log(k)})$. 
\futuresketch{
by constructing~$2^{\mathcal{O}(k'\cdot\log(k'))}$ many interpretations for some~$k'$.
In order to quantify~$k'$, we take the $\log$ on both sides and end up
with $k$ being in $\mathcal{O}(k'\cdot\log(k'))$.
As a result, we have that~$k'$ is bound by~$\mathcal{O}(\frac{k}{\log(k')})$.
Since~$\log(k')$ is in $\Omega(\log(k))$, we end up with~$k'$ being bounded by~$\mathcal{O}(\frac{k}{\log(k)})$.
}
%
%
%
%
\end{proof}


Interestingly, our reduction cannot be significantly improved, making
further treewidth decreases unlikely. 

\begin{theorem}[Treewidth Decrease is Optimal]\label{thm:runtime1}
Assume a reduction 
from a formula~$F$ 
to a normal program~$\Pi$, running in time~$2^{o(k)}\cdot\poly(\Card{\var(F)})$, where~$k=\tw{\mathcal{G}_F}$.
Then, unless ETH fails, the treewidth of~$\mathcal{G}_\Pi$
cannot be in~$o(\frac{k}{\log(k)})$. 
\end{theorem}
\begin{proof}
\vspace{-.5em}
Assume towards a contradiction that such a reduction, call it~$\mathcal{R}^*$, exists.
Then, we apply this reduction on any~$F$, resulting in program~$\Pi=\mathcal{R}^*(F)$.
Then, we know that~$\Pi$ can be decided~\cite{Hecher22} in time~$2^{\mathcal{O}(k'\cdot\log(k'))}\cdot\poly(\Card{\var(F)})$, 
where~$k'$ is in~$o(\frac{k}{\log(k)})$. 
As a result, we have that~$\Pi$ and therefore~$F$ can be decided in time~$2^{o(\frac{k}{\log(k)}\cdot\log(k))}\cdot\poly(\Card{\var(F)})$, which is in~$2^{o(k)}\cdot\poly(\Card{\var(F)})$, contradicting the ETH. 
%
\end{proof}

\futuresketch{
Later we will see the lower bound for deciding $\iota$-tight programs based on a slightly adapted version of the reduction,
which implies that we cannot expect to significantly improve this decrease of treewidth. 
Next, we present consequences for auxiliary variables and runtime.

\begin{corollary}[Runtime]
The reduction from a HCF program~$\Pi$ and a nice TD~$\mathcal{T}$ 
of~$\mathcal{G}_\Pi$ to SAT formula~$F$ 
consisting of Formulas~(\ref{red:checkrules}) to~(\ref{red:checkfirst}) 
uses at most~$\mathcal{O}(k\cdot\log(k)\cdot h)$ many variables
and runs in time~$\mathcal{O}(k\cdot \log(k)\cdot h + \Card{\Pi})$,
where~$k$ and $h$ form the width and the number of nodes of~$\mathcal{T}$, respectively.
\end{corollary}
\begin{proof}
The result follows from Theorem~\ref{thm:runtime1}.
Linear time in~$\Pi$ can be obtained 
by slightly modifying Formulas~(\ref{red:checkrules}) and~(\ref{red:checkfirst})
such that each rule~$r\in \Pi$ is used in only one node~$t$,
where~$r\in\Pi_{t'}$, but~$r\notin\Pi_t$, for some~$t'\in\children(t)$.
\end{proof}
Note that a nice TD of~$\mathcal{G}_\Pi$ of width~$k=\tw{\mathcal{G}_\Pi}$, 
having only~$h=\mathcal{O}(\Card{\at(\Pi)})$ many nodes~\cite{Kloks94a}[Lem.\ 13.1.2] always exists. 
Further, since $k\cdot\log(k)$ might be much smaller than~$\log(\Card{\at(\Pi)})$,
for some programs this reduction might pay off compared to global or component-based
orderings used in tools like lp2sat~\cite{Janhunen06} or lp2acyc~\cite{GebserJanhunenRitanen14,BomansonEtAl16}.}

%
\futuresketch{
Before we show correctness, we need to formally define the concepts of level mappings over a set of atoms.
Given a program~$\Pi$, and a set~$A\subseteq\at(\Pi)$ of atoms. Then, a function~$\varphi: A \rightarrow \{0,\ldots,\Card{A}\}$ is an \emph{ordering} for~$\Pi$ over~$A$.
\begin{definition}\label{def:provinglvmapping}
Given a level mapping~$\varphi$ for~$\Pi$ over~$\at(A)$.
Then, $\varphi$ is \emph{proving}
if there is an answer set~$M$ of~$\Pi$, called \emph{proven} answer set, such that $a\in M$ if and only if
there is a rule~$r\in \Pi$, called \emph{justifying rule for~$a$ of~$\varphi$}, with: (1) $a\in H_r$, (2) $M\cap (H_r\setminus\{a\}\cup B_r^-)=\emptyset$, and (3) for each~$b\in B_r^+$ we have~$\varphi(b)<\varphi(a)$.
\end{definition}

For each answer set of a program, there exists a level mapping.

\begin{proposition}[\cite{Janhunen06}]\label{prop:mapping}
Given a program~$\Pi$. Then, for each answer set~$M$ of~$\Pi$ there is a level mapping~$\varphi$ over~$\at(\Pi)$ proving~$M$.
\end{proposition}}
Correctness
establishes that the reduction~$\mathcal{R}$
encoded by Formulas~(\ref{red:start}) to~(\ref{red:checkunused})  indeed characterizes the satisfying assignments of a Boolean formula.


\begin{theorem}[Correctness]\label{thm:corr}
The reduction from a Boolean formula~$F$ and a TD~$\mathcal{T}=(T,\chi)$ of~$\mathcal{G}_\Pi$ to a logic program~$\Pi$, consisting of Formulas~(\ref{red:start}) to~(\ref{red:checkunused}), is correct.
Concretely, for every model of~$F$, there is an answer set of~$\Pi$. Vice versa, for every answer set of~$\Pi$, there is a unique model of~$F$. 
\end{theorem}

\begin{example}
Figure~\ref{fig:sketch} sketches the dependency graph over the rules of Blocks 1,3, and 4 on two simple TD nodes; showing
how incompatible queries would cause cyclic dependencies
that are unproven, which can therefore not occur. 
\end{example}

\futuresketch{
\begin{definition}
Given a level mapping~$\varphi$ for~$\Pi$ over~$A$.
Then, $\varphi$ is \emph{proving}
if there is an answer set~$M$, called \emph{proving} answer set, of~$\Pi$ such that $a\in M$ if and only if
there is a rule~$r\in \Pi$, called \emph{justifying rule for~$a$ of~$\varphi$}, with: (1) $a\in H_r$, (2) $M\cap (H_r\setminus\{a\}\cup B_r^-)=\emptyset$, and (3) for each~$b\in B_r^+$ we have~$\varphi(b)<\varphi(a)$.
\end{definition}

\begin{definition}
Given a level mapping~$\varphi$ for~$\Pi$ over~$A$.
Then, $\varphi$ is \emph{canonical} if there exists an answer set~$M$ of~$\Pi$ such that
(1) $a\in M$ if and only if $\varphi(a)>0$, (2) for every~$a\in A$ with~$\varphi(a) > 0$ there exists~$a'\in A$ with $\varphi(a') - \varphi(a) = 1$, and (3) for every~$a\in M$, and every justifying rule~$r\in\Pi$ for~$a$ of~$\varphi$, there is some~$b\in B_r^+$ with~$\varphi(b) - \varphi(a) = 1$.
\end{definition}

\begin{proposition}\label{prop:mappingunique}
Given a program~$\Pi$. Then, for each proving, canonical level mapping~$\varphi$ for~$\Pi$ over~$\at(\Pi)$ there is one unique proving answer set~$M$ of~$\Pi$ for~$\varphi$ and vice versa.
\end{proposition}
\begin{proof}
$\Rightarrow$: Assume towards a contradiction that there are two different answer sets~$M,M'$ of~$\Pi$ for~$\varphi$. Since~$M, M'$ are both subset-minimal,
there is some~$a\in M$, such that $a\notin M'$. 
Then, by canonicity, $\varphi(a)=0$ since $a\notin M'$, but~$\varphi(a)\neq 0$ due to~$M$.

$\Leftarrow$: Assume that for an answer set~$M$ there are two canonical level mappings~$\varphi,\varphi'$ such that~$M$ is proving for~$\varphi,\varphi'$.
By (1) of canonicity, for~$a\in M$, we have $\varphi(a)>0$.
Assume towards a contradiction that there is some~$a\in M$ with~$\varphi(a) < \varphi'(a)$.
But then $\varphi'$ cannot be canonical, as it violates (3) of the canonicity definition, as at least one justifying rule exists since~$\varphi$ is proving.
\end{proof}


Given a TD~$\mathcal{T}=(T,\chi)$ of~$\mathcal{G}_\Pi$, and a node~$t$ of~$T$.
We refer to a canonical level mapping of~$\Pi_t$ over~$\chi(t)$ by \emph{$t$-local level mapping~$\varphi_t$}. 
Then, we refer by \emph{$\mathcal{T}$-local level mappings} to the set~$\mathcal{M}$ consisting of one~$t$-local level mapping~$\varphi_t$ for each~$t$ of~$T$
such that we have \emph{compatibility} as follows: For each nodes~$t,t'$ of~$T$ and every~$a,b\in\chi(t)\cap\chi(t')$, we have (1) $\varphi_t(a) = 0$ if and only if $\varphi_{t'}(a)=0$, (2) whenever $\varphi_t(a) < \varphi_t(b)$  then~$\varphi_{t'}(a) < \varphi_{t'}(b)$, and (3) for each justifying rule~$r\in\Pi_{t}$ for~$a$ respecting~$\varphi_{t}$, we require that~$\varphi_t$ gives the maximum level for~$a$ in~$T$, i.e., $\varphi_t(a)=\max_{t''\text{ of } T, a\in\chi(t'')}(\varphi_{t''}(a))$.

Further, $\mathcal{M}$ requires \emph{provability}, where for each $a\in\at(\Pi)$, there exists a node~$t$ of~$T$ s.t.\ either~$\varphi_t(a)=0$, or there is a justifying rule~$r\in\Pi_{t}$ for~$a$ respecting~$\varphi_{t}$.
%
%
%

%
\begin{lemma}[Equivalence of level mappings]
Given a program~$\Pi$, and a TD~$\mathcal{T}=(T,\chi)$ of~$\mathcal{G}_\Pi$.
Then, 
there is a bijective correspondence between a set~$\mathcal{M}$ of~$
\mathcal{T}$-local level mappings of~$\varphi$ and a proving, minimal level mapping~$\varphi$.
Concretely, for each such set~$\mathcal{M}$ there is a unique mapping~$\varphi$ (and vice versa) assuming that for each~$a\in\at(\Pi)$, $\varphi(a) = 0$ if and only if there is $\varphi_t\in\mathcal{M}$ with~$\varphi_t(a) = 0$.
\end{lemma}
\begin{proof}
$\Rightarrow$: Given the set~$\mathcal{M}$ of~$\mathcal{T}$-local level mappings. 
Then, we construct a proving, canonical level mapping~$\varphi$ as follows.
We set~$\varphi(a)\eqdef 0$ for each~$a\in \at(\Pi)$, where
there exists a node~$t$ of~$T$ with~$\varphi_t(a)=0$.
Then, we set~$\varphi(a)\eqdef 1$ for each~$a\in \at(\Pi)$,
where there is no node~$t$ of~$T$ with~$\varphi_t(b) < \varphi_t(a)$ for some~$b$
with~$\varphi(b)=0$, and so on.
In turn, we construct~$\varphi$ in rounds, where each round assigns an increasing value of~$2,3,\ldots$.
Observe that $\varphi$ is well-defined, i.e., each atom~$a\in\at(\Pi)$ gets a unique value since~$\mathcal{M}$ only contains canonical mappings and by compatibility of~$\mathcal{M}$.

Next, we show that~$\varphi$ is a proving, canonical level mapping of~$\Pi$.
Indeed, $\varphi$ is canonical by canonicity of each~$t$-local level mapping within~$\mathcal{M}$ and compatibility of~$\mathcal{M}$.
In particular, by (3) of compatibility of~$\mathcal{M}$,
for each atom~$a$, every nodes~$t',t''$ of~$T$,
and every justifying rule for~$a$ respecting~$\varphi_{t'}, \varphi_{t''}$, respectively, we have~$\varphi_{t'}(a)=\varphi_{t''}(a)$.
Further, $\varphi$ is proving by provability of~$\mathcal{M}$ and since the properties of TDs ensure that for TD~$\mathcal{T}$ for each rule~$r\in\Pi$ there is a node~$t$ of~$T$ with~$r\in\Pi_t$.

Towards a contradiction assume that there is a proving, canonical level mapping~$\varphi'$ 
for~$\Pi$ that is incomparable to~$\varphi$ with~$\varphi(a)=0$ if and only if~$\varphi'(a)=0$ for each~$a\in\Pi_t$.
However, both~$\varphi,\varphi'$ lead by construction to the same proving answer set~$M$. 
Consequently, by Proposition~\ref{prop:mappingunique}, $\varphi=\varphi'$, which leads to a contradiction.

$\Leftarrow$:
Assuming proving, canonical level mapping $\varphi$, and the proving answer set~$M$ of~$\varphi$.
Then, we define the \emph{canonical $t$-local level mapping~$\hat\varphi_t$} for each~$t$ of~$T$ as follows.
\noindent$\hat\varphi_t(a) {\eqdef} \begin{cases}0 & \text{for each } a\in \chi(t) \text{ with } \varphi(a) = 0,\\
\rank_{\chi(t)}(a,\varphi) &\text{for each } a\in \chi(t) \text{ with } \varphi(a) \neq 0,\end{cases}$
where~$\rank_A(a,\varphi)$ is the ordinal number of $a$ according to~$\varphi(a)$ among all elements in~$A$.
The resulting set~$\mathcal{M}$ of~$\mathcal{T}$-local level mappings consisting of~$\hat\varphi_{t}$ for each~$t$ of~$T$ is well-defined.
Observe that any~$\hat\varphi_t$ is canonical, as for~$M$, (1) and (3) of canonicity remains satisfied since the relation between level numbers remains as in~$\varphi$ (by construction of~$\hat\varphi_t$). Also Condition (2) is satisfied by construction of~$\hat\varphi_t$, since the absolute numbers are within~$1,\ldots, \Card{\chi(t)}$ and~$\varphi$ is canonical. Further, compatibility of~$\mathcal{M}$ holds, as well as provability, as by properties of a TD any rule~$r\in\Pi$ is in some node~$t$, i.e., $r\in\Pi_t$.

Next, we show that the set~$\mathcal{M}$ of~$\mathcal{T}$-local level mappings consisting of~$\hat\varphi_{t}$ for each~$t$ of~$T$ is unique.
Assume towards a contradiction that there is an alternative set~$\mathcal{M}'$ of~$\mathcal{T}$-local level mappings~$\psi_t$ for each~$t$ of~$T$. 
Then, by compatibility of~$\mathcal{M}'$ and assumptions of this lemma, $\psi_t(a)=0$ if and only if~$\hat\varphi_t(a)=0$.
Consequently, there must be an atom~$a\in\at(\Pi_t)$ with~$\psi_t(a) \neq \varphi_t(a)$ such that~$\psi_t(a) > 0$ and~$\varphi_t(a) > 0$.
We distinguish the cases (a) $\psi_t(a) > \varphi_t(a)$ and (b) $\varphi_t(a) > \psi_t(a)$. 

\noindent Case (a):~$\psi_t(a) > \varphi_t(a)$. if~$\psi_t(a) > \max_{t'\text{ of }T, a\in\chi(t')}\psi_{t'}(a)$,
then (3) of compatibility of~$\mathcal{M}'$ is dissatisfied.
Then, there is by canonicity~$b\in\chi_t$ where~$\psi(a)>\psi(b)$, but~$\varphi_t(a)\leq\varphi_t(b)$.
There cannot be a node~$t'$ of~$T$, with~$t\neq t'$ where~$b\in\chi_{t'}$,
otherwise~$\mathcal{M}'$ does not have compatibility. Consequently, since~$b$ was arbitrarily chosen, this has to hold for every such atom~$b$.
Therefore

\noindent Case (b):~$\varphi_t(a) > \psi_t(a)$. Then, there cannot be a rule~$r\in \Pi$
\end{proof}

\begin{theorem}
The reduction~$R$ from a HCF program~$\Pi$ and a TD~$\mathcal{T}=(T,\chi)$ of~$\mathcal{G}_\Pi$ to SAT formula~$F$ consisting of Formulas~(\ref{red:checkrules}) to~(\ref{red:cnt:checkfirst}) is correct.
Concretely, for each answer set of~$\Pi$ there exists exactly one model of~$F$ and vice versa. TODO: make it precise!
\end{theorem}

\begin{proof}
Given any answer set~$M$ of~$\Pi$. 
Then, by Proposition~\ref{prop:mappingunique} there has to be a unique, minimal level mapping~$\varphi: \at(\Pi) \rightarrow \Nat$, proven by~$M$.
We construct a model~$I$ of~$F$ as follows.
For each~$x\in\at(\Pi)$ we let $I(x)\eqdef 1$ if $x\in M$ and~$I(x)\eqdef 0$ otherwise.
For each node~$t$ of~$T$, 
and~$x\in\chi(t)$ we assign the following variables:
(1) For every~$l\in \bvali{x}{t}{i}$ where~$i=\rank_{\chi(t)}(x,\varphi)$, we set $I(b)\eqdef 1$ if~$l$ is an atom and~$I(b)\eqdef 0$ if~$l$ is not an atom. 
(2) If there is a justifying rule~$r$ for~$x$ respecting~$\varphi$, we set~$I(p_{x_t})\eqdef I(pf_{x_t})\eqdef 1$.
(3) If~$I(p_{x_{t'}})=1$ for~$t'\in\children(t)$, then we set~$I(p_{x_t})\eqdef 1$.
%

%
\end{proof}

\subsubsection{Parameterized Algorithm for Counting and Enumeration.}
The reduction presented in the previous section immediately gives rise to an algorithm for counting and enumerating answer sets of an \ASP program.

\begin{proposition}
Runtime for counting and enumerating with linear delay.
\end{proposition}

\begin{corollary}
Counting by reduction
\end{corollary}

\begin{corollary}
Enumeration by reduction, mention linear delay
\end{corollary}
}

\section{Further Insights Into Hardness of \ASP} 
In this section, we provide deeper insights into the characterization of hardness for normal logic programs.
First, we discuss consequences on the length of the largest SCC.

\subsection{Are Unbounded Cycles (SCCs) Vital for Hardness?} 

By the construction of the reduction in Section~\ref{sec:main} and the observation that~$\mathcal{R}$ causes cycles (SCCs) in the program's dependency graph of size~$2^{\mathcal{O}(k\cdot\log(k))}$,
we obtain the following new hardness and precise lower bound result for deciding the consistency 
of normal programs.

\begin{corollary}[LB Largest SCC]\label{cor:hcycle}
Let~$\Pi$ be a normal logic program, where the treewidth of~$\mathcal{G}_\Pi$ is~$k$
such that the largest SCC size of~$D_\Pi$ is in~$2^{\mathcal{O}(k\cdot\log(k))}$.
Then, unless ETH fails, 
the consistency of~$\Pi$ can not be decided in time~$2^{o(k\cdot\log(k))}\cdot\poly(\Card{\at(\Pi)})$.
\end{corollary}
Corollary~\ref{cor:hcycle} gives insights into the SCC size required for hardness, which is in contrast to known lower bounds, which could not bound the cycle length or SCC size in the treewidth~\cite{Hecher22}.
Interestingly, this corollary is 
in line with the corresponding upper bound of Proposition~\ref{prop:iota} (for~$\iota=k$). 
We do not expect that Corollary~\ref{cor:hcycle} can be 
significantly strengthened, but we show below how for $\iota$-tight programs the SCC size can be decreased to~$2^{o(k\cdot\log(\iota))}$. 
\futuresketch{
If the largest SCC size was indeed significantly smaller, say polynomial, then actually runtimes~$2^{o(k\cdot\log(k))}\cdot\Card{\at(\Pi)}$ were possible, since only~$o(\log(k))$ bits per bag atom were required for solving. 
}

This also provides the corresponding lower bound for projected answer set counting,
which was left open~\cite{FichteHecher19}. Our reduction allows us to close the gap to the  upper bound, showing that it is expected for the problem to be harder than plain counting on disjunctive programs.

\begin{theorem}[LB Projected Counting]\label{cor:pcount}
Let~$\Pi$ be a normal logic program, $P{\subseteq}\at(\Pi)$ be atoms, and~$k$ be the treewidth of~$\mathcal{G}_\Pi$, such that the largest SCC size of~$D_\Pi$ is in~$2^{\mathcal{O}(k\cdot\log(k))}$.
Then, under ETH, the cardinality~$\Card{\{M\cap P \mid M\models \Pi\}}$ cannot be computed in time~$2^{2^{o(k\cdot\log(k))}}\cdot\poly(\Card{\at(\Pi)})$.
\end{theorem}

%

\subsection{An ETH-Tight Lower Bound for $\iota$-Tight Programs} 

Recall the fragment of~$\iota$-tight programs, which is motivated by the idea of providing almost tight programs that are simpler to solve than normal programs. 
%
%
Every formula~$F$ can be compiled into a~($1$-)tight program, cf., Example~\ref{ex:encoding}. 

As a result, in the following, we generalize our reduction~$\mathcal{R}$ of the previous section from~$\iota=k$ to the case~$\iota \geq 2$, 
resulting in~$\mathcal{R}'$.
Intuitively, this reduction~$\mathcal{R}'$ allows us to decrease the treewidth~$k$, but not necessarily to the maximum of~$\mathcal{O}(\frac{k}{\log(k)})$ of normal programs. 
Instead, $\iota$ provides a precise handle on the tightness,
thereby decreasing treewidth to~$\mathcal{O}(\frac{k}{\log(\iota)})$.

\paragraph{Adapted Reduction.} We adapt the construction of reduction~$\mathcal{R}$ and obtain~$\mathcal{R}'$.
\futuresketch{where instead of 
%
SCC sizes of the positive dependency graph up to~$2^{\mathcal{O}(k\cdot\log(k))}$,
the adapted approach causes SCC sizes up to~$2^{\mathcal{O}(k\cdot\log(\iota))}$ for some integer~$\iota$, $2\leq \iota\leq k$. 
}
%
%
To this end, we take an instance~$F$ of \SAT, i.e., a Boolean formula, and an ordering-augmented TD~$\mathcal{T}=(T,\chi,\varphi,\psi)$ of~$\mathcal{G}_F$ of width~$k$.
Then, we simulate for each node of~$T$, the up to~$2^k$ many assignments via~$(\iota!)^{\frac{k'}{\iota}}$ many orderings, where~$\iota$ is any fixed~$2\leq \iota \leq k'$.
So we decrease the treewidth from~$k$ to~$k'$ such that~$(\iota!)^{\frac{k'}{\iota}} \geq 2^k$, where the special case of the previous section  corresponds to~$\iota=k'$.
Consequently, since there are still up to~$k'$ many elements per bag, but we only have~$\iota$ many positions, we require~$\frac{k'}{\iota}$ many SCCs per bag.
%
As a result, the orderings~$\ord(t)$ for a node~$t$ are not total.
So, instead of one source and destination vertex for ordering set~$V_t$, we require up to $\frac{k'}{\iota}$ reachability atoms of the form~$r_{s_{V_t}^j}, r_{d_{V_t}^j}$ for every~$1\leq j\leq\frac{k'}{\iota}$. This requires minor adaptions in the definition of~$E_t$, Formulas~(\ref{red:start}), (\ref{red:reached}), (\ref{red:reachfin}), as well as~(\ref{red:fintrans}) and~(\ref{red:fin}), as sketched in the appendix. 

\smallskip
\noindent One can show a generalization of Observation~\ref{obs:klogk},
where~$\iota < k$. 

\begin{observation}\label{obs:logk}
Let 
$2\leq \iota\leq k$. Then, $(\iota!)^\frac{k}{\iota}$ is in~$2^{\Omega(k\cdot\log(\iota))}$.
\end{observation}
\futuresketch{\begin{proof}\vspace{-.5em}
By Observation~\ref{obs:klogk}, $\iota!$ is bounded by ~$2^{\Omega(\iota\cdot\log(\iota))}$. As a result, $(\iota!)^\frac{k}{\iota}$ is then in~$2^{{\Omega(\iota\cdot\log(\iota))}^\frac{k}{\iota}}$, which is in~$2^{{\frac{k}{\iota}\cdot\Omega(\iota\cdot\log(\iota))}}$, yielding the result: $(\iota!)^\frac{k}{\iota}$ is in~$2^{{\Omega(k\cdot\log(\iota))}}$.
\end{proof}}

%
By Observation~\ref{obs:logk}, $(\iota!)^{\frac{k'}{\iota}}$ is in~$2^{\Omega(k'\cdot\log(\iota))}$, so we have that~$2^{\Omega(k'\cdot\log(\iota))}$ is at least~$2^k$ and therefore~$k'=\mathcal{O}(\frac{k}{\log(\iota)})$. 
As a result, $\mathcal{R}'$
slightly reduces treewidth to $\mathcal{O}(\frac{k}{\log(\iota)})$. 
%
%
%
%
%
%
%
%
%
%
Consequently,   for the result as given in Proposition~\ref{prop:iota}
it is unexpected that it can be significantly improved (under ETH). 
More precisely, we obtain the following lower bound result.

\begin{theorem}[LB $\iota$-Tightness]\label{cor:hscc}
Let~$\Pi$ be a $\iota$-tight logic program, where the treewidth of~$\mathcal{G}_\Pi$ is~$k$
such that the largest SCC size of~$D_\Pi$ is in~$2^{\mathcal{O}(k\cdot\log(\iota))}$.
Then, under ETH, the consistency of~$\Pi$ cannot be decided in~$2^{o(k\cdot\log(\iota))}\cdot\poly(\Card{\at(\Pi)})$.
\end{theorem}

%
%

\futuresketch{
\subsection{Why Epistemic Normal \ASP is Harder than Disjunctive \ASP}

Having the basic concept of dynamic programming in mind,
we use this idea to design a reduction of a HCF program~$\Pi$ to a \SAT formula~$F$,
which is treewidth-aware.
The reduction is inspired by ideas of a DP algorithm for consistency of \HCF programs~\cite{FichteHecher19} and the idea of level mappings~\cite{Janhunen06}.
Intuitively, \emph{global} orderings can cause already huge blowup in the treewidth,
e.g., reductions, where all atoms are ordered at once, often cause long rules with more than treewidth many atoms.
As a result, we apply these numbers only locally within the bags of a TD.
More concretely, our reduction is \emph{guided} by a TD~$\mathcal{T}=(T,\chi)$ of primal graph~$\mathcal{G}_\Pi$
and uses core ideas of dynamic programming along TD~$\mathcal{T}$
to ensure only a slight increase in treewidth of the resulting \SAT formula.
Intuitively, thereby the aforementioned reduction takes care to keep the increase of width
local, i.e., the increase of width happens \emph{within} the bags of~$\mathcal{T}$.
Concretely, if~$\width(\mathcal{T})$ is bounded by some value~$\mathcal{O}(k)$,
the treewidth of the resulting formula~$F$ is at most~$\mathcal{O}(k\cdot\log(k))$.



For encoding orderings along a TD, we need the following notation.
Given a TD~$\mathcal{T}=(T,\chi)$ of~$\mathcal{G}_\Pi$, and a node~$t$ of~$T$.
We refer to an ordering over~$\chi(t)$ by \emph{$t$-local ordering}. 
\begin{definition}
A \emph{$\mathcal{T}$-local ordering} is a set containing one~$t$-local ordering~$\varphi_t$ for every~$t$ of~$T$
such that there is an interpretation~$I$ with 
(1) \emph{satisfiability}: $I\models\Pi_t$ for every node~$t$ of~$T$, 
(2) \emph{provability}: for every $a\in I$, there is a node~$t$ of~$T$ and a rule~$r\in\Pi_{t}$ proving~$a$, and
(3) \emph{compatibility}: for every nodes~$t,t'$ of~$T$ and every~$a,b\in\chi(t)\cap\chi(t')$, 
whenever $\varphi_t(a) < \varphi_t(b)$ then~$\varphi_{t'}(a) < \varphi_{t'}(b)$. 
%
%
\end{definition}
For an ordering~$\varphi$, we use the \emph{canonical $t$-local ordering~$\hat\varphi_t$} for each~$t$ of~$T$ as follows. 
Intuitively, atoms~$a\in\chi(t)$ with smallest ordering position $\varphi(a)$
among all atoms in~$\chi(t)$ get~$\hat\varphi_t(a)=0$, second-smallest get value~$1$, and so on.
%
%
Formally, we define $\hat\varphi_t(a) \eqdef 
\rank_{t}(a,\varphi) - 1 \text{ for each } a\in \chi(t)$, 
where~$\rank_t(a,\varphi)$ is the ordinal number (rank) of $a$ according to smallest ordering position~$\varphi(a)$ among~$\chi(t)$.

\begin{example}\label{ex:localorderings}
Consider program~$\Pi$,
answer set $I=\{b, c, d\}$, and ordering
$\varphi =\{b\mapsto 0, d\mapsto 1, c\mapsto 2\}$
 of Example~\ref{ex:running}.
Ordering~$\varphi$ can easily be extended to ordering
$\varphi' \eqdef\{a\mapsto 0, e\mapsto 0,b\mapsto 0, d\mapsto 1, c\mapsto 2\}$
over~$\at(\Pi)$. 
Then, using TD~$\mathcal{T}$ of~${\mathcal{G}}_{\Pi}$,
we can construct~$\mathcal{T}$-local ordering $\mathcal{M}\eqdef\{\hat\varphi_{t_1}, \hat\varphi_{t_2}, \hat\varphi_{t_3}\}$ of~$\varphi'$, where $\hat\varphi_{t_1}= \{e\mapsto 0,d\mapsto 1, c\mapsto 2\}$, $\hat\varphi_{t_2}= \{a\mapsto 0, b\mapsto 0\}$, and $\hat\varphi_{t_3}= \{e\mapsto 0,b\mapsto 0, d\mapsto 1\}$.
Consider a TD~$\mathcal{T}'$ of~$\mathcal{G}_{\Pi}$,
which is similar to~$\mathcal{T}$, but $t_1$ has a child node~$t'$, whose bag is $\{c,e\}$. Then, $\mathcal{M}\cup\{\hat\varphi_{t'}\}$ with $\hat\varphi_{t'}=\{e\mapsto 0,c\mapsto 1\}$ is a $\mathcal{T}'$-local ordering.
\end{example}

In our reduction, we use the following propositional variables.
For each atom~$x\in\at(\Pi)$, we use~$x$ also as propositional variable.
For each atom~$x\in\chi(t)$ of each node~$t$ of~$T$, we use~$\ceil{\log(\Card{\chi(t)})}$
many variables of the form~$b^i_{x_t}$ forming the $i$-th bit of the $t$-local ordering position (in binary) of~$x$.
By the shortcut notation~$\bvali{x}{t}{j}$, we refer to the \emph{conjunction of literals over bits~$b^i_{x_t}$} for~$1\leq i\leq \ceil{\log(\Card{\chi(t)})}$ according to the representation of the number~$j$ in binary.
%
For atoms~$x,x'\in\chi(t)$ of node~$t$ of~$T$, we use
the following notation to indicate that atom~$x$ is ordered before atom~$x'$:
\vspace{-.25em}
$$x \prec_t x' \eqdef \hspace{-1.6em}\bigvee_{1\leq i \leq \ceil{\log(\Card{\chi(t)})}}\hspace{-1em}(b_{x'_t}^i \wedge \neg b_{x_t}^i \wedge \bigwedge_{i < j \leq \ceil{\log(\Card{\chi(t)})}}\hspace{-1.6em} (b_{x_t}^j \longrightarrow b_{x'_t}^j)).$$

\begin{example}
Consider Example~\ref{ex:localorderings}
and the~$\mathcal{T}$-local ordering~$\mathcal{M}=\{\varphi_{t_1}, \varphi_{t_2}, \varphi_{t_3}\}$.
One could encode ordering position~$\varphi_{t_1}(e)=0$ using two bit variables $b^1_{e_{t_1}},b^2_{e_{t_1}}$ and forcing it to false.
This results in formula~$\bvali{e}{t_1}{0}=\neg b^1_{e_{t_1}} \wedge \neg b^0_{e_{t_1}}$.
Then, we formulate~$\varphi_{t_1}(d)=1$ by~$\bvali{d}{t_1}{1}=\neg b^1_{d_{t_1}} \wedge b^0_{d_{t_1}}$, and~$\varphi_{t_1}(c)=2$ by~$\bvali{c}{t_1}{2}=b^1_{c_{t_1}} \wedge \neg b^0_{c_{t_1}}$. 
For the whole resulting formula, $e\prec_{t_1}d$, $d\prec_{t_1}c$ as well as $e\prec_{t_1}c$ hold.
%
\end{example}


\subsubsection{Reduction for Consistency guided by a TD.}

\noindent 
For solving consistency, we require to construct the following Formulas~(\ref{red:checkrules})--(\ref{red:checkfirst}) below \emph{for each TD node~$t$} of~$T$
having child nodes~$\children(t)=\{t_1, \ldots, t_\ell\}$.
Thereby, these formulas aim at constructing $\mathcal{T}$-local orderings along the TD~$\mathcal{T}$, where Formulas~(\ref{red:checkrules}) ensure satisfiability,
Formulas~(\ref{red:prop}) take care of compatibility along the TD,
and Formulas~(\ref{red:checkfirst}) enforce provability within a node, which is then guided along the TD by Formulas~(\ref{red:checkremove}) to~(\ref{red:check}).

\noindent Concretely, Formulas~(\ref{red:checkrules}) ensure that
the variables of the constructed \SAT formula~$F$ are such that
all (bag) rules are satisfied.
Then, whenever in node~$t$ an atom~$x$ has a smaller ordering position than an atom~$x'$ (using $\prec_{t}$), 
this must hold also for the parent node of~$t$ and vice versa,
cf., Formulas~(\ref{red:prop}).
Formulas~(\ref{red:checkremove}) guarantee, for nodes~$t$
removing bag atom~$x$, i.e., $x\in\chi(t)\setminus\chi(t')$, 
that~$x$ is proven if~$x$ is set to true.
Similarly, this is required for atoms~$x\in\chi(n)$ that are in the root node~$n=\rootOf(T)$
and therefore never forgotten, cf., Formulas~(\ref{red:checkremove2}).
At the same time we ensure by Formulas~(\ref{red:check}) that an atom~$x$ is proven up to node~$t$
if and only if it is proven up to some child node of~$t$ or freshly proven in node~$t$.
Finally, Formulas~(\ref{red:checkfirst}) take care that an atom~$x$ is freshly proven in node~$t$
if and only if there is at least one rule~$r\in \Pi_t$ proving~$x$.
%
%
%
%
%

First, we show the lower bound result for inner treewidth.
Then, since treewidth is larger or equal to the inner treewidth,
the corresponding lower bound result for treewidth follows as a corollary.

$\exists V_1. \forall V_2. D$

$\exists V_1. \neg \exists V_2. \neg C$

\begin{align}
	\label{red:epc}  &\leftarrow \text{not }red\\
	\label{red:epswitch}  &red \leftarrow \neg \bar{red}\\
	\label{red:epswitch2}  &\bar{red} \leftarrow \neg red\\
	\label{red:epistemic} &a\leftarrow \text{not } \bar{a} &\text{for every }a\in V_1\\
	\label{red:epistemic2} &\bar{a}\leftarrow \text{not } {a} &\text{for every }a\in V_1\\
	\label{red:check2}&\leftarrow \neg red,q_{t}^x, a &\text{for every }t\text{ in }T, \alpha\in\{\varphi_t,\psi_t\}, x\in\dom(\alpha) \text{ has no }\prec_{\alpha}\text{successor}, a\in V_1, [\mathcal{I}_t(\alpha)](a)=0\\ 
	\label{red:check2b}&\leftarrow \neg red,q_{t}^x, \neg a &\text{for every }t\text{ in }T, \alpha\in\{\varphi_t,\psi_t\}, x\in\dom(\alpha) \text{ has no }\prec_{\alpha}\text{successor}, a\in V_1, [\mathcal{I}_t(\alpha)](a)=1
\end{align}
}

\section{Discussion and Conclusion}
The complexity of \ASP has already been studied in different facets and flavors. Recently, it has been shown that under the exponential time hypothesis (ETH), the evaluation of normal logic programs is expected to be \emph{slightly harder} for the structural parameter treewidth, than deciding satisfiability (\SAT) of Boolean formulas.
However, the hardness proof relies on \emph{large cycles (SCCs)}, unbounded in the treewidth. Further, compared to standard reductions, see Example~\ref{ex:encoding}, 
utilizing the ``hardness'' of normal \ASP remained unclear.

In this paper, we address both shortcomings. 
%
%
The idea is to reduce from \SAT to normal \ASP, thereby actively decreasing structural dependency in the form of treewidth. 
We design such a reduction that \emph{reduces treewidth} from~$k$ to~$\frac{k}{\log(k)}$. We show that under ETH, this decrease cannot be significantly improved.
Even further, with the help of the reduction, the existing hardness result for normal programs and treewidth can be improved: 
The constructed cycles (SCCs) are not required to be unbounded in the treewidth; indeed, hardness is preserved in case of a \emph{single-exponential bound in the treewidth}.
Then, we further improve this bound for the class of~$\iota$-tight programs, which allows us to \emph{close the gap} to the known upper bound, which our results render ETH-tight.

Finally, we 
apply our reduction for 
establishing further ETH-tight lower bounds on normal logic programs.
We hope that the reduction of this work enables further consequences and insights on hardness for problems on normal logic programs. 
In the light of a known result~\cite{AtseriasFichteThurley11} on the correspondence of treewidth and resolution width applied in \SAT solving, this might pave the way towards such insights for \ASP.
%
Currently, we are working on the comparison of different reductions from \SAT to \ASP and how they perform in practice.
Given the unsuccessful application of directed measures for \ASP~\cite{BliemWoltranOrdyniak16}, structural parameters between treewidth and directed variants could lead to new insights. 

\futuresketch{
The work in this paper gives rise to plenty of future work.
On the one hand, we are currently working on the comparison of different treewidth-aware reductions to \SAT and variants thereof, and how these variants perform in practice.
Further, we are curious about treewidth-aware reductions to \SAT, which preserve answer sets bijectively or are modular~\cite{Janhunen06}.

Also investigating further structural parameters ``between'' treewidth and directed variants of treewidth could lead to new insights,
since for \ASP directed measures~\cite{BliemWoltranOrdyniak16} often do not yield efficient algorithms.

The curiosity of studying and determining the hardness of \ASP and the underlying reasons has attracted the attention of the KR community for a long time.
This paper discusses this question from a different angle,
which hopefully will provide new insights into the hardness of \ASP and foster follow-up work.
The results in this paper indicate that, at least from a structural point of view, deciding the consistency of \ASP is already harder than 
\SAT, since \ASP programs might compactly represent structural dependencies within the formalism.
More concretely, compiling the hidden structural dependencies of a program to a \SAT formula, measured in terms of the well-studied parameter treewidth, most certainly causes a blow-up of the treewidth of the resulting formula.
In the light of a known result~\cite{AtseriasFichteThurley11} on the correspondence of treewidth and the resolution width applied in \SAT solving, this reveals that \ASP consistency might be indeed harder than solving \SAT.
We further presented a reduction from \ASP to \SAT that is aware of the treewidth in the sense that the reduction causes not more than this inevitable blow-up of the treewidth in the worst-case.
Currently, we are working on the comparison of different reductions from \SAT to \ASP and how they perform in practice.
The study of further structural parameters ``between'' treewidth and directed variants of treewidth could lead to new insights,
since for \ASP directed measures~\cite{BliemWoltranOrdyniak16} often do not yield efficient algorithms.
}



\clearpage

\subsubsection{Acknowledgments.}
This work has been supported by the Austrian Science Fund (FWF),
Grants J 4656 and P 32830, as well as the Vienna Science and
Technology Fund, Grant WWTF ICT19-065.

\bibliography{references}

\appendix
\clearpage

\section{Proof of Correctness}

Let~$t$ be a node of~$T$. Then, we define~$Q(t)\eqdef \{q_t^x \mid x\in V_t\}$ to be the complete \emph{set of query markers (for~$t$)}. 

\begin{lemma}[Decidability]\label{lem:compat1}
The reduction from a Boolean formula~$F$ and a TD~$\mathcal{T}=(T,\chi)$ of~$\mathcal{G}_\Pi$ to a logic program~$\Pi$, consisting of Formulas~(\ref{red:start}) to~(\ref{red:checkunused}), ensures \emph{decidability}:
For every answer set~$M$ of~$\Pi$ and every node~$t$ of~$T$, there is at least one node~$t'$ of~$T$ with~$V_{t'}=V_t$ and $Q(t')\subseteq M$. 
\end{lemma}
\begin{proof}[Proof (Sketch)]
Assume towards a contradiction that there exists such an answer set~$M$ of~$\Pi$ and a node~$t$ of~$T$
such that there is no node~$t'$ of~$T$ with~$V_{t'}=V_t$ and~$Q(t')\subseteq M$.

However, then, by Formulas~(\ref{red:choose}), for every~$y\in V_t$ (except destination vertex~$d_{\rootOf(T)}$) there has to exist some~$e_{y,x}\in M$.
Formulas~(\ref{red:fintrans})--(\ref{red:fin}) ensure that (1) every~$y\in V_t$ is allowed only
one outgoing edge.
Assume towards a contradiction that for~$y\in V_t$ there are two edges~$e_{y,x}, e_{y,x'}\in M$.
Then, without loss of generality we encounter~$e_{y,x}\in E_{t''}$ in a node~$t''$.
By Formulas~(\ref{red:finnew}), we require~$o_{t''}^y$. Then, by Formulas~(\ref{red:fintrans}),
we require~$o_{t_i}^y\in M$ for every~$t_i$ that is between~$t''$ and the node, say~$t'''$, where
we encounter~$e_{y,x'}^y\in E_{t'''}$. Observe that by the connectedness of tree decompositions, 
such a path from~$t''$ to~$t'''$ has to exist. Finally, we arrive at a contradiction,
since Formulas~(\ref{red:fin}) for~$t'''$ fail due to~$o_{t_i}^y\in M$.
Consequently, we conclude (2) every~$y\in V_t$ (except destination vertex~$d_{\rootOf(T)}$) is admitted at most one outgoing edge.
Combining (1) with (2) and by Formulas~(\ref{red:reached}), we derive that every~$y\in V_t$ (except~$d_{\rootOf(T)}$)
there is precisely one outgoing edge required in~$M$.

So by construction of~$\Pi$, there is a vertex among~$V_t$ that is reached first, which is then used to reach the second, and so on. 
We can arrange these vertices along a path in the order of their reachability, say~$P=v_1, \ldots, v_{\Card{V_t}}$.
However, by construction of~$\Pi$, there has to exist a node~$t^*$, where~$\varphi_{t^*}$ precisely reflects the ordering of~$P$.
Observe that~$p_{\prev(t^*,v_1)}^{v_1}\in M$, which is either derived by Formula~(\ref{red:reach}) if $\prev(t^*,v_1)=\epsilon$ or by Formulas~(\ref{red:copy}) if
for the child node of~$t^*$, we have~$\dom(\varphi_{\prev(t^*,v_1)})=\emptyset$. If $\dom(\varphi_{\prev(t^*,v_1)})\neq\emptyset$,
we have that $p_{\prev(t^*,v_1)}^{v_1}\in M$ is either by Formulas~(\ref{red:linkpq}) or via Formulas~(\ref{red:linkbypass}) or (\ref{red:linkbypass2}). 
In any case, we can derive~$Q(t^*)\subseteq M$, which are derived precisely in the order of~$P$ by using Formulas~(\ref{red:linknopred}), followed by several applications of Formulas~(\ref{red:linkprevious}).
This contradicts that~$t'$ does not exist, since~$t'=t^*$.
\end{proof}

\begin{lemma}[Compatibility]\label{lem:compat}
The reduction from a Boolean formula~$F$ and a TD~$\mathcal{T}=(T,\chi)$ of~$\mathcal{G}_\Pi$ to a logic program~$\Pi$, consisting of Formulas~(\ref{red:start}) to~(\ref{red:checkunused}) ensures \emph{compatibility}:
Let~$M$ be an answer set of~$\Pi$. Then, for every two nodes~$t,t'$ of~$T$ with~$Q(t),Q(t')\subseteq M$,
we have that $\mathcal{I}_t(\varphi_t)$ and $\mathcal{I}_{t'}(\varphi_{t'})$ are compatible.
%
%
\end{lemma}
\begin{proof}
Case $V_t\neq V_{t'}$:
Assume towards a contradiction that there exists such an answer set~$M$ of~$\Pi$ and two nodes~$t,t'$ of~$T$
with~$Q(t),Q(t')\subseteq M$, where
$\mathcal{I}_t(\varphi_t)$ and $\mathcal{I}_{t'}(\varphi_{t'})$ are incompatible.
Then, however, by definition of incompatibility, both assignments must share at least one common variable~$x$ with
$[\mathcal{I}_t(\varphi_t)](x)\neq[\mathcal{I}_{t'}(\varphi_{t'})](x)$.
Consequently, by construction of the nice TD~$\mathcal{T}$, there has to exist a node
~$t^*$ in~$T$ with~$\varphi_{t^*}=\varphi_t$ and~$\psi_{t^*}=\varphi_{t'}$ or vice versa ($\varphi_{t^*}=\varphi_{t'}$ and~$\psi_{t^*}=\varphi_{t}$). But then, we have that Formulas~(\ref{red:probecons}) do not hold for~$t^*$, which contradicts that~$M$ is an answer set of~$\Pi$.
%
%

Case $V_t=V_{t'}$:
Assume towards a contradiction that there exists such an answer set~$M$ of~$\Pi$ with an ordering~$\alpha$ that proves~$M$ and two nodes~$t,t'$ of~$T$
with~$Q(t),Q(t')\subseteq M$, where
$\mathcal{I}_t(\varphi_t)$ and $\mathcal{I}_{t'}(\varphi_{t'})$ are incompatible.
Then, there exist at least two vertices~$x,y\in V_t$ with~$\varphi_t(x) > \varphi_t(y)$, but
$\varphi_{t'}(y) > \varphi_{t'}(x)$, since otherwise~$\varphi_t=\varphi_{t'}$. 
%
By construction of~$\mathcal{R}$, in particular, Formulas~(\ref{red:linkprevious}), 
there is a chain of rule dependencies from~$q_{t'}^y$ to $q_{t'}^x$, requiring that~$q_{t'}^y$ is used to derive~$q_{t'}^x$, i.e., $\alpha(q_{t'}^y) < \alpha(q_{t'}^x)$.
On the other hand, for~$t$ we have a chain of rule dependencies from~$q_{t}^x$ to~$q_{t}^y$, i.e., $\alpha(q_{t}^x) < \alpha(q_{t}^y)$. 
%
%

Subcase $r_y$ is used to derive $r_x$, i.e., $\alpha(r_y) < \alpha(r_x)$: This subcase is is similar to the one in Figure~\ref{fig:sketch} with the red edges. So, the chain of dependencies indirectly derives~$r_y$ before derive~$p_\epsilon^x$, which is required for deriving~$r_x$, by Formulas~(\ref{red:linkplast}) and~(\ref{red:linkroot}). 
So we have that~$\alpha(r_y) < \alpha(p_\epsilon^x) < \alpha(r_x)$. Then, by Formulas~(\ref{red:linknopred}) and~(\ref{red:linkprevious}), we require that~$\alpha(p_\epsilon^x) < \alpha(q_t^x)$.
Further, by construction, e.g., Formulas~(\ref{red:linkpq}) and~(\ref{red:copy})--(\ref{red:linkbypass2}) we have that~$\alpha(q_t^y) < \alpha(r_y)$ and by assumption above we have~$\alpha(q_t^x) < \alpha(q_t^y)$. As a result, we arrive at the contradiction that~$\alpha(r_y) < \alpha(r_y)$.

Subcase $r_x$ is used to derive $r_y$, i.e., $\alpha(r_x) < \alpha(r_y)$: Here, the chain of dependencies 
indirectly derives~$r_x$ before derive~$p_\epsilon^y$, which is required for deriving~$r_y$, by Formulas~(\ref{red:linkplast}) and~(\ref{red:linkroot}). 
So we have that~$\alpha(r_x) < \alpha(p_\epsilon^y) < \alpha(r_y)$. Then, by Formulas~(\ref{red:linknopred}) and~(\ref{red:linkprevious}), we require~$\alpha(p_\epsilon^y) < \alpha(q_{t'}^y)$.
By construction, e.g., Formulas~(\ref{red:linkpq}) and~(\ref{red:copy})--(\ref{red:linkbypass2}) we have that~$\alpha(q_{t'}^x) < \alpha(r_x)$ and by assumption above we have~$\alpha(q_{t'}^y) < \alpha(q_{t'}^x)$. As a result, we arrive at the contradiction that~$\alpha(r_x) < \alpha(r_x)$.
%
%
\end{proof}

\begin{restatetheorem}[thm:corr]
\begin{theorem}[Correctness]
The reduction from a Boolean formula~$F$ and a TD~$\mathcal{T}=(T,\chi)$ of~$\mathcal{G}_\Pi$ to a logic program~$\Pi$, consisting of Formulas~(\ref{red:start}) to~(\ref{red:checkunused}), is correct.
Concretely, for every model of~$F$, there is at least one answer set of~$\Pi$. Vice versa, for every answer set of~$\Pi$, there is a unique model of~$F$. 
\end{theorem}
\end{restatetheorem}
\begin{proof}
``$\Rightarrow$'':
Let~$I: \var(F) \rightarrow \{0,1\}$ be a model of~$F$, i.e., $I\models F$.
From this we construct an answer set~$A$ of~$\Pi$ as follows.
For every~$t$ in~$T$, we let~$\alpha_t\eqdef \mathcal{I}_t^{-1}(I_{|V_t})$ be the corresponding
ordering of the assignment~$I$ restricted to bag~$V_t$.
We construct~$A'$ as follows:
\begin{itemize}%
	\item[(i)] for every~$t$ in~$T$ with~$\alpha_t=\langle a_1, \ldots, a_\ell \rangle$ and~$1\leq i\leq \ell$, we require 
	$p_\epsilon^{a_i}\in A'$ as well as $r_{a_i}\in A'$, and additionally~$r_{s_t}\in A'$ if $t$ is a leaf node
	\item[(ii)] for every~$t$ in~$T$ with~$\alpha_t=\langle a_1, \ldots, a_\ell \rangle$, we require 
	$e_{a_i, a_{i+1}}\in A'$ for every~$1\leq i<\ell$,
	\item[(iii)] for every~$t$ in~$T$ with~$\alpha_t=\langle a_1, \ldots, a_\ell \rangle$, $t'\in\children(t)$, and $\alpha_{t'}=\langle a'_1, \ldots, a'_{\ell'}\rangle$ s.t.\ $\alpha_t\neq\alpha_{t'}$, we require $e_{a'_{\ell'}, a_1}\in A'$,
%
%
	\item[(iv)] for every~$t$ in~$T$, $\alpha\in \{\varphi_t, \psi_t\}$, and $\alpha=\langle a_1, \ldots, a_\ell\rangle$ we need to have~$p_t^{a_i}\in A'$ for every~$1\leq i\leq \ell$. Further, for the largest sequence~$\beta=\langle a_1, \ldots, a_{\ell'} \rangle$ of~$\alpha$ ($\ell'\leq \ell$) that is compatible with~$\alpha_t$, we also require~$q_t^{a_j}\in A'$ for every~$1\leq j \leq \ell'$.
\end{itemize}
Then, we define~$A''$ based on~$A'$, where 
%
(v) for every~$t$ in~$T$ with~$x,y\in V_t$ and~$e_{x,y}\in A'$, we need to have~$o_t^{y}\in A''$. 
	Further, for every parent and ancestor~$t^*$ of~$t$ in~$T$ with~$y\in V_{t^*}$, we require~$o_{t^*}^y\in A''$.
%

Then, we define~$A$ based on $A'$ and~$A''$, where (vi) $A\eqdef A' \cup A'' \cup \{\hat e_{x,y} \mid e_{x,y}\in (\bigcup_{t\text{ in }T} E_t \setminus A')\}$. 

We show that~$A$ is indeed an answer set of~$\Pi$.
First, we establish in Part (1) that~$A$ indeed satisfies every rule in~$\Pi$
and then, we show in Part (2) that~$A$ is indeed a subset-minimal model of~$\Pi$.
Part (1): By (i), $A$ satisfies both Formulas~(\ref{red:start}), (\ref{red:reached}), and~(\ref{red:reach}).
By (ii), (iii), and (vi), we have that~$A$ satisfies Formulas~(\ref{red:choose}).
By (v), $A$ satisfies Formulas~(\ref{red:fintrans}), (\ref{red:finnew}), and~(\ref{red:fin}).
Due to (iv), $A$ satisfies Formulas~(\ref{red:probecons}) and (\ref{red:check}), and~(\ref{red:checkunused}) as well.
By (iv), $A$ satisfies those formulas with variables of the form~$p_t^x$ in the head, which are Formulas~(\ref{red:copy}), (\ref{red:linkpq}), (\ref{red:linkbypass}), and (\ref{red:linkbypass2}). Further, by~(i), $A$ satisfies Formulas~(\ref{red:linkplast}), (\ref{red:linkroot}).
Finally, by (iv), we have that $A$ satisfies Formulas~(\ref{red:linknopred}) and~(\ref{red:linkprevious}).

Part (2):
Assume towards a contradiction that there is a set~$B\subsetneq A$ that satisfies Formulas~(\ref{red:start})--(\ref{red:check}). 
Observe that the difference between~$B$ and~$A$ cannot be due to (ii), (iii), or (vi), since otherwise
$B$ would not satisfy Formulas~(\ref{red:choose}). 
Further, Definition (i) is required for~$B$ to be an answer set, due to Formulas (\ref{red:start}) 
as well as Formulas~(\ref{red:reach}), 
which can be only satisfied for every~$y\in V_t$ of every~$t$ in~$T$, by
requiring~$p_\epsilon^y\in B$, see Formulas~(\ref{red:linkplast}), (\ref{red:linkroot}).
Then, Definition (v) is required in order to satisfy Formulas~(\ref{red:fintrans}) and~(\ref{red:finnew}).
If~$B$ does not set at least some~$p_t^x$ to true, then $B$ does not satisfy some of the Formulas~(\ref{red:copy}), (\ref{red:linkpq}),
(\ref{red:linkbypass}), or~(\ref{red:linkbypass2}). If, on the other hand,~$B$ does not set some~$q_t^x$ to true, then either Formulas~(\ref{red:linknopred}) or~(\ref{red:linkprevious}) are not satisfied by~$B$.

``$\Leftarrow$'': 
%
%
Let~$M$ be an answer set of~$\Pi$. 
Then, by Lemma~\ref{lem:compat1}, for every node~$t$ in~$T$, there ist at least one node~$t'$ of~$T$, with~$V_{t'}=V_t$ and $Q(t')\subseteq M$. 
From this, we construct the following assignment~$I: \var(F) \rightarrow \{0,1\}$, where for every~$x\in \var(F)$, we define~$I(x)\eqdef v$ whenever there is one $t$ in $T$ with $Q(t)\subseteq M$ such that~$\mathcal{I}_t(\varphi_t)(x)=v$.
Observe that by Lemma~\ref{lem:compat1} every variable~$x\in\var(F)$ is indeed addressed and finally by compatibility of Lemma~\ref{lem:compat},
$I$ is well-defined.
By construction of~$I$ and due to Formulas~(\ref{red:check}) and~(\ref{red:checkunused}), 
we have~$I\models F$.
To be more precise, assume towards a contradiction that~$I\not\models F$.
Then, there is at least one clause~$c\in F$ with~$I\not\models \{c\}$.
By the properties of a TD, there has to exist at least one node~$t$ of~$T$ with~$c\in F_t$.
By Lemma~\ref{lem:compat1} for every node~$t$, there ist at least one node~$t'$ of~$T$, with~$V_{t'}=V_t$ and $Q(t')\subseteq M$. Then, 
by Formulas~(\ref{red:check}) (and~(\ref{red:checkunused})) 
we have that~$\mathcal{I}_{t'}(\alpha_{t'})\models F_{t'}$.
Therefore, by construction of~$I$ and by compatibility of Lemma~\ref{lem:compat}, 
we have that~$I\models F_{t'}$. 
This contradicts the assumption that $I\not\models \{c\}$.

Observe that~$I$ is uniquely defined for every different answer set~$M$.
Assume towards a contradiction that there were two different answer sets~$M,M'$ of~$\Pi$,
with the same~$I$ as defined and constructed above.
But then, if both~$M$ and~$M'$ differ by variables of the form~$e_{y,x}$,
it can be shown that
there is a node~$t$ of~$T$ with~$Q(t)\subseteq M$, but~$Q(t)\not\subseteq M'$.
In the proof of Lemma~\ref{lem:compat1}, it is shown how~$Q(t)\in M$ precisely depends on these edges~$e_{y,x}$,
reflecting the ordering $y$ before~$x$.
%
%
As a result, $M$ and~$M'$ can not differ by variables of the form~$e_{y,x}$.
Consequently, since the remaining formulas are just propagating information
(no further choice, except Formulas~(\ref{red:choose})), we have~$M=M'$,
which contradicts our assumption.
\end{proof}%

\section{Additional Proofs}

\begin{lemma}[Largest SCC Size]\label{cor:cycle}
Let~$F$ be a Boolean formula, and~$\mathcal{T}=(T,\chi,\varphi,\psi)$ be an ordering-augmented TD
of~$\mathcal{G}_F$ of width $k$.
Then, the reduction~$\mathcal{R}$ constructs a program
with
strongly connected components (SCCs) of~$D_\Pi$ of size at most
~$2^{\mathcal{O}(k'\cdot\log(k'))}$ with~$k'=\frac{k}{\log(k)}$,
which equals~$2^{\mathcal{O}(k)}$. 
\end{lemma}
\begin{proof}
Observe that the largest SCC of~$D_\Pi$ can only be
due to
Formulas~(\ref{red:choose})--(\ref{red:linkroot}) and (\ref{red:copy})--(\ref{red:linkpq}),
as these form the only potentially cyclic rules. 
%
The SCC is constructed over variables of the form~$r_x$, $e_{x,y}$, $p_\epsilon^x$, $p_{t}^x$, $q_t^x$ with~$x,y\in V_{t^*}$ for some~$t,t^*$ of~$T$ such that~$V_t=V_{t^*}$.
Thereby, in the SCC we have at most~$k'$ many variables of the form~$r_x$ and~$p_\epsilon^x$ and at most~$k'^2$ many variables of the form~$e_{x,y}$.
However, there are at most~$(k'!)^2$ many variables of the form~$p_t^x$ and~$q_t^x$
due to the fact that any potential combination of two orderings~$\varphi_t,\psi_t$ is analyzed by Formulas~(\ref{red:linkbypass})--(\ref{red:linkpq}).
While the variables constructed for the orderings~$\varphi_t$ and~$\psi_t$ are by construction over different SCCs, cf., Definition~\ref{def:ord},
the combinations cause the investigation of duplicate orderings over variables of the form~$p_t^x, q_t^x$, which still cannot exceed the number~$(k'!)^2$ of combinations. 
Obviously, $(k'!)^2$ dominates the SCC size, which is bounded by~${2^{{\mathcal{O}(k'\cdot\log(k'))}}}^2 = 2^{\mathcal{O}(2\cdot k'\cdot\log(k'))} = 2^{\mathcal{O}(k'\cdot\log(k'))}$. 
Then, $2^{\mathcal{O}(k'\cdot\log(k'))} = 2^{\mathcal{O}(\frac{k}{\log(k)}\cdot\log(k'))} = 2^{\mathcal{O}(\frac{k}{\log(k)}\cdot(\log(k)-\log(\log(k)))} = 2^{\mathcal{O}(k)}$ establishes the claim. 
%
%
\end{proof}

\begin{restatecorollary}[cor:hcycle]
\begin{corollary}[LB Largest SCC]
Let~$\Pi$ be a normal logic program, where the treewidth of~$\mathcal{G}_\Pi$ is~$k$
such that the largest SCC size of~$D_\Pi$ is in~$2^{\mathcal{O}(k\cdot\log(k))}$.
Then, unless ETH fails, it is not expected
that the consistency of~$\Pi$ can be decided in time~$2^{o(k\cdot\log(k))}\cdot\poly(\Card{\at(\Pi)})$.
\end{corollary}
\end{restatecorollary}
\begin{proof}
This result follows from the correctness of reduction~$\mathcal{R}$, see Theorem~\ref{thm:corr}, 
the treewidth-awareness by Theorem~\ref{thm:tw}, and the largest SCC size bound by Lemma~\ref{cor:cycle}.
%
Assume towards a contradiction that
one can solve~$\Pi$ in time~$2^{o(k\cdot\log(k))}\cdot\Card{\poly(\at(\Pi))}$.
 %
%
%
%
%
Then, we take an arbitrary Boolean formula~$F$ and an ordering-augmented TD~$\mathcal{T}$ of~$\mathcal{G}_F$ of width~$k'$ and construct~$\Pi$ by means of~$\mathcal{R}$.
By Theorem~\ref{thm:corr}, the reduction is correct, i.e., $\Pi$ admits an answer set if and only if there is a satisfying assignment of~$F$. Further, the largest SCC of~$D_\Pi$ is bounded by~$2^{\mathcal{O}(k\cdot\log(k))}$, according to Lemma~\ref{cor:cycle}. Then, the treewidth~$k$ of~$\Pi$ is bounded by~$\mathcal{O}(\frac{k'}{\log(k')})$ by Theorem~\ref{thm:tw}.
So assuming that~$\Pi$ can be decided in time~$2^{o(k\cdot\log(k))}\cdot\poly(\Card{\at(\Pi)})$,
results in solving~$F$ in time~$2^{o(\frac{k'}{\log(k')}\cdot\log(\frac{k'}{\log(k')}))}\cdot\poly(\Card{\at(F)}) = 2^{o(k')}\cdot \poly(\Card{\var(F)})$, which contradicts the ETH.
\end{proof}

The following result for QBFs is known, where
it turns out that deciding $\QBFSAT$ remains $\ell$-fold exponential in the treewidth
of the primal graph (even when restricting the graph 
to the variables of the inner-most quantifier block).

\begin{proposition}[\cite{FichteHecherPfandler20}]\label{prop:lb}
Given a QBF~$Q=\forall V_1. \exists V_2, \ldots, \exists V_\ell. F$ of quantifier depth~$\ell$, whose treewidth of~$\mathcal{G}_F$ is $k+1$ such that  
$k$ is the treewidth of~$\mathcal{G}_F$ restricted to vertices in~$V_\ell$.
%
Then, under ETH, the validity of~$Q$ cannot be decided
in time~$\tower(\ell, o(k))\cdot\poly(\var(Q))$. 
\end{proposition}
%

This proposition will be used as the basis for deriving the following result.

\begin{restatetheorem}[cor:pcount]
\begin{theorem}[LB Projected Counting]
Let~$\Pi$ be a normal logic program, $P{\subseteq}\at(\Pi)$ be atoms, and~$k$ be the treewidth of~$\mathcal{G}_\Pi$, such that the largest SCC size of~$D_\Pi$ is in~$2^{\mathcal{O}(k\cdot\log(k))}$.
Then, under ETH, the cardinality~$\Card{\{M\cap P \mid M\models \Pi\}}$ cannot be computed in time~$2^{2^{o(k\cdot\log(k))}}\cdot\poly(\Card{\at(\Pi)})$.
\end{theorem}
\end{restatetheorem}
\begin{proof}
The idea is to apply the result of Proposition~\ref{prop:lb},
where we reduce from~$Q=\forall V_1. \exists V_2. F$, whose treewidth of~$\mathcal{G}_F$ is~$k'$.
In the proof of Proposition~\ref{prop:lb},
the authors actually show an even stronger result for a restricted QBF, such that there exists a nice TD~$\mathcal{T}=(T,\chi)$ of~$\mathcal{G}_F$, where for every variable~$v\in V_1\cup \ldots \cup V_{\ell-1}$ there is a unique node~$t$ in~$T$ with~$v\in \chi(t)$, i.e., $\{v\} = \chi(t)\cap (V_1 \cup \ldots \cup V_{\ell-1})$. 
Assumption (A): We assume such a constructed nice TD~$\mathcal{T}$ for~$\ell=2$ and turn it into an ordering-augmented TD~$\mathcal{T}'=(T',\chi',\varphi,\psi)$ of~$\mathcal{G}_F$. 

Then, we set~$P\eqdef V_1$ and construct a program~$\Pi'$ from~$F$ by using reduction~$\mathcal{R}$.
Then, for every variable~$v\in V_1$, we construct
the rules~\inlineequation[g1]{v \leftarrow \neg \hat v} and \inlineequation[g2]{\hat v \leftarrow \neg v}
over fresh atoms~$v$ and~$\hat v$. These rules are responsible for guessing truth values over fresh variables. 
Further, we construct~$\Pi$ by adding to~$\Pi'$ the following rules.
First, we add compatibility rules, similar to Formulas~(\ref{red:check}), where we ensure that these truth values are reflected in the reduction, as follows: 
\inlineequation[sync]{\leftarrow v, q_t^x}
for every node~$t$ in~$T'$ with~$x\in\dom(\varphi_t)$ having no~$\prec_{\varphi_t}$ successor such that~$\mathcal{I}_t(\varphi_t)[v] = 0$.
Analogously, we add \inlineequation[sync2]{\leftarrow \hat v, q_t^x}
for every node~$t$ in~$T'$ with~$x\in\dom(\varphi_t)$ having no~$\prec_{\varphi_t}$ successor such that~$\mathcal{I}_t(\varphi_t)[v] = 1$.

Given the correctness proof of Theorem~\ref{thm:corr} and the adaptions above, it is easy to see that~$Q$ evaluates to true, whenever we have that~$\Card{\{M\cap P \mid M\models \Pi\}} = 2^{\Card{V_1}}$.
Even further, since by Assumption (A), every variable~$v\in V_1$ appears uniquely in a node~$t$ of~$T$; by Definition~\ref{def:ord} no other variable in~$V_1$ than~$v$ appears in assignments~$\varphi_{t_i}$ for~$t_i$ on a path~$t_1, \ldots, t_o$ of nodes below~$t$ in~$T'$.
As a result, we can easily modify the TD constructed by Theorem~\ref{thm:tw}, by adding~$v$ and~$\hat v$ to previsely these bags~$\chi(t_1), \ldots, \chi(t_o)$.
Consequently, the treewidth compared to Theorem~\ref{thm:tw} 
increases at most by~$2$, i.e., 
 the treewidth guarantee 
 of~$k$ in $\mathcal{O}(\frac{k'}{\log(k')})$ is preserved. 
%

Now assume towards a contradiction that
$\Card{\{M\cap P \mid M\models \Pi\}}$ can be computed in time~$2^{2^{o(k\cdot\log(k))}}\cdot\poly(\Card{\at(\Pi)})$.
Then, however, we can decide~$Q$ in time $2^{2^{o(\frac{k'}{\log(k')}\cdot(\log(k')-\log(\log(k'))))}}\cdot\poly(\Card{\at(F)})$ = $2^{2^{o(k')}}\cdot\poly(\Card{\at(F)})$, which contradicts the ETH.
%
%
%
%
%
%
\end{proof}

\begin{restateobservation}[obs:logk]
\begin{observation}
Let~$\iota, k\in \Nat$ with $2\leq \iota\leq k$. Then, $(\iota!)^\frac{k}{\iota}$ is in~$2^{\Omega(k\cdot\log(\iota))}$.
\end{observation}
\end{restateobservation}
\begin{proof}
By Observation~\ref{obs:klogk}, $\iota!$ is bounded by ~$2^{\Omega(\iota\cdot\log(\iota))}$. As a result, $(\iota!)^\frac{k}{\iota}$ is then in~$2^{{\Omega(\iota\cdot\log(\iota))}^\frac{k}{\iota}}$, which is in~$2^{{\frac{k}{\iota}\cdot\Omega(\iota\cdot\log(\iota))}}$, yielding the result: $(\iota!)^\frac{k}{\iota}$ is in~$2^{{\Omega(k\cdot\log(\iota))}}$.
\end{proof}

\section{Reduction~$\mathcal{R}'$ for~$\iota$-Tightness}

\begin{figure*}[t]
{
\begin{flalign}
	&\textbf{Block 1: Orderings \& Reachability}\hspace{-10em}\notag\\[-.4em]
	\label{red:start2}&r_{s_{V_t}^j}\leftarrow & \text{for every }t\text{ in }T,  z\in V_t, j=\block(z,t)  \tag{\ref{red:start}'}\\
	\label{red:reached2}&\leftarrow \neg r_y & \text{ for every }t\text{ in }T, y\in V_t \cup \{d_{V_t}^j \mid z\in V_t, j=\block(z,t)\}\tag{\ref{red:reached}'}\\
	%
	\label{red:reachfin2}&{r_{d_{V_t}^j}} \leftarrow e_{y,d_{V_t}^j}  & \text{ for every }t\text{ in }T, e_{y,d_{V_t}^j}\in E_t\tag{\ref{red:reachfin}'}\\
%
%
%
%
	&\textbf{Block 2: ${\leq}1$ Outgoing Edge}\hspace{-10em}\notag\\[-.4em]
	\label{red:fintrans2}&o_t^y\leftarrow o_{t'}^y &\text{ for every }t\text{ in }T, t'\in\children(t), y\in (V_t\cap V_{t'}) \cup \{s_{V_t}^j\mid V_t{=}V_{t'}, z\in V_t, z\in\block(z,t)\}\tag{\ref{red:fintrans}'}\\
	\label{red:fin2}&\leftarrow o_{t'}^y, e_{y,x} &\text{ for every }t\text{ in }T, e_{y,x}\in E_t, t'\in\children(t),y\in V_{t'}\cup \{s_{V_t}^j\mid V_t{=}V_{t'}, z\in V_t, z\in\block(z,t)\}\tag{\ref{red:fin}'}
\end{flalign}
\vspace{-1.75em}
}\caption{The reduction~$\mathcal{R}'$ adapted from~$\mathcal{R}$, thereby replacing~(\ref{red:start}), (\ref{red:reached}), (\ref{red:reachfin}), (\ref{red:fintrans}), and~(\ref{red:fin}) by (\ref{red:start2}), (\ref{red:reached2}), (\ref{red:reachfin2}), (\ref{red:fintrans2}), and~(\ref{red:fin2}), respectively. Reduction~$\mathcal{R}'$ takes a formula~$F$ and an ordering-augmented TD~$\mathcal{T}=(T,\chi,\varphi,\psi)$ of~$\mathcal{G}_F$.}\vspace{-.35em}\label{fig:red2}
\end{figure*}

In the following, we describe the modifications of~$\mathcal{R}$ required for reduction~$\mathcal{R}'$.
Let~$F$ be a Boolean formula and~$\mathcal{T}=(T,\chi,\varphi_t,\psi_t)$ be an ordering-augmented TD of~$\mathcal{G}_F$ of width~$k$.
Recall that we decrease the treewidth from~$k$ to the smallest~$k'$ such that~$(\iota!)^{\frac{k'}{\iota}} \geq 2^k$ for fixed~$2\leq \iota\leq k'$.
Then, for every node~$t$ of~$T$, $k'_t$ is the \emph{smallest integer} with~$(\iota)!^{\frac{k'_t}{\iota}} \geq 2^{\Card{\chi(t)}}$.
Analogously to above, $V_t$ is uniquely defined by the bag~$\chi(t)$ and consists of~$k'_t$ many fresh elements.

In the context of reduction~$\mathcal{R}'$, for every node~$t$ of~$T$, we let $\ord(t)$ be the \emph{set of partial orderings over~$V_t$}, where each element in~$V_t$ is totally ordered among \emph{the same block of~$\iota$ many elements} (less for remaining $<\iota$ elements).
So, essentially, each ordering in~$\ord(t)$ is the combination of
up to~$\frac{k'_t}{\iota}$ many individual total orderings.
Assume an arbitrary, but fixed total ordering among those blocks of up to~$\iota$ elements. For every~$v\in V_t$ we refer by~$\block(v,t)$ to the \emph{block number}~$1\leq j\leq \frac{k'}{\iota}$ of up to~$\iota$ elements, where~$v$ belongs to. 
%
%


For every node~$t$ of~$T$, let $E_t$ be the largest subset s.t.~$E_t{\,\subseteq\,} \{e_{x,y} {\,|\,} x \prec_{\varphi_t} y\} \cup \{e_{s_{V_t}^j,a} | a{\in}\dom(\varphi_t) \text{ has no}\prec_{\varphi_t}\text{predecessor},j{=}\block(a,t)\} {\,\cup\,} \{e_{b, d_{V_t}^j} | b{\in}\dom(\varphi_t) \text{ has no}\prec_{\varphi_t} \text{successor},j{=}\block(a,t)\}$ 
such that for any node~$t'$ of~$T$ with~$t'\neq t$ we have~$E_t\cap E_{t'}{=}\emptyset$.

Figure~\ref{fig:red2} presents the slightly adapted reduction~$\mathcal{R}'$ taken from~$\mathcal{R}$,
where Formulas~(\ref{red:start}), (\ref{red:reached}), (\ref{red:reachfin}), (\ref{red:fintrans}), and~(\ref{red:fin}), are replaced by Formulas~(\ref{red:start2}), (\ref{red:reached2}), (\ref{red:reachfin2}), (\ref{red:fintrans2}), and~(\ref{red:fin2}), respectively. 
%
%
%
%
With this reduction, we obtain the following result.

\begin{theorem}[Treewidth-Awareness of~$\mathcal{R}'$]\label{thm:tw2}
The reduction~$\mathcal{R}'$ from a Boolean formula~$F$ and an ordering-augmented TD~$\mathcal{T}{=}(T,$ $\chi,\varphi,\psi)$ of~$\mathcal{G}_F$ of width~$k$ to normal program~$\Pi$, using some fixed~$2\leq\iota\leq\frac{k}{\log(\iota)}$,  slightly decreases treewidth.
Precisely, 
the treewidth of~$\mathcal{G}_\Pi$ is at most~$\mathcal{O}(\frac{k}{\log(\iota)})$. Further, there is a TD showing that the tightness width of~$\Pi$ is
in~$\mathcal{O}(\iota)$.
\end{theorem}
\begin{proof}[Proof (Sketch)]
We construct a TD~$\mathcal{T}'=(T,\chi')$ of~$\mathcal{G}_\Pi$ to show that the width of~$\mathcal{T}'$ increases only slightly (compared to~$k$).
To this end, let~$t$ be a node of~$T$ with~$\children(t)=\langle t_1, \ldots, t_\ell \rangle$
and let~$\hat t$ be the parent of~$t$ (if exists).
We define (i)~$R(t)\eqdef \{r_x \mid x\in V_t\}\cup \{r_x \mid x\in V_{t'}, t'\in\children(t)\}\cup\{r_{s_{V_t}^j}, r_{d_{V_t}^j} \mid z\in V_t, j=\block(z,t)\}$, 
(ii)~$E(t)\eqdef \{e_{x,y}, \hat e_{x,y}\mid e_{x,y}\in E_t\}$, 
(iii)~$P(t)\eqdef \{q_t^x, p_t^x, p_{t'}^x, p_\epsilon^x \mid x\in \dom(\varphi_t)\cup\dom(\psi_t), t'\in\prev(x,t)\}$, and
(iv)~$O(t)\eqdef \{o_t^x, o_t^{s_{V_t}^j} \mid x\in V_t, j=\block(x,t)\}\cup\{ o_{t'}^x, o_{t'}^{s_{V_{t'}}^j} \mid t'\in\children(t), x\in V_{t'}\cap V_t\}$.
Then, we let
$\chi'(t) \eqdef R(t) \cup E(t) \cup P(t) \cup O(t)$.
%
%
Observe that~$\mathcal{T}'$ is a TD of~$\mathcal{G}_\Pi$ and by construction~$\Card{\chi'(t)}$ is in $\mathcal{O}(\Card{V_t})$.
%
By definition of~$V_t$, $\Card{V_t}\leq k'$, where~$\iota!^{\frac{k'}{\iota}} \geq 2^k$.
Then, since $\iota!^{\frac{k'}{\iota}}$ is in~$2^{\Omega(k'\cdot\log(\iota))}$ (see Observation~\ref{obs:logk}), we have that~$2^{\Omega(k'\cdot\log(\iota))}$ is at least~$2^k$ and therefore~$k'=\mathcal{O}(\frac{k}{\log(\iota)})$. 
Further, by construction of (i), $R(t)$ contains at most~$2\cdot\iota$ many elements of any SCC of~$D_\Pi$; (ii) $E(t)$ has at most $\iota$ many elements of any SCC of~$D_\Pi$; (iii) $P(t)$ intersects with at most $4\cdot\iota$ many elements of any SCC of~$D_\Pi$. Overall~$\Pi$ has a tightness width over~$\mathcal{T}'$ of~$\mathcal{O}(\iota)$. 
\futuresketch{
by constructing~$2^{\mathcal{O}(k'\cdot\log(k'))}$ many interpretations for some~$k'$.
In order to quantify~$k'$, we take the $\log$ on both sides and end up
with $k$ being in $\mathcal{O}(k'\cdot\log(k'))$.
As a result, we have that~$k'$ is bound by~$\mathcal{O}(\frac{k}{\log(k')})$.
Since~$\log(k')$ is in $\Omega(\log(k))$, we end up with~$k'$ being bounded by~$\mathcal{O}(\frac{k}{\log(k)})$.
}
%
%
%
%
\end{proof}

Interestingly, reduction~$\mathcal{R}'$ cannot be significantly improved either. 

\begin{theorem}[Treewidth Decrease of~$\mathcal{R}'$ is Optimal]\label{thm:runtime2}
Assume a reduction 
from a formula~$F$ 
to a normal logic program~$\Pi$ running in time~$2^{o(k)}\cdot\poly(\Card{\var(F)})$, where~$k$ is the treewidth of~$\mathcal{G}_F$.
Then, under ETH, the treewidth of~$\mathcal{G}_\Pi$
cannot be~$o(\frac{k}{\log(\iota)})$ for 
 fixed~$2\leq\iota\leq\frac{k}{\log(\iota)}$. 
\end{theorem}
\begin{proof}
Assume towards a contradiction that such a reduction, call it~$\mathcal{R}^*$, exists.
Then, we apply this reduction on any~$F$ and for~$\iota$, resulting in program~$\Pi=\mathcal{R}^*(F)$.
We know that~$\Pi$ can be decided~\cite{FandinnoHecher21} in time~$2^{\mathcal{O}(k'\cdot\log(\iota))}\cdot\poly(\Card{\var(F)})$, 
where~$k'$ is in~$o(\frac{k}{\log(\iota)})$, since~$\Pi$ admits tightness width~$\iota$ on some TD of width~$k'$ by Theorem~\ref{thm:tw2}. 
As a result, we have that~$\Pi$ and therefore~$F$ can be decided in time~$2^{o(\frac{k}{\log(\iota)}\cdot\log(\iota))}\cdot\poly(\Card{\var(F)})$, which is in~$2^{o(k)}\cdot\poly(\Card{\var(F)})$, contradicting the ETH. 
%
\end{proof}

Further, we show the following SCC bound, cf., Lemma~\ref{cor:cycle}.

\begin{lemma}[Largest SCC Size by~$\mathcal{R}'$]\label{cor:scc2}
Let~$F$ be a Boolean formula, and~$\mathcal{T}=(T,\chi,\varphi,\psi)$ be an ordering-augmented TD
of~$\mathcal{G}_F$ of width $k$.
Then, the reduction~$\mathcal{R}'$ constructs for some 
 fixed~$2\leq\iota\leq k'$, a program
with
strongly connected components (SCCs) of~$D_\Pi$ of size at most
~$2^{\mathcal{O}(k'\cdot\log(\iota))}$ with~$k'=\frac{k}{\log(\iota)}$,
which equals~$2^{\mathcal{O}(k)}$.
%
\end{lemma}
\begin{proof}
Observe that the largest SCC of~$D_\Pi$ can only be
due to
Formulas~(\ref{red:choose}), (\ref{red:reach}), (\ref{red:reachfin2}), (\ref{red:linkplast}), (\ref{red:linkroot}) and (\ref{red:copy})--(\ref{red:linkpq}),
as these form the only potentially cyclic rules. 
%
The SCC is constructed over variables of the form~$r_x$, $e_{x,y}$, $p_\epsilon^x$, $p_{t}^x$, $q_t^x$ with~$x,y\in V_{t^*}$ for some~$t,t^*$ of~$T$ such that~$V_t=V_{t^*}$.
Thereby, in the SCC we have at most~$k'$ many variables of the form~$r_x$ and~$p_\epsilon^x$ and at most~$k'^2$ many variables of the form~$e_{x,y}$.
However, there are at most~${(\iota!)^{\frac{k'}{\iota}}}^2$ many variables of the form~$p_t^x$ and~$q_t^x$
due to the fact that any potential combination of two orderings~$\varphi_t,\psi_t$ is analyzed by Formulas~(\ref{red:linkbypass})--(\ref{red:linkpq}).
%
%
Obviously, ${(\iota!)^{\frac{k'}{\iota}}}^2$ dominates the SCC size, which is bounded by~${2^{{\mathcal{O}(k'\cdot\log(\iota))}}}^2 = 2^{\mathcal{O}(2\cdot k'\cdot\log(\iota))} = 2^{\mathcal{O}(k'\cdot\log(\iota))}$. 
Then, $2^{\mathcal{O}(k'\cdot\log(\iota))} = 2^{\mathcal{O}(\frac{k}{\log(\iota)}\cdot\log(\iota))} = 
2^{\mathcal{O}(k)}$ establishes the claim. 
%
%
\end{proof}

\begin{restatetheorem}[cor:hscc]
\begin{theorem}[LB $\iota$-Tightness]
Let~$\Pi$ be a $\iota$-tight logic program, where the treewidth of~$\mathcal{G}_\Pi$ is~$k$
such that the largest SCC size of~$D_\Pi$ is in~$2^{\mathcal{O}(k\cdot\log(\iota))}$.
Then, under ETH, the consistency of~$\Pi$ cannot be decided in time~$2^{o(k\cdot\log(\iota))}\cdot\poly(\Card{\at(\Pi)})$.
\end{theorem}
\end{restatetheorem}
\begin{proof}
The result follows from 
the treewidth-awareness by Theorem~\ref{thm:tw2} and the largest SCC size bound by Lemma~\ref{cor:scc2}.
%
Assume towards a contradiction that
one can solve~$\Pi$ in time~$2^{o(k\cdot\log(\iota))}\cdot\Card{\poly(\at(\Pi))}$.
 %
%
%
%
%
So, we take an arbitrary Boolean formula~$F$ and an ordering-augmented TD~$\mathcal{T}$ of~$\mathcal{G}_F$ of width~$k'$ and construct~$\Pi$ by means of~$\mathcal{R}'$ on~$\iota$.
%
Then, by Lemma~\ref{cor:scc2}, the largest SCC of~$D_\Pi$ is bounded by~$2^{\mathcal{O}(k\cdot\log(\iota))}$. Further, by Theorem~\ref{thm:tw2},  the treewidth~$k$ of~$\Pi$ is bounded by~$\mathcal{O}(\frac{k'}{\log(\iota)})$ and there is a TD demonstrating tightness width~$\mathcal{O}(\iota)$ of~$\Pi$. 
So assuming that~$\Pi$ can be decided in time~$2^{o(k\cdot\log(\iota))}\cdot\poly(\Card{\at(\Pi)})$,
results in solving~$F$ in time~$2^{o(\frac{k'}{\log(\iota)}\cdot\log(\iota))}\cdot\poly(\Card{\at(F)}) = 2^{o(k')}\cdot \poly(\Card{\var(F)})$, which contradicts the ETH.
\end{proof}

\end{document}

\clearpage

\section{Why \ASP Consistency is Harder than \SAT}

Even further, there are some crucial key points that need to be fulfilled:
\begin{itemize}
	\item Decouple different assignments
	\item Query the order of atoms
	\item Ensure the order is synchronized among different nodes
\end{itemize}

Assumptions
\begin{itemize}
	\item Every variable appears in at most two TD nodes of a normalized TD 
	\item $V_t$ contains the vertices used to address assignments over~$\chi(t)$ 
	\item Start vertex~$s_t$ 
	for every leaf node~$t$; has edges to every other vertex in~$V_t$ 
	\item For the root node~$t$ of~$T$, there is a final destination vertex~$d_t$ and edges from every~$V_t$ to~$d_t$ 
	\item $E_t$ contains for non-join node~$t$ of~$T$, edges among new variables~$V_t$ in~$t$ used for capturing orderings as well as from~$V_{t'}$ to~$V_t$ (directed) for non-join child nodes~$t$; further edges from~$s_t$ to every vertex in~$V_t$ for every leaf node~$t$ are also in~$E_t$ (as mentioned in the previous point)
	\item $\pred(x,t)\eqdef \children(t) \text{ if } x\in V_{t'} \text{ for some }t'\in \children(t)\text{, otherwise }\{\epsilon\}$
%
	\item for non-join node~$t$ we have that $\varphi_t$ is over variables~$V_t$ (or empty) and $\psi_t$ is over variables~$V_{t^\star}$, where~$t^\star$ is the closest ancestor of~$t$ in~$T$ (without going over join nodes) such that $V_t \neq V_{t^\star}$ (or empty). 
	\item there, we define~$\prec_{\varphi_t}$ and~$\prec_{\psi_t}$ to mark the \emph{direct neighbor relation}
	\item for non-join node~$t$, we let~$\mathcal{I}_t: \ord(t) \rightarrow 2^{\chi(t)}$ be the mapping from ordering among~$V_t$ to assignment of variables in~$\chi(t)$
	\item every pair of~$\varphi_t, \psi_t$ has to occur
\end{itemize}

Extensions for~$\iota$ with~$1\leq i \leq k$:
We need more than one start~$s$, do not cross swords!
Also, one needs to ensure that each~$\psi_t$ or $\varphi_t$ has at least~$\frac{k}{\iota}$ many SCC parts (i.e., at least one vertex per SCC is sufficient).

\begin{lemma}[Inconsistency Causes Cycle]
Let~$F$ be a Boolean formula, $\mathcal{T}=(T,\chi)$ be a TD of~$\mathcal{G}_F$.
Further, let~$I$ be any interpretation of~$F$.
Then, there is a corresponding answer set
\end{lemma}

New approach: Only directed from the leaves towards the root,
start nodes in every leaf node. it seems that the largest cycle size
can be $2^{\mathcal{O}(k\cdot\log(k))}$, i.e., also the cycle
does not need to be "well-spread", namely using distance about~$2^{\mathcal{O}(k\cdot\log(k))}$.
Note that this fragment already requires ordering / numbers $\log(2^{\mathcal{O}(k\cdot\log(k)})$.

First lemma: 
Then, there cannot be a node~$t'$ of~$T$
such that~$\varphi_t=\langle a_1 \ldots, a_\ell\rangle$ is incompatible 
with~$\psi_t=\langle a'_1, \ldots, a'_{\ell'}\rangle$
and we have that both~$\{q_t^{a_\ell}, q_t^{a'_{\ell'}}\}\subseteq M$.

Second lemma:
Then, we have that there is no cycle (due to out degree at most one + every one needs to be reached; also, join nodes wait for every child reachability).

Third lemma:
Then, for every~$V_t$ precisely one ordering is decided / contained in such an answer set.
Assume that this is not the case, i.e., no ordering is compatible. however, lemma 2?

Then, there is an ordering~$\varphi$ over~$\at(\Pi)$,
where every atom of~$M$ is proven. 
%
%
Next, we construct a model~$I$ of~$F$ as follows.
For each~$x\in\at(\Pi)$, we let (c1) $x\in I$ if $x\in M$.
For each node~$t$ of~$T$, 
and~$x\in\chi(t)$: 
(c2) For every~$l\in \bvali{x}{t}{i}$ with~$i=\hat\varphi_t(x)$, 
we set $l\in I$ if~$l$ is a variable. 
(c3) If there is a rule~$r\in\Pi_t$ proving~$x$, 
we let both~$p^x_{<t}, p^x_{t}\in I$.
Finally, (c4) we set $p^x_{<t} \in I$, if~$p^x_{<{t'}}\in I$ for~$t'\in\children(t)$. 

It remains to show that~$I$ is indeed a model of~$F$.
By (c1), Formulas~(\ref{red:checkrules}) are satisfied by~$I$.
Further, by (c2) of~$I$, the order of~$\varphi$ is preserved
among~$\chi(t)$ for each node~$t$ of~$T$, therefore Formulas~(\ref{red:prop}) are satisfied by~$I$.
Further, by definition of TDs, for each rule~$r\in \Pi$ there is a node~$t$ with~$r\in\Pi_t$.
Consequently, $M$ is proven with ordering~$\varphi$, 
for each~$x\in M$ there is a node~$t$ and a rule~$r\in\Pi_t$ proving~$x$.
Then, Formulas~(\ref{red:checkfirst}) are satisfied by~$I$ due to (c3), 
and Formulas~(\ref{red:check}) are satisfied by~$I$ due to (c4). 
Finally, by connectedness of TDs, also Formulas~(\ref{red:checkremove}) and~(\ref{red:checkremove2}) are satisfied.

%

%

``$\Leftarrow$'':
Given any model~$I$ of~$F$.
Then, we construct an answer set~$M$ of~$\Pi$ as follows.
We set~$a\in M$ if~$a\in I$ for any~$a\in\at(\Pi)$.
We define for each node~$t$ a $t$-local ordering~$\varphi_t$,
where we set $\varphi_t(x)$ to~$j$ for each~$x\in\chi(t)$ such that~$j$ is the decimal number of the binary number for~$x$ in~$t$ given by~$I$.
Concretely, $\varphi_t(x)\eqdef j$, where~$j$ is such $I \models\bvali{x}{t}{j}$.
Then, we define an ordering~$\varphi$ iteratively as follows.
We set~$\varphi(a)\eqdef 0$ for each~$a\in \at(\Pi)$,
where there is no node~$t$ of~$T$ with~$\varphi_t(b) < \varphi_t(a)$.
Then, we set~$\varphi(a)\eqdef 1$ for each~$a\in \at(\Pi)$,
where there is no node~$t$ of~$T$ with~$\varphi_t(b) < \varphi_t(a)$ for some~$b\in\chi(t)$ not already assigned in the previous iteration, and so on. 
In turn, we construct~$\varphi$ iteratively by assigning increasing values to~$\varphi$.
Observe that $\varphi$ is well-defined, i.e., each atom~$a\in\at(\Pi)$ gets a unique value since it cannot be the case for two nodes~$t,t'$ and atoms $x,x'\in\chi(t)\cap\chi(t')$ that~$\varphi_t(x) < \varphi_t(x')$, but~$\varphi_{t'}(x) \geq \varphi_{t'}(x')$.
Indeed, this is prohibited by Formulas~(\ref{red:prop})
and connectedness of~$\mathcal{T}$ ensuring that $\mathcal{T}$ restricted to~$x$ is still 
connected.

It remains to show that~$\varphi$ is an ordering for~$\Pi$ proving~$M$.
Assume towards a contradiction that there is an atom~$a\in M$
that is not proven.
%
Observe that either~$a$ is in the bag~$\chi(n)$ of the root node~$n$ of~$T$, or it is forgotten below~$n$.
In both cases we require a node~$t$ such that~$p^x_{<t} \notin I$
by Formulas~(\ref{red:checkremove2}) and~(\ref{red:checkremove}), respectively.
Consequently, by connectedness of $\mathcal{T}$ and  Formulas~(\ref{red:check}) there is a node~$t'$,
where~$p_{t'}^x \in I$.
But then, since Formulas~(\ref{red:checkfirst}) are satisfied by~$I$,
there is a rule~$r\in\Pi_{t'}$ proving~$a$ with~$\varphi_{t'}$.
Therefore, since by construction of~$\varphi$, there cannot be a node~$t$ of~$T$ with~$x,x'\in\chi(t)$, $\varphi_t(x) < \varphi_t(x')$, but~$\varphi(x) \geq \varphi(x')$, 
$r$ is proving~$a$ with~$\varphi$. 
%
%
%
%
\smallskip
\noindent\textbf{
Strengthening the \SAT Formula.}
Next, we strengthen the previous reduction to ensure to
get rid of duplicate $\mathcal{T}$-local orderings
for a particular answer set of~$\Pi$.
%
%
%

In Formulas~(\ref{red:cnt:ineq}), we ensure that if a variable~$x\in\at(\Pi)$ is set to false, then its ordering position is zero.
Formulas~(\ref{red:cnt:ineq2}) make sure that if the position of~$x$ is set to~$i\geq 1$ in node~$t$, there has to be a bag atom~$y$ having position~$i-1$. 
Intuitively, if this is not the case we could shift the position of~$x$ from~$i$ to~$i-1$.
Finally, Formulas~(\ref{red:cnt:checkfirst}) ensure that whenever in a node~$t$ there is a rule~$r\in\Pi_t$ with~$x\in H_r$ and~$x$ has position~$i\geq 1$, 
either 
there is at least one atom~$y\in B_r^+$ having position~$i-1$, or $r$ is not proving~$x$.
%
%
%
%
%
%
{
\vspace{-.25em}
\begin{align}
	\label{red:cnt:ineq}&\neg x \longrightarrow \bigwedge_{1 \leq j \leq \ceil{\log(\Card{\chi(t)})}}\hspace{-.75em}\neg b_{x_t}^j&&{\text{for each } x\in \chi(t)}\raisetag{2.5em}\\
	\label{red:cnt:ineq2}&\bvali{x}{t}{i} \longrightarrow\hspace{-1em}\bigvee_{y \in \chi(t)\setminus\{x\}}\hspace{-1em}\bvali{y}{t}{i-1}&&{\text{for each } x\in \chi(t),}\notag\\[-1.5em]
	&&&{1 \leq i < {\Card{\chi(t)}}}\raisetag{1.15em}\\
	&\hspace{-1em}\bigwedge_{r\in\Pi_t, x\in H_r, 1 \leq i <\Card{\chi(t)}}\hspace{-2em}(\bvali{x}{t}{i}\longrightarrow 
	\label{red:cnt:checkfirst}&&\hspace{-1.75em}\bigvee_{a\in B_r^+}\hspace{-.75em}\neg a \vee {(a \not\prec_t x)} \vee\hspace{-2.5em}\bigvee_{b\in B_r^- \cup (H_r \setminus \{x\})}\hspace{-2.5em}b\, \vee\notag\\ 
	&
	\bigvee_{y\in B_r^+}\hspace{-.25em}\bvali{y}{t}{i-1})&&{\text{for each } x\in \chi(t)}\raisetag{2.5em} 
	%
	%
	%
	%
	%
\end{align}%
\vspace{-.65em}
}

In general, we do not expect to get rid of
all redundant $\mathcal{T}$-local orderings for an answer set, though.
The reason for this expectation lies in the fact that
the different (chains of) rules required for setting the position for an atom~$a$
might be ``spread'' among the whole tree decomposition.
Consequently, one would need to compare different,
\emph{absolute} values of orderings, cf.,~\cite{Janhunen06}, instead of the 
ordering positions relative to one TD node as presented here, 
which requires to store for each atom in the worst case numbers up to~$\Card{\at(\Pi)}$.
Obviously, this number is then not bounded by the treewidth,
and one cannot encode it without increasing the treewidth in general.
However, if for each answer set~$M$ of~$\Pi$, and every~$a\in M$, there can be only one rule~$r\in \Pi$, where~$a\in H_r \cap M$ and~$M\cap (H_r\setminus\{a\})=\emptyset$,
that is satisfied by~$M$,
then there is a bijective correspondence between answer sets of~$\Pi$ and models of Formulas~(\ref{red:checkrules}) to~(\ref{red:cnt:checkfirst}).
One example of such programs is constructed in the next section.

This section concerns the hardness of ASP consistency when
considering treewidth.
The high-level reason for \ASP being harder than \SAT when
assuming bounded treewidth, lies in the issue
that a TD, while capturing the structural dependencies
of a program, might force an evaluation that is
completely different from the orderings proving answer sets. 
Consequently, during dynamic programming for \ASP, one needs to store in each table~$\tab{t}$ for each node~$t$ during post-order traversal, in addition to an interpretation (candidate answer set), also an ordering among the atoms in those interpretations.
We show 
that under reasonable assumptions in complexity theory,
this worst-case cannot be avoided. Then, the resulting runtime consequences cause \ASP to be slightly harder than \SAT, where 
in contrast to \ASP
storing a table~$\tab{t}$ of only assignments for each node~$t$~suffices.

\begin{figure}
	\centering%
	\begin{tikzpicture}[node distance=7mm,every node/.style={fill,circle,inner sep=2pt}]%
		\node (s1) [label={[text height=1.5ex,yshift=0.0cm,xshift=0.05cm]left:$s_1$}] {};
		\node (s2) [right=of s1,label={[text height=1.5ex,yshift=0.0cm,xshift=0.05cm]left:$s_2$}] {};
		\node (s3) [right=of s2,label={[text height=1.5ex,yshift=0.0cm,xshift=0.05cm]left:$s_3$}] {};
		%
		\node (d1) [above of=s1,yshift=.75em,label={[text height=1.5ex,xshift=0.05cm]left:$d_1$}] {};
		\node (d2) [above of=s2,yshift=.75em,label={[text height=1.5ex,xshift=0.05cm]left:$d_2$}] {};
		\node (d3) [above of=s3,yshift=.75em,label={[text height=1.5ex,xshift=0.05cm]left:$d_3$}] {};
		%
		%
		\node (x1) [rectangle,above=of s1,yshift=-1.15em,xshift=.35em,label={[text height=1.5ex,yshift=0.0cm,xshift=0.05cm]left:$x$}] {};
		\node (x2) [rectangle,above=of s2,yshift=-.85em,xshift=.65em,label={[text height=1.5ex,yshift=0.0cm,xshift=0.05cm]left:$y$}] {};
		\node (x4) [rectangle,above=of s3,yshift=-1.4em,label={[text height=1.5ex,yshift=0.0cm,xshift=0.05cm]left:$z$}] {};
		%
		%
		%
		\draw [->,thick] (s1) to (x1);
		\draw [->,ultra thick, red] (s1) to (x1);
		\draw [->,thick] (s2) to (x1);
		\draw [->,thick] (x1) to (d1);
		\draw [->,ultra thick,red] (x1) to (d1);
		\draw [->,thick] (x2) to (d2);
		\draw [->,ultra thick, red] (x2) to (d2);
		\draw [->,thick] (s2) to (x2);
		\draw [->,ultra thick, red] (s2) to (x2);
		\draw [->,thick] (s3) to (x4);
		\draw [->,ultra thick, red] (s3) to (x4);
		\draw [->,thick] (x4) to (d3);
		\draw [->,ultra thick, red] (x4) to (d3);
		\draw [->,thick] (x2) to (d3);
		\draw [<->,thick] (x4) to (x2);
	\end{tikzpicture}%
	\begin{tikzpicture}[node distance=0.75mm]
\tikzset{every path/.style=thick}

\node (leaf1) [tdnode,label={[yshift=-0.25em,xshift=0.0em]left:$t_1$}] {$\{s_1,x\}$};
\node (leaf1b) [tdnode, right=of leaf1,label={[yshift=0.25em,xshift=0.65em]below right:$t_3$}] {$\{s_2,x,y\}$};
\node (leaf1c) [tdnode,below=of leaf1b,xshift=-.5em,label={[yshift=-0.25em,xshift=0.25em]left:$t_2$}] {$\{d_2,y\}$};
\node (leaf2) [tdnode,label={[xshift=-1.0em, yshift=-0.15em]above right:$t_5$}, right = of leaf1b]  {$\{y,d_3,z\}$};
\node (leaf2b) [tdnode,label={[xshift=-.25em, yshift=-0.15em] right:$t_4$}, below = of leaf2]  {$\{s_3,z\}$};
\coordinate (middle) at ($ (leaf1.north east)!.5!(leaf2.north west) $);
\node (join) [tdnode,ultra thick,label={[]left:$t_6$}, above  = .75mm of middle] {$\{x,y,d_1\}$};


\draw [<-] (join) to (leaf1);
\draw [<-] (join) to (leaf1b);
\draw [<-] (join) to (leaf2);
\draw [<-] (leaf1b) to (leaf1c);
\draw [<-] (leaf2) to (leaf2b);
\end{tikzpicture}
	\vspace{-.45em}
	\caption{An instance $I=(G,P)$ (left) of the \problemFont{Disjoint Paths Problem} and a TD of~$G$ (right).}
	\label{fig:disjpaths}
\end{figure}

We show our novel hardness result by reducing from the \problemFont{(directed) Disjoint Paths Problem}, which is a graph problem defined as follows.
Given a directed graph~$G=(V,E)$, and a set~$P\subseteq V\times V$ 
of disjoint pairs of the form $(s_i, d_i)$ consisting of 
\emph{source}~$s_i$ and \emph{destination}~$d_i$,
where~$s_i, d_i\in V$ such that each vertex occurs at most once in~$P$, 
i.e., $\Card{\bigcup_{(s_i,d_i)\in P}\{s_i,d_i\}}=2\cdot\Card{P}$.
Then, $(G,P)$ 
is an instance of the \problemFont{Disjoint Paths Problem},
asking whether there exist $\Card{P}$ many (vertex-disjoint) paths
from~$s_i$ to~$d_i$ for~$1\leq i\leq \Card{P}$.
Concretely, each vertex of~$G$ is allowed to appear in at most one of these paths.
For the ease of presentation, we assume without loss of generality~\cite{LokshtanovMarxSaurabh11}
that sources~$s_i$ have no incoming edge $(x,s_i)$,
and destinations~$d_i$ have no outgoing edge~$(d_i,x)$.
%

\begin{example}
Figure~\ref{fig:disjpaths} (left) shows an instance~$I=(G,P)$ of the \problemFont{Disjoint Paths Problem},
where $P$ consists of pairs of the form $(s_i,d_i)$.
The only solution to~$I$ is both emphasized and colored in red.
Figure~\ref{fig:disjpaths} (right) depicts a TD of~$G$.
\end{example}

While under ETH, \SAT cannot be solved in time~$2^{o(k)}\cdot\poly(\Card{\var(F)})$,
where~$k$ is the treewidth of the primal graph of a given propositional formula~$F$,
the  \problemFont{Disjoint Paths Problem} is considered to be even harder.
Concretely, the problem has been shown to be slightly superexponential
as stated in the following proposition.
\begin{proposition}[\cite{LokshtanovMarxSaurabh11}]\label{prop:slightlysuper}
Under ETH, the~\problemFont{Disjoint Paths Problem} is slightly superexponential,
i.e., any instance~$(G,P)$ with~$G=(V,E)$ cannot be solved in time~$2^{o(k\cdot\log(k))}\cdot \poly({\Card{V}})$, where $k=\tw{G}$.
\end{proposition}

It turns out that the 
\problemFont{Disjoint Paths Problem}
is a suitable problem candidate for showing the hardness of \ASP.
Next, we require the following notation of open pairs, 
whose result is then applied in our reduction.
Given an instance~$(G,P)$ of the~\problemFont{Disjoint Paths Problem}, a TD~$\mathcal{T}=(T,\chi)$ of~$G$, and a node~$t$ of~$T$.
Then, a pair~$(s,d)\in P$ is \emph{open in node~$t$}, if either~$s\in \chi_{\leq t}$ (``\emph{open due to source $s$'}') or~$d\in\chi_{\leq t}$ (``\emph{open due to destination $d$}''), but not both.

\begin{proposition}[\cite{Scheffler94}]\label{prop:earlyout}
An instance~$(G,P)$ of the~\problemFont{Disjoint Paths Problem} does not have a solution if there is a TD~$\mathcal{T}=(T,\chi)$ of~$G$ and a bag~$\chi(t)$ with more than~$\Card{\chi(t)}$ many pairs in~$P$ that are open in a node~$t$ of~$T$.
\end{proposition}
\begin{proof}
The result, cf.,~\cite{Scheffler94}, boils down to the fact that each bag~$\chi(t)$, when removed from~$G$, results in a disconnected graph consisting of two components.
Between these components can be at most~$\Card{\chi(t)}$ different paths. 
\end{proof}


\begin{figure}
\centering%
	\begin{tikzpicture}[node distance=0.75mm]
\tikzset{every path/.style=thick}

\node (leaf1b) [tdnode,label={[yshift=-0.25em,xshift=0.0em]left:$t_2$}] {$\{s_2,x,y\}$};
\node (leaf1c) [tdnode,below=of leaf1b,label={[yshift=-0.25em,xshift=0.0em]left:$t_1$}] {$\{s_2,d_2,y\}$};
\node (leaf2) [tdnode,label={[xshift=-.5em, yshift=-0.15em]above right:$t_4$}, right = of leaf1b]  {$\{y,d_3,z\}$};
\node (leaf2b) [tdnode,xshift=.4em,label={[xshift=-.25em, yshift=-0.15em] right:$t_3$}, below = of leaf2]  {$\{s_3,d_3,z\}$};
\coordinate (middle) at ($ (leaf1.north east)!.5!(leaf2.north west) $);
\node (join) [tdnode,ultra thick,label={[]left:$t_5$}, above  = .75mm of middle] {$\{s_1,d_1,x,y\}$};


\draw [<-] (join) to (leaf1);
\draw [<-] (join) to (leaf2);
\draw [<-] (leaf2) to (leaf2b);
\draw [<-] (leaf1b) to (leaf1c);
\end{tikzpicture}%
\hspace{-4em}\begin{tikzpicture}[node distance=0.5mm]%
\tikzset{every path/.style=thick}

\node (abstand) [white] {};
\node (leaf2) [tdnode,label={[xshift=-.25em, yshift=-0.15em] right:$t_1$}, right = 4.1em of abstand]  {$\{s_3,d_3,y,z\}$};
\node (leaf3) [tdnode,above=of leaf2, label={[xshift=-.30em, yshift=-0.15em] right:$t_2$}]  {$\{s_3,d_3,y\}$};
\coordinate (middle) at ($ (leaf3.north east)!.5!(leaf3.north west) $);
\node (join) [tdnode,ultra thick,label={[xshift=-3em,yshift=.25em]below left:$t_3$}, above  = .75mm of middle] {$\{s_1,d_1,s_2,d_2,s_3,d_3,x,y\}$};


\draw [<-] (join) to (leaf3);
\draw [<-] (leaf3) to (leaf2);
\end{tikzpicture}
	\vspace{-.45em}
	\caption{A pair-respecting TD (left), and a pair-connected TD~$\mathcal{T}$ (right) of $(G,P)$ of Figure~\ref{fig:disjpaths}.}
	\label{fig:tds}
\end{figure}

\noindent
\textbf{Preparing pair-connected TDs.}
Before we present the actual reduction, we need to define a \emph{pair-respecting} tree decomposition of an instance~$(G,P)$ of the~\problemFont{Disjoint Paths Problem}.
Intuitively, such a TD of~$G$ additionally ensures 
that each pair in~$P$ is encountered together in some TD bag.

\begin{definition}
A TD~$\mathcal{T}=(T,\chi)$ of~$G$ is a
\emph{pair-respecting TD} of~$(G,P)$ if for any pair~$p=(s,d)$ with~$p\in P$, (1) whenever~$p$
is open in a node~$t$ due to~$s$, or due to~$d$, then~$s\in\chi(t)$, or $d\in\chi(t)$, respectively.
Further, (2) whenever $p$ is open in a node~$t$, 
but not open in the parent~$t'$ of $t$ (``$p$ is \emph{closed in~$t'$}''), both~$s,d\in\chi(t')$. 
\end{definition}

We observe that such a pair-respecting TD can be computed with
only a linear increase in the (tree)width in the worst case.
%
%
%
Concretely, we can turn any TD~$\mathcal{T}=(T,\chi)$ of~$G$ into
a pair-respecting TD~$\mathcal{T}'=(T,\chi')$ of~$(G,P)$.
Thereby, the tree~$T$ is traversed for each~$t$ of~$T$ in post-order, and vertices of~$P$ are added to~$\chi(t)$
accordingly, resulting in~$\chi'(t)$, such that conditions (1) and (2) of pair-respecting TDs are met.
Observe, that this 
doubles the sizes of the bags in the worst case,
since by Proposition~\ref{prop:earlyout}
there can be at most bag-size many open pairs.

\begin{example}
Figure~\ref{fig:tds} (left) shows a pair-respecting TD of~$(G,P)$ of Figure~\ref{fig:disjpaths},
which can be obtained by transforming the TD of Figure~\ref{fig:disjpaths} (right), followed by simplifications.
\end{example}

Given a sequence~$\sigma$ 
of
pairs of~$P$ in the order of closure
with respect to the post-order of~$T$.
We refer to~$\sigma$ by the \emph{closure sequence} of~$\mathcal{T}$.
We denote by~$p\in_i\sigma$ that pair~$p$ is the \emph{pair 
closed $i$-th} in the order of~$\sigma$.
Intuitively, e.g., the first pair~$p \in_1\sigma$
indicates that pair $p\in P$ 
is the first 
to be closed 
when traversing~${T}$ in post-order.

\begin{definition}
%
A \emph{pair-connected TD}~$\mathcal{T}{=}(T,\chi)$ of $(G,P)$ is a 
pair-respecting TD of $(G,P)$, if,
whenever a pair~$p\in_i\sigma$ with~$i{>}1$
is closed in a node~$t$ of~$T$,
also for the pair $(s,d)\in_{i-1}\sigma$
closed directly 
before $p$ in $\sigma$, both~$s,d\in\chi(t)$.
\end{definition}

%
%
%
We can turn any pair-respecting, \emph{nice} TD~$\mathcal{T}'{=}(T,\chi')$ of width~$k$ into a pair-connected TD~$\mathcal{T}''{=}(T,\chi'')$ with constant increase in the width.
Let therefore pair~$p\in_i\sigma$ be closed ($i{>}1$) in a node~$t$,
and pair~$(s,d)\in_{i-1}$ be closed before~$p$ in node~$t'$.
Intuitively, we need to add~$s,d$ to all bags~$\chi'(t'), \ldots, \chi'(t)$ of nodes encountered
after node~$t'$ and before node~$t$ of
the post-order tree traversal, resulting in~$\chi''$.
However, 
the width of~$\mathcal{T}''$ 
is at most~$k+3\cdot \Card{\{s,d\}} = k+6$,
since in the tree traversal each node of~$T$ is passed at most $3$ times, 
namely when traversing down, when going from the left branch to the right branch, 
and then also when going upwards.
Indeed, to ensure~$\mathcal{T}''$ is a TD (connectedness condition),
we add at most $6$ additional atoms to every bag.
%
%

\begin{example}
Figure~\ref{fig:tds} (right) depicts a pair-connected TD of~$(G,P)$ of Figure~\ref{fig:disjpaths},
obtainable by transforming the pair-respecting TD of Figure~\ref{fig:tds} (left), followed by simplifications.
\end{example}

\subsection{Reducing from \problemFont{Disjoint Paths} to \ASP}

In this section, we show the main reduction~$R$ of this paper,
assuming any instance~$I=(G,P)$ of the \problemFont{Disjoint Paths Problem}.
Before we construct our program~$\Pi$, we require a  nice,
pair-connected TD~$\mathcal{T}=(T,\chi)$ of~$G$,
whose width is~$k$ and a corresponding closure sequence~$\sigma$.
By Proposition~\ref{prop:earlyout}, for each node~$t$ of~$\mathcal{T}$,
there can be at most~$k$ many open pairs of~$P$, which we assume in the following.
If this was indeed not the case, we can immediately output, 
e.g., $\{a\leftarrow \neg a\}$.

Then, we use the following atoms in our reduction.
Atoms $e_{u,v}$, or $ne_{u,v}$ indicate that edge~$(u,v)\in E$ is used, or unused,
respectively.
Then, $r_u$ for any vertex~$u\in V$ indicates that~$u$ is reached via used edges,
and $r^*_d$ are auxiliary reachability atoms for destination vertices~$d$ (i.e., where $(s,d)\in P$). 
Finally, we also need atom~$f^u_t$ for a node~$t$ of~$T$, and vertex~$u\in\chi(t)$,
to indicate that vertex~$u$ is already finished in node~$t$,
i.e., $u$ has one used, outgoing edge.
The presence of this atom~$f^u_t$ in an answer set prohibits to take additional 
edges of~$u$ in parent nodes of~$t$, which is needed due to the need of disjoint paths
of the~\problemFont{Disjoint Paths Problem}.

The instance~$\Pi=R(I,\mathcal{T})$ constructed by reduction~$R$ consists of three program parts,
namely \emph{reachability}~$\Pi_\mathcal{R}$, 
\emph{linking}~$\Pi_\mathcal{L}$ of two pairs in~$P$, as well as \emph{checking}~$\Pi_\mathcal{C}$ of disjointness of constructed paths.
Consequently, $\Pi=\Pi_\mathcal{R}\cup\Pi_\mathcal{L}\cup\Pi_\mathcal{C}$.
All three programs~$\Pi_{\mathcal{R}}$, $\Pi_\mathcal{L}$, and~$\Pi_\mathcal{C}$ are guided along TD~$\mathcal{T}$,
which ensures that the width of~$\Pi$ is only linearly increased.
Note that this has to be carried out carefully.
In particular, since the number of atoms of the form~$e_{u,v}$ using only vertices~$u,v$ that appear in one bag, can be already quadratic in the bag size.
The goal of this reduction, however, admits only a linear overhead in the bag size.
Consequently, we are, e.g., not allowed to construct rules in~$\Pi$ that require more than~$\mathcal{O}(k)$ edges in one bag of a TD of~$\mathcal{G}_\Pi$.

To this end, 
let the \emph{ready edges~$E^{\text{re}}_{t}$  in node~$t$}
be the set of edges~$(u,v)\in E$  
not present in~$t$ anymore, i.e., 
$\{u,v\}\subseteq \chi(t')\setminus \chi(t)$ 
for any child node~$t'\in\children(t)$. 
Further, let $E^{\text{re}}_{n}$ for the root node~$n=\rootOf(T)$ 
additionally contain also all edges of~$n$, i.e., $E\cap (\chi(n) \times \chi(n))$.
Intuitively, ready edges for~$t$ will be processed in node~$t$.
Note that each edge occurs in exactly one set of ready edges.
Further, for nice TDs~$\mathcal{T}$, we always have~$\Card{E^{\text{re}}_{t}}\leq k$, i.e., 
ready edges are linear in~$k$.
\begin{example}
Recall instance~$I{=}(G,P)$ with $G{=}(V,E)$ of Figure~\ref{fig:disjpaths}, and pair-connected TD~$\mathcal{T}{=}(T,\chi)$ of~$I$ of Figure~\ref{fig:tds} (right).
Then, $E_{t_1}^{\text{re}}{=}\emptyset$,
$E_{t_2}^{\text{re}}{=}\{(y,z), (z,y), (z, d_3), (s_3, z)\}$, since~$z\notin\chi(t_2)$, and
$E_{t_3}^{\text{re}}{=}E\setminus E_{t_2}^{\text{re}}$ for root~$t_3$ of~$\mathcal{T}$.
\end{example}

\vspace{-.6em}
\paragraph{Reachability~$\Pi_\mathcal{R}$.} 
Program~$\Pi_\mathcal{R}$ is constructed as follows.
{
\vspace{-1.5em}
\begin{align}
	\label{red:edgeguess1}&e_{u,v}\leftarrow r_u, \neg ne_{u,v}\hspace{-.4em}&&{\text{for each }(u,v)\in E^{\text{re}}_t}\\
	\label{red:edgeguess2}&ne_{u,v}\leftarrow \neg e_{u,v}&&{\text{for each }(u,v)\in E^{\text{re}}_t}\\
	\label{red:reachx}&r_{v}\leftarrow e_{u,v}&&{\text{for each }(u,v)\in E^{\text{re}}_t}, (s,v)\notin P \\
	\label{red:reach2}&r^*_{d}\leftarrow e_{u,d}&&{\text{for each }(u,d)\in E^{\text{re}}_t}, (s,d)\in P
\end{align}%
}

\vspace{-1em}
\noindent Rules~(\ref{red:edgeguess1}) and~(\ref{red:edgeguess2}) ensure that
there is a partition of edges in used edges~$e_{u,v}$ and unused edges~$ne_{u,v}$.
Additionally, Rules~(\ref{red:edgeguess1}) take care that only edges of adjacent, reachable vertices are used.
Naturally, this requires that initially at least one vertex is reachable (constructed below).
Rules~(\ref{red:reach}) and~(\ref{red:reach2}) ensure reachability~$r_v$ and~$r^*_v$ over used edges~$e_{u,v}$ for non-destination vertex~$v$ and destination~$v$, respectively.
%
%

\medskip
\noindent
\textbf{Linking of pairs~$\Pi_\mathcal{L}$.} 
Program~$\Pi_\mathcal{L}$ is constructed as follows.\hspace{-1em}
{
\vspace{-.35em}
\begin{align}
	\label{red:pair1}&\hspace{-.5em}\leftarrow \neg r_{d}&&{\text{for each }(s,d)\in P}\\
	\label{red:pair1st}&\hspace{-.5em}r_{s_1}\leftarrow &&{\text{for }(s_1,d)\in_1 \sigma}\\
	\label{red:paireven}&\hspace{-.5em}r_{s_i}\leftarrow r_{d_{i-1}}&&{\text{for each }(s_i,d)\in_i\hspace{-.1em} \sigma, (s,d_{i-1})\hspace{-.1em}\in_{i-1}\hspace{-.1em}\sigma} \\
\label{red:paircycles2}&\hspace{-.5em}r_{d_1}\leftarrow r^*_{d_1}&&\text{for }(s, d_1)\in_1\sigma\\
\label{red:paircycles}&\hspace{-.5em}r_{d_{i}}\leftarrow r^*_{d_i}, r_{d_{i-1}}\hspace{-.9em}&&\text{for each }(s, d_i)\hspace{-.1em}\in_i\hspace{-.1em}\sigma,(s',d_{i-1})\hspace{-.1em}\in_{i-1}\hspace{-.2em}\sigma\hspace{-.5em}
\end{align}%
\vspace{-1em}
}

\noindent Rules~(\ref{red:pair1}) 
make sure that, ultimately, destination vertices of all pairs are reached.
As an initial, reachable vertex, Rule~(\ref{red:pair1st}) sets the source vertex~$s$ reachable, whose pair is closed first.
Then, the linking of pairs is carried out along the TD in the order of closure, as given by~$\sigma$.
%
%
%
Thereby, Rules~(\ref{red:paireven}) conceptually construct auxiliary links (similar to edges) between different pairs, in the order of~$\sigma$,
which is guided along the TD 
to ensure only a linear increase in treewidth of~$\mathcal{G}_\Pi$ of the resulting program~$\Pi$. 
Interestingly, these additional dependencies, since guided along the TD, do not increase the treewidth by much as we will see in the next subsection.
Rule~(\ref{red:paircycles2}) makes sure that if destination vertex~$d_1$ of the pair closed first is auxiliary-reached ($r^*_{d_1}$), reachability~$r_{d_1}$ is set.

Then, it is \emph{crucial} that we prevent a source vertex~$s_i$ of a pair~$(s_i,d_i)\in_i \sigma$ 
from reaching a destination vertex~$d_j$ of a pair~$(s_j,d_j)\in_j\sigma$ preceding~$(s_i,d_i)$ in~$\sigma$, i.e., $j<i$.
To this end, we need to construct parts of 
cycles that prevent this.
Concretely, if some source~$s_i$ reaches to~$d_j$, i.e., $d_j$ is reachable via~$s_i$, 
the goal is to have a cyclic reachability from~$d_j$ to~$s_i$, with no external support for corresponding reachability atoms. 
%
%
Actually, Rules~(\ref{red:paireven}) 
also have the purpose of aiding in construction of these potential positive cycles.
Together with Rules~(\ref{red:paircycles}) we achieve that if~$d_j$ is reachable, this cannot be due to~$s_i$, since reachability of~$d_{i-1}, d_{i-2}, \ldots, d_j$ (therefore~$s_i$ itself) is required for reachability of~$s_{i}$. 
Consequently, assuming that there is no external support for these reachability atoms 
(which we will ensure in program~$\Pi_{\mathcal{C}}$ below),
and that if~$s_i$ is reachable, $d_j$ is reachable, we end up with cyclic reachability
without external support.
%
%
Figure~\ref{fig:cycles} shows the positive dependency graph~$D_{R_{\mathcal{L}}}$
of Rules~(\ref{red:paireven})--(\ref{red:paircycles}), where pairs~$(s_i,d_i)\in_i\sigma$, discussed in the following example.

\begin{example}
Consider the dependency graph~$D_{R_{\mathcal{L}}}$ of Rules (\ref{red:paireven}) and~(\ref{red:paircycles}),
as depicted in Figure~\ref{fig:cycles}.
Observe that whenever $s_i$ reaches some $d_j$ with~$j<i$,
this causes a cycle~$C{=}r_{s_i},\ldots, r_{d_j}, r_{d_{j+1}}, \ldots, r_{d_{i-1}}, r_{s_i}$ over reachability atoms (cyclic dependency).
%
%
If each vertex~$u$ of~$G$ can have at most one outgoing edge, i.e., 
only one atom~$e_{u,v}$ in an answer set of~$\Pi=R(I,\mathcal{T})$,
no atom of $C$ can be proven (no external support).
Note that~$C$ could also be constructed by adding
$\mathcal{O}(\Card{P}^2)$ many edges from~$d_i$ to~$d_j$ for~$j>i$.
However, this would cause an increase of structural dependency for~$d_i$,
and in fact, the treewidth increase would be beyond linear. 
\end{example}

\begin{figure}[t]%
  \centering%
  \vspace{-1em}
    \resizebox{.25\linewidth}{!}{%
	\begin{tikzpicture}[node distance=7mm,every node/.style={fill,circle,inner sep=2pt}]%
		\node (s0) [white] {};
		\node (s1) [red,right=of s0,label={[text height=1.5ex,yshift=0.0cm,xshift=0.05cm]left:$r_{s_1}$}] {};
		\node (s2) [right=of s1,label={[text height=1.5ex,yshift=0.0cm,xshift=0.05cm]left:$r_{s_2}$}] {};
		\node (s5) [right=of s2,label={[text height=1.5ex,yshift=0.0cm,xshift=0.05cm]left:$r_{s_3}$}] {};
		\node (sdots) [white,right=of s5,label={[text height=1.5ex,yshift=0.0cm,xshift=0.05cm]left:$ $}] {};
		\node (sdots2) [white,right=of sdots,xshift=-2em,label={[text height=1.5ex,yshift=0.0cm,xshift=0.05cm]left:$\ldots$}] {};
		\node (sn1) [right=of sdots2,label={[text height=1.5ex,yshift=0.0cm,xshift=0.15cm]left:$r_{s_{\Card{P}-1}}$}] {};
		\node (sn) [right=of sn1,label={[text height=1.5ex,yshift=0.0cm,xshift=0.2cm]left:$r_{s_{\Card{P}}}$}] {};
		
		\node (d1) [above of=s1,yshift=.0em,label={[text height=1.5ex,xshift=0.05cm]left:$r_{d_1}$}] {};
		\node (stard1) [below of=d1,yshift=0.95em,label={[text height=1.5ex,xshift=0.05cm]left:$r^*_{d_i}{:}$}] {};
		\node (d2) [above of=s2,yshift=.0em,label={[text height=1.5ex,xshift=0.05cm]left:$r_{d_2}$}] {};
		\node (stard2) [below of=d2,yshift=0.95em,label={[text height=1.5ex,xshift=0.05cm]left:$ $}] {};
		\node (d5) [above of=s5,yshift=.0em,label={[text height=1.5ex,xshift=0.05cm]left:$r_{d_3}$}] {};
		\node (stard3) [below of=d5,yshift=0.95em,label={[text height=1.5ex,xshift=0.05cm]left:$ $}] {};
		\node (ddots) [white,above of=sdots,yshift=.0em,label={[text height=1.5ex,xshift=0.05cm]left:$ $}] {};
		\node (ddots2) [white,above of=sdots2,yshift=.0em,label={[text height=1.5ex,xshift=0.05cm]left:$\ldots$}] {};
		\node (starddots) [white,below of=ddots,xshift=.0em,yshift=0.95em,label={[text height=1.5ex,xshift=0.05cm]left:$ $}] {};
		\node (starddots2) [white,below of=ddots2,yshift=.95em,label={[text height=1.5ex,xshift=0.05cm]left:$\ldots$}] {};
		\node (dn1) [above of=sn1,yshift=.0em,label={[text height=1.5ex,xshift=0.15cm]left:$r_{d_{\Card{P}{-1}}}$}] {};
		\node (stardn1) [below of=dn1,yshift=0.95em,label={[text height=1.5ex,xshift=0.05cm]left:$ $}] {};
		\node (dn) [above of=sn,yshift=.0em,label={[text height=1.5ex,xshift=0.2cm]left:$r_{d_{\Card{P}}}$}] {};
		\node (stardn) [below of=dn,yshift=0.95em,label={[text height=1.5ex,xshift=0.05cm]left:$ $}] {};
		
		\draw [dashed,->,red] (stard1) to (d1);
		\draw [dashed,->,red] (stard2) to (d2);
		\draw [dashed,->,red] (stard3) to (d5);
		\draw [dashed,->,red] (stardn) to (dn);
		\draw [dashed,->,red] (stardn1) to (dn1);
		\draw [dashed,->,red] (starddots) to (ddots);
		\draw [dashed,->,red] (d1) to (s2);
		\draw [dashed,->,red] (d2) to (s5);
		\draw [dashed,->,red] (d5) to (sdots);
		\draw [dashed,->,red] (ddots) to (sn1);
		\draw [dashed,->,red] (dn1) to (sn);
		\draw [out=90,in=90,dashed,->,red] (d1) to (d2);
		\draw [out=90,in=90,dashed,->,red] (d2) to (d5);
		\draw [out=90,in=90,dashed,->,red] (dn1) to (dn); 
		\draw [out=90,in=90,dashed,->,red] (d5) to (ddots);
		\draw [out=90,in=90,dashed,->,red] (ddots) to (dn1);
	\end{tikzpicture}}%
    \vspace{-.95em}
 \caption{Positive dependency graph~$D_{R_{\mathcal{L}}}$ of Rules~(\ref{red:paireven})--(\ref{red:paircycles}) constructed for any closure sequence~$\sigma$ such that~$(s_i,d_i)\in_i\sigma$.}%
    \label{fig:cycles}
\end{figure}

%

\smallskip
\noindent\textbf{Checking of disjointness~$\Pi_\mathcal{C}$.} 
Finally, we create rules in $\Pi$ that enforce at most one outgoing, used edge per vertex.
This is required to ensure that we do not use a vertex twice, as required by the~\problemFont{Disjoint Paths Problem}.
%
We do this by guiding the information, whether the corresponding outgoing edge was used,
via atoms~$f^u_t$ along the TD to ensure that the treewidth is not increased significantly.
%
Having at most one outgoing, used edge per vertex of~$G$ further ensures 
that when a source of a pair~$p$ reaches a destination of a pair preceding~$p$ in~$\sigma$, 
then no atom of the resulting cycle as constructed in $\Pi_{\mathcal{L}}$ 
will be provable (no external support).
Consequently, in the end every source of~$p$ has to reach the destination of~$p$ by the pigeon hole principle.
Program~$\Pi_\mathcal{C}$ is constructed for every node~$t$ with~$t',t''{\in}\children(t)$, if~$t$ has child nodes, as follows.
{
\vspace{-.3em}
\begin{flalign}
	\label{red:setf}&f_{t}^u\leftarrow e_{u,v}&&{\text{for each }(u,v)\in E^{\text{re}}_{t}, u\in\chi(t)}\\
	\label{red:propf}&f_t^u\leftarrow f_{t'}^u&&{\text{for each }u\in\chi(t)\cap\chi(t')}\\
	\label{red:prohibitf}&\leftarrow f_{t'}^u, f_{t''}^u&&{\text{for each }u\in\chi(t')\cap\chi(t''), t'\neq t''}\\
	\label{red:tddegree}&\leftarrow f_{t'}^u, e_{u,v} &&{\text{for each }(u,v)\in E^{\text{re}}_{t}, u\in\chi(t')}\\
	\label{red:localdegree}&\leftarrow e_{u,v}, e_{u,w}&&{\text{for each }(u,v),(u,w)\in E^{\text{re}}_{t}\hspace{-.3em},\ v{\neq}w}
\end{flalign}%
\vspace{-1em}
}

\noindent Rules~(\ref{red:setf}) ensure that the finished flag~$f^u_t$ is set for used edges~$e_{u,v}$.
Then, this information of~$f^u_{t'}$ is guided along the TD from child node~$t'$ to parent node~$t$ by Rules~(\ref{red:propf}).
If for a vertex~$u\in V$ we have~$f^u_{t'}$ and~$f^u_{t''}$ for two different child nodes~$t', t''\in\children(t)$,
this indicates that two different edges were encountered both below~$t'$ and below~$t''$. Consequently,
this situation is avoided by Rules~(\ref{red:prohibitf}).
Rules~(\ref{red:tddegree}) make sure to disallow additional edges for vertex~$u$ in a TD node~$t$,
if the flag~$f^u_{t'}$ of child node~$t'$ is set.
Finally, Rules~(\ref{red:localdegree}) prohibit 
two different edges 
for the same vertex~$u$ within a TD node.

\begin{example}
Recall instance~$I=(G,P)$ with $G=(V,E)$ of Figure~\ref{fig:disjpaths}, pair-connected TD~$\mathcal{T}=(T,\chi)$ of~$I$ of Figure~\ref{fig:tds} (right),
and $E_{t_2}^{\text{re}}=\{(y,z), (z,y), (z, d_3), (s_3, z)\}$.
We briefly present the construction of~$\Pi_{\mathcal{C}}$ for node~$t_2$.
%

\noindent\begin{tabular}{@{\hspace{0.15em}}l@{\hspace{0.15em}}|@{\hspace{0.15em}}l@{\hspace{0.0em}}}
Rules & $\Pi_{\mathcal{L}}$\\
\hline
(\ref{red:setf})& $f^y_{t_2} \leftarrow e_{y,z}$; $f^{s_3}_{t_2} \leftarrow e_{s_3,z}$\\
(\ref{red:propf}) & $f^{s_3}_{t_2} \leftarrow f^{s_3}_{t_1}$; $f^{d_3}_{t_2} \leftarrow f^{d_3}_{t_1}$; $f^y_{t_2} \leftarrow f^y_{t_1}$\\
(\ref{red:tddegree}) & $\leftarrow f^y_{t_1}, e_{y,z}$; $\leftarrow f^z_{t_1}, e_{z,y}$; $\leftarrow f^z_{t_1}, e_{z,d_3}$; $\leftarrow f^{s_3}_{t_1}, e_{s_3,z}$\\
(\ref{red:localdegree})& $\leftarrow e_{z,y}, e_{z,d_3}$
\end{tabular}

\end{example}

%

\futuresketch{\paragraph{Modification for non-nice TDs.}
Note that the reduction as presented above, conceptually works for a \emph{non-nice}, pair-connected TD~$\mathcal{T}'=(T',\chi')$ of~$(G,P)$.
However, one needs to take care of two issues.
First, the reduction lies on~$\mathcal{T}'$ being pair-connected as defined above.
If there is a node in~$\mathcal{T}'$ with~$c$ many child nodes, pair-connectedness could cause an 
increase of~$3\cdot(c+1)$ in the treewidth, depending on how distributed the pairs of~$P$ in~$\mathcal{T}'$ are.
%
Then, one needs to take care that~$\Card{E^{\text{re}}_t}$ is still manageable for a node~$t$ of~$T'$, since already a number quadratic
in the width of~$\mathcal{T}'$ might be insufficient. 
However, one can keep~$\Card{E^{\text{re}}_t}$ small by adding intermediate nodes between~$t$ and the parent of~$t$ to the TD.
}

\subsection{Correctness and Runtime Analysis}

First, we show that the reduction is indeed correct,
followed by a result stating that the treewidth of the reduction is at most linearly worsened,
which is crucial for the runtime lower bound to hold.
Then, we present the runtime and the (combined) main result of this paper.

\begin{lemma}[$\leq 1$ Outgoing Edge]\label{lem:degree}
Given any instance~$I=(G,P)$ of the \problemFont{Disjoint Paths Problem},
and any answer set~$M$ of~$R(I, \mathcal{T})$ using any pair-connected TD~$\mathcal{T}$ of~$(G,P)$.
Then, there cannot be two edges of the form $e_{u,v}, e_{u,w}\in M$.
\end{lemma}
\begin{proof}
Assume towards a contradiction that there are three different vertices~$u,v,w\in V$ with~$e_{u,v}, e_{u,w}\in M$.
Then, by Rules~(\ref{red:localdegree}) there cannot be a node~$t$ with~$(u,v),(u,w)\in E_t^{\text{re}}$.
However, by the definition of TDs, there are nodes~$t', t''$ with~$(u,v)\in E_{t'}^{\text{re}}$ and~$(u,w)\in E_{t''}^{\text{re}}$.
By connectedness of TDs, $u$ appears in each bag of any node of the path~$X$ between~$t'$ and~$t''$.
Then, either~$t'$ is an ancestor of~$t''$ (or vice versa, symmetrical) or there is a common ancestor~$t$.
In the former case, $f^u_{t''}$ is justified by Rules~(\ref{red:setf}) and so is~$f^u_{\hat t}$ on each node~$\hat t$ of~$X$
by Rules~(\ref{red:propf}) and therefore ultimately Rules~(\ref{red:tddegree}) fail due to~$f^u_{t'}, e_{u,w}\in M$.
In the latter case, $f^u_{t''}, f^u_{t'}$ is justified by Rules~(\ref{red:setf}) and so is~$f^u_{\hat t}$ on each node~$\hat t$ of~$X$
by Rules~(\ref{red:propf}). Then, 
Rules~(\ref{red:prohibitf}) fail due to~$f^u_{t'}, f^u_{t''}\in M$.
\end{proof}

\begin{theorem}[Correctness]\label{thm:corr}
Reduction~$R$ as proposed in this section is correct.
More concretely, given an instance~$I=(G,P)$ of the \problemFont{Disjoint Paths Problem},
and a pair-connected TD~$\mathcal{T}=(T,\chi)$ of~$G$.
Then, $I$ has a solution if and only if the program~$R(I,\mathcal{T})$ admits an answer set.
\end{theorem}
\begin{proof}
``$\Rightarrow$'': Given any positive instance~$I$ of~\problemFont{Disjoint Paths Problem}.
Then, there are disjoint paths $P_1, \ldots, P_i, \ldots P_{\Card{P}}$ 
from~$s_1$ to~$d_1$, \ldots, $s_i$ to~$d_i$, \ldots, $s_{\Card{P}}$ to~$d_{\Card{P}}$ for each pair~$(s_i,d_i)\in P$.
Assuming further pair-connected TD~$\mathcal{T}$ of~$I$, we construct in the following
an answer set~$M$ of~$\Pi=R(I,\mathcal{T})$.
To this end, we collect reachable atoms~$A\eqdef \{u \mid u\text{ appears in some }P_i, 1\leq i \leq \Card{P}\}$
and used edges~$U\eqdef\{(u,v) \mid v\text{ appears immediately after }u\text{ in some }P_i, 1\leq i \leq \Card{P}\}$.
Then, we construct answer set candidate~$M\eqdef \{r^*_{d_i}\mid 1\leq i\leq \Card{P}\} \cup \{r_u \mid u\in A\}\cup\{e_{u,v}\mid (u,v)\in U\} \cup\{ne_{u,v}\mid (u,v)\in E\setminus U\}\cup \{f^u_t \mid (u,v) \in U \cap E^{\text{re}}_t\} \cup \{f^u_{t} \mid (u,v) \in U\cap E^{\text{re}}_{t'}, u\in\chi(t), t'\text{ is a descendant of }t\text{ in }T\}$.
It remains to show that~$M$ is an answer set of~$\Pi$.
Observe that~$M$ indeed satisfies all the rules of~$\Pi_{\mathcal{R}}$.
In particular, by construction, we have reachability~$r_v$ for every vertex~$v$ of every pair in~$P$, 
and the partition in used edges~$e_{u,v}$ and unused edges~$ne_{u,v}$ is ensured.
Further, $\Pi_{\mathcal{L}}$ is satisfied, as, again by construction, for each vertex~$v$ of every pair in~$P$, we have~$r_v\in M$.
Finally, $\Pi_{\mathcal{C}}$ is satisfied as by construction~$f^u_t\in M$ iff $e_{u,v}\in M\cap E^{\text{re}}_t$ or
$e_{u,v}\in M\cap E^{\text{re}}_{t'}$ for any descendant node~$t'$ of~$t$ with~$u\in\chi(t)$.
%
It is easy to see that~$M$ is indeed a $\subseteq$-smallest model of the reduct~$\Pi^M$,
since, atoms for used and unused edges form a partition of~$E$.

``$\Leftarrow$'': Given any answer set~$M$ of~$\Pi$.
First, we observe that we can only build paths from sources towards destinations,
as sources have only outgoing edges and destinations allow only incoming edges.
Further, by construction, vertices can only have one used, outgoing edge, cf., Lemma~\ref{lem:degree}.
Consequently, if a vertex had more than one used, incoming edge, one cannot match at least one pair of~$P$ (by combinatorial pigeon hole principle).
Hence, in an answer set~$M$ of $\Pi$, there is at most one incoming edge per vertex.
By construction of~$\Pi$, in order to reach each~$d_i$ with~$(s_i,d_i)\in_i\sigma$, 
$s_i$ cannot reach some~$d_{j'}$ with~$j'< i$.
Towards a contradiction assume otherwise, i.e., $s_i$ reaches~$d_{j'}$.
But then, by construction of the reduction, we also have a reachable path from~$d_{j'}$ to~$s_i$, consisting of~$d_{j'}, d_{j'+1}, \ldots, d_{i-1}, s_i$.
Since every vertex has at most one incoming edge, $d_{j'}$ cannot have any other justification for being reachable, nor does any source on this path.
Hence, this forms a cycle 
without external support, which can not be present in an answer set.
Therefore, $s_i$ only reaches~$d_i$, since otherwise there would be at least one vertex~$s_j$ required to reach~$s_{i'}$ with~$(s_{i'}, d_{i'})\in_{i'}\sigma$, $i'<j$.
Consequently, we construct a witnessing path~$P_i$ for each pair~$(s,d)\in_i \sigma$
as follows: $P_i\eqdef s, p_1, \ldots, p_m, d$ where~$\{e_{s, p_1}, e_{p_1, p_2}, \ldots, e_{p_{m-1}, p_m}, e_{p_m, d}\}\subseteq M$.
Thus, $P_i$ starts with~$s$, follows used edges in~$M$ and reaches~$d$.
%
\end{proof}

\begin{lemma}[Treewidth-Awareness]\label{lem:treewidthaware}
Given an instance~$I=(G,P)$ of the \problemFont{Disjoint Paths Problem},
and a pair-connected, nice TD~$\mathcal{T}$ of~$I$ of width~$k$. 
Then, the treewidth of~$\mathcal{G}_\Pi$,
where~$\Pi=R(I,\mathcal{T})$ is obtained by~$R$, is
at most~$\mathcal{O}(k)$.
\end{lemma}
\begin{proof}
Given any pair-connected, nice TD~$\mathcal{T}=(T,\chi)$ of~$I=(G,P)$. Since~$\mathcal{T}$ is nice, a node in~$T$ has at most~$\ell=2$ many child nodes.
From~$\mathcal{T}$ we construct a TD~$\mathcal{T}'=(T,\chi')$ of~$\mathcal{G}_\Pi$.
Thereby we set for every node~$t$ of~$T$, $\chi'(t)\eqdef \{r_u, f^u_t \mid u \in \chi(t)\} \cup \{r^*_d \mid d\in\chi(t), (s,d)\in P\} \cup \{e_{u,v}, ne_{u,v}, r_u, r_v, f^u_{t'} \mid (u,v)\in E^{\text{re}}_t, t'\in\children(t), u\in\chi(t')\} \cup \{f^u_{t'}, f^u_t \mid t' \in\children(t), u\in\chi(t)\cap\chi(t')\}$. 
Observe that~$\mathcal{T}'$ is a valid TD of~$\mathcal{G}_\Pi$.
Further, by construction we have~$\Card{\chi'(t)} \leq 2\cdot \Card{\chi(t)} + \Card{\chi(t)} + (4 + \ell) \cdot k + (\ell + 1) \cdot \Card{\chi(t)}$,
since~$\Card{E_t^{\text{re}}}\leq k$.
The claim sustains for nice TDs ($\ell=2$).
\end{proof}

\begin{corollary}[Runtime]
Reduction~$R$ as proposed in this section runs for a given
instance~$I=(G,P)$ of the \problemFont{Disjoint Paths Problem} with~$G=(V,E)$,
and a pair-connected, nice TD~$\mathcal{T}$ of~$I$ of width~$k$
and~$h$ many nodes, in time $\mathcal{O}(k\cdot h)$.
\end{corollary}

Next, we are in the position of showing the main result, namely the normal ASP lower bound.

\begin{theorem}[Lower bound]\label{thm:lowerbound}
Given an arbitrary normal or HCF program~$\Pi$, where~$k$ is the treewidth
of the primal graph of~$\Pi$. 
Then, unless ETH fails, the consistency problem for~$\Pi$ 
cannot be solved in time~$2^{o(k\cdot \log(k))}\cdot \poly(\Card{\at(\Pi)})$.
\end{theorem}
\begin{proof}

%

Let~$(G,P)$ be an instance of the~\problemFont{Disjoint Paths Problem}.
%
%
First, we construct~\cite{BodlaenderEtAl13} a nice TD~$\mathcal{T}$ of~$G=(V,E)$ of treewidth~$k$ in time~$c^k\cdot \Card{V}$
for some constant~$c$ such that the width of~$\mathcal{T}$ is at most~$5k+4$.
%
Then, we turn the result into a pair-connected TD~$\mathcal{T}'=(T',\chi')$, thereby having width at most~$k'= 2\cdot(5k+4)+ 6$.
%
%
%
Then, we construct program~$\Pi=R(I,\mathcal{T}')$.
By Lemma~\ref{lem:treewidthaware}, the treewidth of~$\mathcal{G}_\Pi$ is in~$\mathcal{O}(k')$, which is in~$\mathcal{O}(k)$.
Assume towards a contradiction that consistency of $\Pi$ can be decided in time~$2^{o(k\cdot \log(k))}\cdot \poly(\Card{\at(\Pi)})$.
By correctness of~$R$ (Theorem~\ref{thm:corr}), this solves~$(G,P)$, contradicting Proposition~\ref{prop:slightlysuper}.
\end{proof}

Our reduction works by construction for any pair-connected TD.
Consequently, this immediately yields a lower bound for \emph{pathwidth},
which is similar to treewidth, but admits only \emph{path decompositions} (TDs whose tree is just a path).

\begin{corollary}[Pathwidth lower bound]
Given any normal or HCF program~$\Pi$, where~$k$ is the pathwidth
of the primal graph of~$\Pi$. 
Then, unless ETH fails, the consistency problem for~$\Pi$ 
cannot be solved in time~$2^{o(k\cdot \log(k))}\cdot \poly(\Card{\at(\Pi)})$.
\end{corollary}

From Theorem~\ref{thm:lowerbound}, we follow that 
a general reduction from normal or HCF programs to \SAT formulas
can probably not avoid the treewidth (pathwidth) overhead, 
which renders our reduction from the previous section ETH-tight.

\begin{corollary}[ETH-tightness of the Reduction to \SAT]
Under ETH, the increase of treewidth of the reduction using Formulas~(\ref{red:checkrules}) to (\ref{red:checkfirst}) cannot be significantly improved.
\end{corollary}\vspace{-.5em}
\begin{proof}
Assume towards a contradiction that one can 
reduce from an arbitrary normal \ASP
program~$\Pi$, where~$k$ is the treewidth of~$\mathcal{G}_\Pi$ 
to a \SAT formula, whose treewidth is in~$o(k\cdot\log(k))$.
Then, this contradicts Theorem~\ref{thm:lowerbound},
as we can use an algorithm~\cite{SamerSzeider10b} for \SAT being single exponential in the treewidth,
thereby deciding consistency of~$\Pi$ in time~$2^{o(k\cdot\log(k))}\cdot\poly(\Card{\at(\Pi)})$.
\end{proof}


\section{Discussion, Conclusion, and Future Work}

The curiosity of studying and determining the hardness of \ASP and the underlying reasons has attracted the attention of the KR community for a long time.
This paper discusses this question from a different angle,
which hopefully will provide new insights into the hardness of \ASP and foster follow-up work.
The results in this paper indicate that, at least from a structural point of view, deciding the consistency of \ASP is already harder than 
\SAT, since \ASP programs might compactly represent structural dependencies within the formalism.
More concretely, compiling the hidden structural dependencies of a program to a \SAT formula, measured in terms of the well-studied parameter treewidth, most certainly causes a blow-up of the treewidth of the resulting formula.
In the light of a known result~\cite{AtseriasFichteThurley11} on the correspondence of treewidth and the resolution width applied in \SAT solving, this reveals that \ASP consistency might be indeed harder than solving \SAT.
We further presented a reduction from \ASP to \SAT that is aware of the treewidth in the sense that the reduction causes not more than this inevitable blow-up of the treewidth in the worst-case.

The work in this paper gives rise to plenty of future work.
On the one hand, we are currently working on the comparison of different treewidth-aware reductions to \SAT and variants thereof, and how these variants perform in practice.
Further, we are curious about treewidth-aware reductions to \SAT, which preserve answer sets bijectively or are modular~\cite{Janhunen06}.
We hope this work might reopen the quest to study the correspondence of treewidth and \ASP solving
similarly to~\cite{AtseriasFichteThurley11} for \SAT.
Also investigating further structural parameters ``between'' treewidth and directed variants of treewidth could lead to new insights,
since for \ASP directed measures~\cite{BliemWoltranOrdyniak16} often do not yield efficient algorithms.

\futuresketch{
\subsection{Consistency of Head-Cycle-Free Programs}%
%
We can use the algorithm~$\dpa_\PRIM$ to decide the consistency
problem for head-cycle-free programs and simply specify our new local
algorithm (\PRIM) that ``transforms'' tables from one node to another.
As graph representation we use the primal graph.  The idea is to
implicitly apply along the tree decomposition the characterization of
answer sets by~\citex{LinZhao03} extended to head-cycle-free
programs~\cite{Ben-EliyahuDechter94}.
To this end, we store in table~$o(t)$ at each node~$t$ rows of the
form~$\langle I, \mathcal{P}, \sigma\rangle$.
The first position consists of an interpretation~$I$ restricted to the
bag~$\chi(t)$.  For a sequence~$\vec \tabval$, we
write~$\mathcal{I}(\vec \tabval)\eqdef \vec u_{(1)}$ to address the
\emph{interpretation part}.
The second position consists of a set~$\mathcal{P} \subseteq I$ that
represents atoms in~$I$ for which we know that they have already been
proven.
The third position~$\sigma$ is a sequence of the atoms~$I$ such that
there is a super-sequence~$\sigma'$ of~$\sigma$, which induces an ordering~$<_{\sigma'}$.
%
%
%
%
Our local algorithm~\PRIM stores interpretation parts always
restricted to bag~$\chi(t)$ and ensures that an interpretation can be extended
to a model of sub-program~$\prog_{\leq t}$.
More precisely, it guarantees that interpretation~$I$ can be extended
to a model~$I'\supseteq I$ of~$\prog_{\leq t}$ and that the atoms
in~$I'\setminus I$ (and the atoms in $\mathcal{P}\subseteq I$) have
already been \emph{proven}, using some induced ordering~$<_{\sigma'}$
where $\sigma$ is a sub-sequence of~$\sigma'$.
In the end, an interpretation~$\mathcal{I}(\vec u)$ of a row~$\vec u$
of the table~$o(n)$ at the root~$n$ proves that there is a
superset~$I' \supseteq \mathcal{I}(\vec u)$ that is an answer set
of~$\prog = \progt{n}$.

Listing~\ref{fig:prim} presents the algorithm~\PRIM.  Intuitively,
whenever an atom~$a$ is introduced ($\intr$), we decide whether we
include~$a$ in the interpretation, determine bag atoms that can be
proven in consequence of this decision, and update the
sequence~$\sigma$ accordingly.
To this end, we define 
%
for a given interpretation~$I$ and a sequence~$\sigma$ the set
$\gatherproof(I, \sigma, \prog_t) \eqdef \bigcup_{r\in \prog_t, a\in
  H_r}\SB a \mid B_r^+ \subseteq I, I \cap B^-_r = \emptyset, I \cap
(H_r\setminus \{a\}) = \emptyset, B^+_r <_\sigma a \SE$ where
$B^+_r <_\sigma a$ holds if $b <_\sigma a$ is true for every
$b \in B^+_r$.  Moreover, given a
level mapping~$\sigma$ and a
set~$A$ of atoms, we compute the potential level mappings
involving~$A$. Therefore, we let
$\possord(\sigma, a, J) \eqdef \SB \MAI{\sigma}{a\mapsto 0} \SM a\notin J \SE \cup
\{\{a\mapsto i, x \mapsto j \mid \sigma(x) = j, \neg sh \text{ or } j < i\} \cup \{x \mapsto j + 1 \mid \sigma(x) = j, j\geq i, sh\} \mid a\in J, 1\leq i \leq k + 1, sh \in \{1, 0\}\SE$.
When removing ($\rem$) an atom~$a$, we only keep those rows where~$a$ has
been proven (contained in~$\mathcal{P}$) and then restrict remaining rows
to the bag (not containing~$a$). In case the node is of
type~$\join$, we combine two rows in two different child tables,
intuitively, we are enforced to agree on the interpretations~$I$ and
sequences~$\sigma$. However, concerning the individual
proofs~$\mathcal{P}$, it suffices that an atom is proven in one of the
rows.

\input{algorithms/prim}%



\setlength{\tabcolsep}{0.25pt}
\renewcommand{\arraystretch}{0.75}
\begin{figure}[t]
\hspace{-0.75em}
\includegraphics[width=0.5\textwidth]{figure-phc-tables}
\caption{Selected tables of~$\tau$ obtained by~$\dpa_{\PRIM}$ on
  TD~${\cal T}$.} 
\label{fig:running2}
\end{figure}

\begin{example}\label{ex:sat}
  Recall program~$\prog$ from
  Example~\ref{ex:running}. Figure~\ref{fig:running2} depicts a
  TD~$\TTT=(T,\chi)$ of the primal graph~$G_1$ of $\prog$. Further,
  the figure illustrates a snippet of tables of the
  TTD~$(T,\chi,\tau)$, which we obtain when running $\dpa_{\PRIM}$ on
  program~$\prog$ and TD~$\TTT$ according to Listing~\ref{fig:prim}.
  In the following, we briefly discuss some selected rows of those
  tables.
  Note that for simplicity and space reasons, we write $\tau_q$
  instead of $\tau(t_q)$ and identify rows by their node and
  identifier~$i$ in the figure. For example, the row
  $\vec\tabval_{13.3}=\langle I_{13.3}, \mathcal{P}_{13.3},
  \sigma_{13.3}\rangle\in\tab{13}$ refers to the third row of
  table~$\tab{13}$ for node~$t_{13}$. 
  Node~$t_1$ is of type~$\leaf$. Table~$\tab{1}$ has only one row,
  which consists of the empty interpretation, empty set of proven
  atoms, and the empty sequence (Line~\ref{line:primleaf}).
  Node~$t_2$ is of type~$\intr$ and introduces atom~$a$. Executing
  Line~\ref{line:primintr} results in
  $\tab{2}=\{\langle \emptyset, \emptyset,\langle \rangle \rangle,
  \langle \{a\}, \emptyset, \langle a\rangle\rangle\}$.
  Node~$t_3$ is of type~$\intr$ and introduces~$b$. Then, bag-program
  at node~$t_3$ is $\prog_{t_3}=\{a \vee b \hsep\}$.
  By construction (Line~\ref{line:primintr}) we ensure that
  interpretation~$I_{3.i}$ is a model of~$\prog_{t_3}$ for every
  row~$\langle I_{3.i}, \mathcal{P}_{3.i}, \sigma_{3.i}\rangle$
  in~$\tab{3}$.
  Node~$t_{4}$ is of type~$\rem$.  Here, we restrict the rows such
  that they contain only atoms occurring in bag~$\chi(t_4)=\{b\}$.  To
  this end, Line~\ref{line:primrem} takes only
  rows~$\vec\tabval_{3.i}$ of table~$\tab{3}$ where atoms in~$I_{3.i}$
  are also proven,~i.e., contained in~$\mathcal{P}_{3.i}$.
  In particular, every row in table~$\tab{4}$ originates from at least
  one row in~$\tab{3}$ that either proves~$a\in \mathcal{P}_{3.i}$ or
  where~$a\not\in I_{3.i}$. 
  Basic conditions of a TD ensure that once an atom is removed, it
  will not occur in any bag at an ancestor node. Hence, we also
  encountered all rules where atom~$a$ occurs.
  Nodes~$t_5, t_6, t_7$, and~$t_8$ are symmetric to
  nodes~$t_1, t_2, t_3$, and~$t_4$.
  Nodes~$t_9$ and~$t_{10}$ again introduce atoms. 
  Observe that $\mathcal{P}_{10.4} = \{e\}$ since
  $\sigma_{10.4}$ does not allow to prove~$b$ using atom~$e$.
  However, $\mathcal{P}_{10.5}=\{b,e\}$ as the sequence~$\sigma_{10.5}$
  allows to prove~$b$.
  In particular, in row~$\vec\tabval_{10.5}$ atom~$e$ is used to
  derive~$b$.  As a result, atom~$b$ can be proven, whereas
  ordering~$\sigma_{10.4}=\langle b,e\rangle$ does not serve in
  proving~$b$.
  %
  %
  %
  %
  We proceed similar for nodes~$t_{11}$ and $t_{12}$.
  At node~$t_{13}$ we join tables~$\tab{4}$ and~$\tab{12}$ according
  to Line~\ref{line:primjoin}.
  Finally, we have $\tab{14}\neq \emptyset$. Hence, $\prog$ has an
  answer set. We can construct the answer set~$\{b,e\}$ by combining
  the interpretation parts~$I$ of the yellow marked rows of
  Figure~\ref{fig:running2}.
\end{example}

Next, we provide a notion to reconstruct answer sets from a computed
TTD, which allows for computing for a given row its predecessor rows
in the corresponding child tables,~c.f., \cite{FichteEtAl18}.
\newcommand{\llangle}{\ensuremath{\langle\hspace{-2pt}\{\hspace{-0.2pt}}}
\newcommand{\rrangle}{\ensuremath{\}\hspace{-2pt}\rangle}}
\newcommand{\STab}{\ensuremath{\ATab{\AlgA}}}%
%
%
Let $\prog$ be a program, $\TTT=(T, \chi, \tau)$ be an~$\AlgA$-TTD
of~$\mathcal{G}_\prog$, and $t$ be a node of~$T$ where
$\children(t)=\langle t_1, \ldots, t_{\ell}\rangle$. 
%
%
%
%
%
Given a sequence~$\vec s=\langle s_1, \ldots, s_{\ell} \rangle$, we
let
$\llangle \vec s\rrangle \eqdef \langle \{s_1\}, \ldots, \{s_{\ell}\}
\rangle$. 
For a given $\AlgA$-row~$\vec u$, we define the originating
$\AlgA$-rows of~$\vec u$ in node~$t$ by
%
$\orig(t,\vec \tabval) \eqdef \SB \vec s \SM \vec s \in \tau(t_1)
\times \cdots \times \tau({t_\ell}), \vec u \in {\AlgA}(t,\chi(t),
\cdot,(\prog_t,\cdot), \llangle \vec s\rrangle) \SE.$ %
%
%
We extend this to an $\AlgA$-table~$\rho$ by
$\origs(t,\rho) \eqdef \bigcup_{\vec u \in \rho}\orig(t,\vec u)$.
%
%



%
%
%
\begin{example}\label{ex:origins} 
  Consider program~$\prog$ and $\PRIM$-tabled tree
  decomposition~$(T,\chi,\tab{})$ from Example~\ref{ex:sat}.  We focus
  on~$\vec{\tabval_{1.1}} =\langle\emptyset, \emptyset, \langle
  \rangle\rangle$ of table~$\tab{1}$ of leaf~$t_1$. The
  row~$\vec{\tabval_{1.1}}$ has no preceding row,
  since~$\type(t_1)=\leaf$. Hence, we have
  $\origse{\PRIM}(t_1,\vec{\tabval_{1.1}})=\{\langle \rangle\}$.
%
  The origins of row~$\vec{\tabval_{11.1}}$ of
  table~$\tab{11}$ are given by
  $\origse{\PRIM}(t_{11},\vec{\tabval_{11.1}})$, which correspond to
  the preceding rows in
  table~$\tab{10}$ that lead to
  row~$\vec{\tabval_{11.1}}$ of
  table~$\tab{11}$ when running
  algorithm~$\PRIM$,~i.e.,
  $\origse{\PRIM}(t_{11},\vec{\tabval_{11.1}}) = \{\langle
  \vec{\tabval_{10.1}} \rangle, \langle \vec{\tabval_{10.6}} \rangle,
  \langle \vec{\tabval_{10.7}}
  \rangle\}$. Origins of
  row~$\vec{\tabval}_{12.2}$ are given by
  $\origse{\PRIM}(t_{12},\vec{\tabval_{12.2}}) = \{\langle
  \vec{\tabval_{11.2}} \rangle, \langle
  \vec{\tabval_{11.6}} \rangle\}$. Note
  that~$\vec{\tabval_{11.4}}$
  and~$\vec{\tabval_{11.5}}$ are not among those origins, since
  $d$ is not proven.  Observe that
  $\origse{\PRIM}(t_j,\vec\tabval)=\emptyset$ for any
  row~$\vec\tabval\not\in\tab{j}$.
  For node~$t_{13}$ of type~$\join$ and row~$\vec{\tabval_{13.2}}$, we
  obtain
  $\origse{\PRIM}(t_{13},\vec{\tabval_{13.2}}) =
  \{\langle\vec{\tabval_{4.2}},$ $\vec{\tabval_{12.2}} \rangle, \langle\vec{\tabval_{4.2}},$ $\vec{\tabval_{12.3}} \rangle\}$.
%
  %
  %
%
\end{example}

\longversion{\paragraph{Table Algorithm for the Incidence Graph}

\input{algorithms/sinc}%

With the general algorithm in mind (see Figure~\ref{fig:dp-approach}), we are now ready to propose $\INC$, a
new table algorithm for solving \ASP on the semi-incidence graph (see 
Listing~\ref{fig:sinc}). 
%
As in the general approach, \INC computes and stores witness sets, and their
corresponding counter-witness sets. However, in addition, for each witness set
and counter-witness set, respectively, 
we need to store so-called \emph{satisfiability states}
(or \emph{sat-states}, for short), since the atoms of a rule may no longer be
contained in one single bag of the tree decomposition of the semi-incidence graph. Therefore, we
need to remember in each tree decomposition node, ``how much'' of a rule is already satisfied. The
following describes this in more detail.

By definition of tree decompositions and the semi-incidence graph, for
every atom~$a$ and every rule~$r$ of a program, it is true that if atom~$a$ occurs
in rule~$r$, then $a$ and $r$ occur together in at least one bag of
the tree decomposition. As a consequence, the table algorithm encounters every
occurrence of an atom in any rule. In the end, on removal of~$r$, we
have to ensure that $r$ is among the rules that are already
satisfied. However, we need to keep track of whether a witness satisfies a rule, because not all atoms that occur in a rule
occur together in a bag. Hence, when our algorithm traverses
the tree decomposition and an atom is removed we still need to store this
sat-state, as setting the removed atom to a certain truth
value influences the satisfiability of the rule.
Since the semi-incidence graph contains a clique on every set~$A$ of
atoms that occur together in a weight rule body or
choice rule head, 
those atoms~$A$ occur together in a bag in every tree decomposition of the
semi-incidence graph. For that reason, we do {not} need to
incorporate weight or 
choice rules into the satisfiability state, in
contrast to the table algorithm for the incidence graph discussed later~(c.f. Section~\ref{sec:inc}).
%
%
%
%
%

In algorithm~\INC (detailed in Listing~\ref{fig:sinc}), a tuple in the
table~$\tab{t}$ is a triple~$\langle M, \sigma, \CCC \rangle$.  The
set~$M \subseteq \at(\prog)\cap\chi(t)$ represents a witness set. 
The family~$\CCC$ of sets represents counter-witnesses, which we will
discuss in more detail below.
The sat-state~$\sigma$ for $M$ represents rules of $\chi(t)$ satisfied
by a superset of~$M$.  Hence, $M$ witnesses a model~$M'\supseteq M$
where $M' \models \progtneq{t} \cup \sigma$.  
We use the binary operator~$\cup$ to combine sat-states, which ensures
that rules satisfied in at least one operand remain satisfied. For a node~$t$, our algorithm considers a local-program depending on the bag~$\chi(t)$. Intuitively, this provides a local view on the program.

For a node~$t$, our algorithm considers a local-program depending on the bag~$\chi(t)$. Intuitively, this provides a local view on the program.
%




\begin{definition}\label{def:bagprogram}%
  Let $\prog$ be a program, $\TTT=(\cdot,\chi)$ a tree decomposition of $S(\prog)$,
  $t$ a node of $\TTT$ and ${R} \subseteq \prog_t$.
  %
  The \emph{local-program mapping}~${R}^{(t)}: R \rightarrow prog(\at(R)\cap\chi(t))$ assigns to each rule~$r\in R$ a rule obtained from~$r$
  by 
  removing all
    literals~$a$ and $\neg a$ where $a \not\in \chi(t)$.
  \shortversion{}
\end{definition}%

\begin{example}
  Observe
  $\prog_{t_1}^{(t_1)} = \{(r_1, b \hsep \neg a), (r_2, a \hsep \neg b)\}$ and
  $\prog_{t_2}^{(t_2)} = \{(r_1, b \hsep \neg a), (r_2, a \hsep \neg b), (r_3, d\hsep)\}$ for $\prog_{t_1}$, $\prog_{t_2}$ of
  Figure~\ref{fig:graph-td2}. 
\end{example}

Since the local-program mapping~$\prog^{(t)}$ depends on the considered
node~$t$, we may have different local-program mappings for node~$t$ and
its child~$t'$. In particular, the mappings~$\{r\}^{(t)}$ and
$\{r\}^{(t')}$ might already differ for a
rule~$r \in \chi(t) \cap \chi(t')$. In consequence for satisfiability
with respect to sat-states, we need to keep track of a representative
of a rule. 

$\SP(\dot{R}, \sigma) \eqdef \{ a \mid (r, s) \in \dot{R}, a \in H_s, a >_\sigma r\}$

$\checkord(\dot{R}, \sigma, \phi, a) \eqdef \text{true iff } a >_\sigma r \implies r\not\in\phi \text{ for any } (r,s) \in\dot{R}$ where $a\in\at(s)$

$\checkmod(\dot{R}, J, \sigma) \eqdef \text{true iff } J \cap \at(s)  = B_s^+ \cup X$, $\Card{X}\leq 1$, and~$X\subseteq H_s$ where $X = \{a\mid a \in \at(s), a >_\sigma r\} \text{ for any } (r,s) \in\dot{R} \text{ with } \sigma = \langle \ldots, r, \ldots \rangle$

Note that in the end~$X \cap H_s\neq \emptyset$ in at least some bag,
since otherwise the rule is not satisfied anyway, but we also have to
enforce that no body atom can in~$X$!}



Next, we provide statements on correctness and a runtime analysis of
our algorithm.

\begin{theorem}[$\star$\footnote{Proofs of statements marked with ``$\star$'' can be found in the supplemental material.}]\label{thm:primcorrectness}
  The algorithm~$\dpa_\PRIM$ is correct. \\
  In other words, given a head-cycle-free program~$\prog$ and a
  TTD~${\cal T} = (T,\chi,\cdot)$ of~$\mathcal{G}_\prog$
  where~$T=(N,\cdot,n)$ with root~$n$. Then,
  algorithm~$\dpa_\PRIM((\prog,\cdot),\TTT)$ returns the
  $\PRIM$-TTD~$(T,\chi,\tau)$ such that $\prog$ has an answer set if
  and only if
  $\langle \emptyset, \emptyset, \langle
  \rangle\rangle\in\tau(n)$. Further, we can construct all the answer
  sets of~$\prog$ from transitively following the origins
  of~$\tau(n)$.\footnoteitext{\label{foot:nu}$\nu$ contains rows
    obtained by recursively following origins of~$\tau(n)$. Formal details are in Def.~\ref{def:extensions} (supplemental material).}\footnoteitext{\label{foot:phc}Later we use (among others)~\mdpa{\PRIM} where~$\AlgA=\PRIM$.}
\end{theorem}
\begin{proof}[Proof (Idea).]
  %
  For soundness, we state an invariant and establish that this
  invariant holds for every node~$t\in N$.  For each
  row~$\vec\tabval=\langle I, \mathcal{P}, \sigma\rangle\in\tau(t)$,
  we have~$I\subseteq\chi(t), \mathcal{P}\subseteq I$, and~$\sigma$ is
  a sequence over atoms in~$I$. Intuitively, we ensure that
  $I\models\progt{t}$ and that exactly the atoms in~$\attneq{t}$
  and~$\mathcal{P}$ can be proven using a super-sequence~$\sigma'$
  of~$\sigma$.  By construction, we guarantee that we can decide
  consistency if
  row~$\langle \emptyset, \emptyset, \langle
  \rangle\rangle\in\tau(n)$. Further, we can even reconstruct answer
  sets, by following $\origa{\PRIM}$ of this single row back to the
  leaves.
  For completeness, we show that we obtain all the rows required to
  output all the answer sets of~$\prog$.
\end{proof}

\begin{theorem}
  \label{thm:primruntime}
  Given a head-cycle-free program~$\prog$ and a tree
  decomposition~${\cal T} = (T,\chi)$ of~$\mathcal{G}_\prog$ of width~$k$ with $g$
  nodes. Algorithm~$\dpa_{\PRIM}$ runs in time
  $\mathcal{O} (3^{k}\cdot k !  \cdot g)$.
\end{theorem}
\begin{proof}[Proof (Sketch).]
  Let~$d = k+1$ be maximum bag size of the tree
  decomposition~$\TTT$. 
  The table~$\tau(t)$ has at most
  $3^{d} \cdot d!$ rows, since for a row~$\langle I, \mathcal{P}, \sigma\rangle$ we have~$d!$ many sequences~$\sigma$, and by construction of algorithm~$\PRIM$, an atom can be either in~$I$, both in~$I$ and~$\mathcal{P}$, or neither in~$I$ nor in~$\mathcal{P}$.
  %
  %
  %
  %
  %
  In total, with the help of efficient data structures, e.g., for nodes~$t$ with~$\type(t)=\join$, one can establish a runtime bound of~$\bigO{{3^{d}\cdot d!}}$.
  %
  Then, we apply this to every node~$t$ of the tree decomposition,
  which results in running
  time~$\bigO{{3^{d}\cdot d!} \cdot g}\subseteq \bigO{3^{k}\cdot k!\cdot g}$.
  Consequently, the theorem holds.
\end{proof}

In order to obtain an upper bound on factorial, we can simply take
$k! \leq 2^k$ for any fixed~$k\geq 4$. However, more precisely the
factorial is asymptotically bounded as follows.


\begin{proposition}[$\star$]\label{prop:kfact}
  Given any positive integer~$i \geq 1$ and
  functions~$f(k)\eqdef k!, g(k) \eqdef 2^{k^{(i+1)/i}}$. Then,
  $f \in \bigO{g}$.
\end{proposition}

A natural question is whether we can significantly improve this
algorithm for fixed~$k$.  To this end, we take the \emph{exponential
  time hypothesis (ETH)} into account, which states that there is some
real~$s > 0$ such that we cannot decide satisfiability of a given
3-CNF formula~$F$ in
time~$2^{s\cdot\Card{F}}\cdot\CCard{F}^{\mathcal{O}(1)}$.

\begin{proposition}
  Unless ETH fails, consistency of head-cycle-free program~$\prog$
  cannot be decided in time~$2^{o(k)} \cdot \CCard{\prog}^{o(k)}$
  where~$k$ is the treewidth of the primal graph of~$\prog$.
\end{proposition}
\begin{proof}
  The result follows by reducing from SAT to ASP (head-cycle-free)
  similar to the proof of Proposition~\ref{prop:hcfproj}.
\end{proof}

In the construction above, we store an arbitrary but fixed ordering on
the involved atoms. We believe that we cannot avoid these orderings in
general, since we have to compensate arbitrarily ``bad'' orderings
induced by the decomposition, which leads us to the following
conjecture.

\begin{conjecture}
  Unless ETH fails, consistency of a head-cycle-free program~$\prog$
  cannot be decided in
  time~$2^{o(k\cdot \text{log}(k))} \cdot \CCard{\prog}^{o(k)}$
  where~$k$ is the treewidth of the primal graph of~$\prog$.
\end{conjecture}

In other words, we claim that consistency for head-cycle-free programs
is slightly superexponential. We would like to mention
that~\citex{LokshtanovMarxSaurabh11} argue that whenever we cannot
avoid an ordering the problem is expected to be slightly
superexponential.
If the conjecture holds, our algorithm is asymptotically worst-case
optimal, even for fixed treewidth~$k$ since~$\dpa_{\PRIM}$ runs in
time~$\mathcal{O}(2^{k\cdot \text{log}(k)}\cdot g)$, where number~$g$
of decomposition nodes is linear in the size of the
instance~\cite{Bodlaender96}.

\section{Dynamic Programming for~$\PASP$}

\label{sec:projmodelcounting}

\begin{figure}[t]
\centering
\includegraphics[scale=0.8]{figure_projection.pdf}
\caption{Algorithm~$\mdpa{\AlgA}$ consists of~$\dpa_\AlgA$
  and~$\dpa_\PROJ$. 
}
\label{fig:multiarch}
\end{figure}%

In this section, we present our dynamic programming
algorithm$^{\ref{foot:phc}}$~\mdpa{\AlgA}, which allows for solving the projected answer
set counting problem (\PASP).
\mdpa{\AlgA} is based on an approach of projected counting for propositional
formulas~\cite{FichteEtAl18} where TDs are traversed multiple times.
We show that ideas from that approach can be fruitfully extended to
answer set programming.
Figure~\ref{fig:multiarch} illustrates the steps of \mdpa{\AlgA}.
First, we construct the primal graph~$\mathcal{G}_\prog$ of the input program~$\prog$
and compute a TD of $\prog$. Then, we traverse the TD a first time by
running $\dpa_\AlgA$ (Step~3a), which outputs a
TTD~$\TTT_{\text{cons}}=(T,\chi,\tau)$.
Afterwards, we traverse $\TTT_{\text{cons}}$ in pre-order and remove
all rows from the tables 
that cannot be extended to an answer set (\emph{``Purge
  non-solutions''}).
In other words, we keep only rows~$\vec u$ of table~$\tau(t)$ at
node~$t$, if~$\vec u$ is involved in those rows that are used to
construct an answer set of~$\prog$, and let the resulting TTD
be~$\TTT_{\text{purged}}=(T,\chi,\nu)^{\ref{foot:nu}}$. We refer to $\nu$
as~\emph{purged table mapping}.
%
%
%
%
%
%
%
In Step~3b ($\dpa_\PROJ$), we traverse $\TTT_{\text{purged}}$ to count
interpretations with respect to the projection atoms and obtain
$\TTT_{\text{proj}}=(T,\chi,\pi)$. From the table~$\pi(n)$ at the root~$n$ of
$T$, we can then read the projected answer sets count of the input instance.
%
%
In the following, we only describe the local algorithm ($\PROJ$),
since the traversal in $\dpa_\PROJ$ is the same as before.
For \PROJ, 
a row at a node~$t$ is a pair $\langle\rho, c \rangle\in\pi(t)$ where
$\rho \subseteq \nu(t)$ is an $\AlgA$-table and $c$ is a non-negative
integer.
In fact, integer~$c$ stores the number of intersecting solutions
($\ipmc$). However, we aim for the projected answer sets count
($\pmc$), whose computation requires a few additional
definitions. Therefore, we can simply widen definitions from very
recent work~\cite{FichteEtAl18}.


%

%

In the remainder, 
we assume~$(\prog, P)$ to be an instance of~\PASP, $(T, \chi, \tau)$
to be an~$\AlgA$-TTD of~$\mathcal{G}_\prog$ and the mappings~$\tau$, $\nu$, and
$\pi$ as used above. 
Further, let~$t$ be a node of~$T$ with~$\children(t)=\langle t_1, \ldots, t_\ell\rangle$ and let $\rho \subseteq \nu(t)$.
%
%
%
%
\newcommand{\RRR}{\ensuremath{\mathcal{R}}}
%
%
  %
  %
  The relation~$\bucket \subseteq \rho \times \rho$ considers
  equivalent rows with respect to the projection of its
  interpretations by 
  $\bucket \eqdef \SB (\vec u,\vec v) \SM \vec u, \vec v \in \rho,
  \restrict{\mathcal{I}(\vec u)}{P} = \restrict{\mathcal{I}(\vec
    v)}{P}\SE.$
  Let $\buckets_P(\rho)$ be the set of equivalence classes induced
  by~$\bucket$ on~$\rho$,~i.e.,
  $\buckets_P(\rho) \eqdef\, (\rho / \bucket) = \SB [\vec u]_P \SM
  \vec u \in \rho\SE$, where
  $[\vec u]_P = \SB \vec v \SM \vec v \bucket \vec u,\vec v \in
  \rho\}$~\cite{Wilder12a}.
  Further, 
  $\subbuckets_P(\rho) \eqdef \SB S \SM \emptyset \subsetneq S
  \subseteq B, B \in \buckets_P(\rho)\SE$.

\begin{example}\label{ex:equiv} 
  Consider program~$\prog$, set~$P$ of projection atoms,
  TTD~$(T,\chi, \tau)$, and table~$\tab{10}$ from
  Example~\ref{ex:running0} and Figure~\ref{fig:running2}.
  Note that during purging rows~$\vec {u_{10.2}}$ and
  $\vec {u_{10.8}}, \ldots, \vec {u_{10.13}}$ are removed (highlighted gray),
  since they are not involved in any answer set, resulting in table~$\nu_{10}$.
  %
  %
  %
  Then, $\vec{ u_{10.4}} =_P \vec{ u_{10.5}}$ and
  $\vec{ u_{10.6}} =_P \vec{ u_{10.7}}$.  The set~$\nu_{10}/\bucket$
  of equivalence classes of $\nu_{10}$
  is~$\buckets_P(\nu_{10})=\SB \{\vec{ u_{10.1}}\}, \{\vec{ u_{10.3}}\}, \{\vec{ u_{10.4}}, \vec{ u_{10.5}}\}, \{\vec{
    u_{10.6}}, \vec{ u_{10.7}}\}\SE$.
\end{example}

Later, we require to construct already computed projected counts for
tables of children of a given node~$t$. Therefore, we define the
\emph{stored $\ipmc$} of a table~$\rho \subseteq \nu(t)$ in
table~$\pi(t)$ by
$\sipmc(\pi(t), \rho) \eqdef \sum_{\langle \rho, c\rangle \in \pi(t)}
c.$ 
We extend this to a
sequence~$s=\langle \pi(t_1), \ldots, \pi(t_\ell)\rangle$ of tables of
length $\ell$ and a
set~$O = \{\langle \rho_1, \ldots, \rho_\ell\rangle, \langle \rho_1',
\ldots, \rho_\ell'\rangle, \ldots\}$ of sequences of~$\ell$ tables by
$\sipmc(s, O)=\prod_{i \in \{1, \ldots,
  \ell\}}\sipmc(s_{(i)},O_{(i)}).$
In other words, we select the $i$-th position of the sequence together
with sets of the $i$-th positions from the set of sequences.




%
Intuitively, when we are at a node~$t$ in algorithm~$\dpa_\PROJ$ we
have already computed~$\pi(t')$ of $\TTT_{\text{proj}}$ for every node~$t'$
below~$t$.
Then, we compute the projected answer sets count
of~$\rho \subseteq \nu(t)$. Therefore, we apply the
inclusion-exclusion principle to the stored projected answer sets count
of origins.
We define $\pcnt(t,\rho, \langle\pi(t_1),\ldots\rangle) \eqdef 
\sum_{\emptyset \subsetneq O \subseteq {\origs(t,\rho)}}
(-1)^{(\Card{O} - 1)} \cdot \sipmc(\langle \pi(t_1), \ldots\rangle, O)$. 
%
Vaguely speaking, $\pcnt$ determines the $\AlgA$-origins of table~$\rho$, 
goes over all subsets of these origins and looks up the
stored counts ($\sipmc$) in the \PROJ-tables of the children~$t_i$ of~$t$.

\begin{example}\label{ex:pcnt} 
  Consider again program~$\prog$ and TD~$\TTT$ from
  Example~\ref{ex:running1} and Figure~\ref{fig:running2}. First, we
  compute the projected count $\pcnt(t_{4},\{\vec{ u_{4.1}}\}, \langle\pi(t_{3})\rangle)$
  for row~$\vec{ u_{4.1}}$ of table~$\nu(t_{4})$ where
  $\pi(t_3) \eqdef\allowdisplaybreaks[4] \big\SB
  \langle \{\vec{ u_{3.1}}\}, 1\rangle,$
  $\langle \{\vec{ u_{3.2}}\},1\rangle, \langle \{\vec{u_{3.1}}, \vec{
    u_{3.2}}\},1\rangle\big\SE$ with
  $\vec{u_{3.1}}=\langle \emptyset, \emptyset, \langle\rangle \rangle$
  and~$\vec{u_{3.2}}=\langle \{a\}, \emptyset, \langle a\rangle
  \rangle$.
  Note that~$t_5$ has only the child~$t_4$ and therefore the product in~$\sipmc$
  consists of only one factor. 
  Since
  $\origse{\PRIM}(t_4, \vec{ u_{4.1}}) = \{\langle\vec{
    u_{3.1}}\rangle\}$, only the value of~$\sipmc$ for
  set~$\{\langle\vec{ u_{3.1}}\rangle\}$ is non-zero. Hence, we obtain
  $\pcnt(t_4,\{\vec{ u_{4.1}}\}, \langle\pi(t_3)\rangle)=1$. 
  Next, we compute
  $\pcnt(t_{4},\{\vec{ u_{4.1}}, \vec{u_{4.2}}\}, \langle\pi(t_3)\rangle)$. Observe that
  $\origse{\PRIM}(t_4, \{\vec{ u_{4.1}}, \vec{ u_{4.2}}\}) =
  \{\langle\vec{ u_{3.1}}\rangle, \langle\vec{ u_{3.2}}\rangle\}$. We
  sum up the values of~$\sipmc$ for sets~$\{\vec{ u_{4.1}}\}$
  and~$\{\vec{ u_{4.2}}\}$ and subtract the one for
  set~$\{\vec{ u_{4.1}}, \vec{ u_{4.2}}\}$.  Hence, we obtain
  $\pcnt(t_4,\{\vec{ u_{4.1}}, \vec{ u_{4.2}}\}, \langle\pi(t_3)\rangle)=1+1-1=1$.
\end{example}

Next, we provide a definition to compute $\ipmc$, which can be
computed at a node~$t$ for given table~$\rho\subseteq \nu(t)$ by
computing the $\pmc$ for children~$t_i$ of~$t$ using stored $\ipmc$
values from tables~$\pi(t_i)$, subtracting and adding~$\ipmc$ values
for subsets~$\emptyset\subsetneq\varphi\subsetneq\rho$ accordingly.
Formally, $\icnt(t,\rho,s)\eqdef 1$ if $\type(t) = \leaf$ and
otherwise
$\icnt(t,\rho,s)\eqdef \big|\pcnt(t,\rho, s)\;+
\quad\sum_{\emptyset\subsetneq\varphi\subsetneq\rho}(-1)^{\Card{\varphi}}
\cdot \ipmc(t,\varphi, s)\big|$ where
$s = \langle \pi(t_1), \ldots\rangle$.
In other words, if a node is of type~$\leaf$ the $\ipmc$ is one, since
bags of leaf nodes are empty.
%
%
%
%
%
%
Otherwise, we compute the ``non-overlapping'' count of given
table~$\rho\subseteq\nu(t)$ with respect to~$P$, by exploiting the
inclusion-exclusion principle on $\AlgA$-origins of~$\rho$ such that
we count every projected answer set only once. Then we have to 
subtract and add $\ipmc$ values (``all-overlapping'' counts) for
strict subsets~$\varphi$ of~$\rho$, accordingly.

Finally, Listing~\ref{fig:dpontd3} presents the local algorithm~\PROJ,
which stores~$\pi(t)$ consisting of every sub-bucket of the given
table~$\nu(t)$ together with its $\ipmc$.




\input{algorithms/dponpass3}%

\setlength{\tabcolsep}{0.25pt}
\renewcommand{\arraystretch}{0.75}
\begin{figure*}[t]
\centering
\begin{tikzpicture}[node distance=0.5mm]
\tikzset{every path/.style=thick}

\node (l1) [stdnode,label={[tdlabel, xshift=0em,yshift=+0em]right:${t_1}$}]{$\emptyset$};
\node (i1) [stdnode, above=of l1, label={[tdlabel, xshift=0em,yshift=+0em]left:${t_2}$}]{$\{a\}$};
\node (i12) [stdnode, above=of i1, label={[tdlabel, xshift=0em,yshift=+0.35em]left:${t_3}$}]{$\{a,b\}$};
\node (i13) [stdnode, above=of i12, label={[tdlabel, xshift=+0.05em,yshift=+0em]right:${t_4}$}]{$\{b\}$};
\node (lrx) [stdnode, right=2.5em of l1, yshift=-6.9em, label={[tdlabel, xshift=0em,yshift=+0em]left:${t_5}$}]{$\emptyset$};
\node (cc) [stdnode, above=of lrx, label={[tdlabel, xshift=0em,yshift=+0em]left:${t_6}$}]{$\{c\}$};
\node (bc) [stdnode, above=of cc, label={[tdlabel, xshift=0em,yshift=+0em]left:${t_7}$}]{$\{c,e\}$};
\node (l2) [stdnode, above=of bc, label={[tdlabel, xshift=0em,yshift=+0em]left:${t_8}$}]{$\{e\}$};
\node (i2) [stdnode, above=of l2, label={[tdlabel, xshift=0em,yshift=+0em]right:${t_{9}}$}]{$\{d, e\}$};
\node (i22) [stdnode, above=of i2, label={[tdlabel, xshift=0em,yshift=+0em]left:${t_{10}}$}]{$\{b,d,e\}$};
\node (r2) [stdnode, above=of i22, label={[tdlabel, xshift=0em,yshift=+0em]left:${t_{11}}$}]{$\{b, d\}$};
\node (r22) [stdnode, above=of r2, label={[tdlabel, xshift=0.05em,yshift=+0em]left:${t_{12}}$}]{$\{b\}$};
\node (j) [stdnode, above left=of r22, xshift=0.0em, yshift=-0.0em, label={[tdlabel, xshift=0em,yshift=+0.3em]right:${t_{13}}$}]{$\{b\}$};
\node (rt) [stdnode,ultra thick, above=of j, label={[tdlabel, xshift=0em,yshift=+0em]right:${t_{14}}$}]{$\emptyset$};
\node (label) [font=\scriptsize,left=of rt,xshift=0.45em]{${\cal T}$:};
%
%
%
%
\node (leaf1) [stdnodetable, left=3em of i1, yshift=2em, label={[tdlabel, xshift=2em,yshift=1em]below right:$\pi_{3}$}]{%
	\begin{tabular}{l@{\hspace{0.0em}}l@{\hspace{0.0em}}r}%
		\multicolumn{1}{l}{$\langle \tuplecolor{black}{\nu_{3.i}}, $}&\multicolumn{1}{r}{$\tuplecolor{\specialPredColor}{c_{3.i}} \rangle$}\\
		\hline\hline
		$\langle\tuplecolor{black}{\{\langle\tuplecolor{\inputPredColor}{\{a\}}, \tuplecolor{\outputPredColor}{\{a\}}, \tuplecolor{\statePredColor}{\langle a\rangle}\rangle\}}, $&$\tuplecolor{\specialPredColor}{1}\rangle$ \\\hline
		$\langle\tuplecolor{black}{\{\langle\tuplecolor{\inputPredColor}{\{b\}}, \tuplecolor{\outputPredColor}{\{b\}}, \tuplecolor{\statePredColor}{\langle b\rangle}\rangle\}}, $&$\tuplecolor{\specialPredColor}{1}\rangle$ \\\hline
		$\langle\tuplecolor{black}{\{\langle\tuplecolor{\inputPredColor}{\{a\}}, \tuplecolor{\outputPredColor}{\{a\}}, \tuplecolor{\statePredColor}{\langle a\rangle}\rangle}, $& \multirow{2}{*}{$\tuplecolor{\specialPredColor}{1}\rangle$} \\ 
		$\hspace{1.05em}\tuplecolor{black}{\langle\tuplecolor{\inputPredColor}{\{b\}}, \tuplecolor{\outputPredColor}{\{b\}}, \tuplecolor{\statePredColor}{\langle b\rangle}\rangle\},}$ 
	\end{tabular}%
};
\node (leaf1b) [stdnodenum,left=of leaf1,xshift=0.6em,yshift=+0.0em]{%
	\begin{tabular}{c}%
		\multirow{1}{*}{$i$}\\ 
		\hline\hline
		$1$ \\\hline
		$2$ \\\hline
		\multirow{2}{*}{$3$}\\\\ 
	\end{tabular}%
};
\node (leaf0x) [stdnodetable, left=-3.5em of leaf1b, yshift=4.5em, label={[tdlabel, xshift=2em,yshift=-0.15em]below left:$\pi_{4}$}]{%
	\begin{tabular}{l@{\hspace{0.0em}}l@{\hspace{0.0em}}r}%
		\multicolumn{1}{l}{$\langle \tuplecolor{black}{\nu_{4.i}}, $}&\multicolumn{1}{r}{$\tuplecolor{\specialPredColor}{c_{4.i}} \rangle$}\\
		\hline\hline
		$\langle\tuplecolor{black}{\{\langle\tuplecolor{\inputPredColor}{\emptyset}, \tuplecolor{\outputPredColor}{\emptyset}, \tuplecolor{\statePredColor}{\langle \rangle}\rangle\}}, $&$\tuplecolor{\specialPredColor}{1}\rangle$ \\\hline
		$\langle\tuplecolor{black}{\{\langle\tuplecolor{\inputPredColor}{\{b\}}, \tuplecolor{\outputPredColor}{\{b\}}, \tuplecolor{\statePredColor}{\langle b\rangle}\rangle\}}, $&$\tuplecolor{\specialPredColor}{1}\rangle$ \\\hline
		$\langle\tuplecolor{black}{\{\langle\tuplecolor{\inputPredColor}{\emptyset}, \tuplecolor{\outputPredColor}{\emptyset}, \tuplecolor{\statePredColor}{\langle \rangle}\rangle}, $& \multirow{2}{*}{$\tuplecolor{\specialPredColor}{1}\rangle$} \\ 
		$\hspace{1.05em}\tuplecolor{black}{\langle\tuplecolor{\inputPredColor}{\{b\}}, \tuplecolor{\outputPredColor}{\{b\}}, \tuplecolor{\statePredColor}{\langle b\rangle}\rangle\},}$ 
	\end{tabular}%
};
\node (leaf0b) [stdnodenum,left=of leaf0x,xshift=0.6em,yshift=0pt]{%
	\begin{tabular}{c}%
		\multirow{1}{*}{$i$}\\ 
		\hline\hline
		$1$ \\\hline
		$2$ \\\hline
		\multirow{2}{*}{$3$}\\\\ 
	\end{tabular}%
};
\node (leaf2b) [stdnodenum,left=6.5em of j,xshift=-0.25em,yshift=-11.5em]  {%
	\begin{tabular}{c}%
		\multirow{1}{*}{$i$}\\ 
		\hline\hline
		$1$\\\specialrule{.1em}{.05em}{.05em}	
		$2$\\\specialrule{.1em}{.05em}{.05em}	
		$3$\\\hline
		$4$\\\hline
		\multirow{2}{*}{$5$}\\\\ 
	\end{tabular}%
};
\node (leaf2) [stdnodetable,right=-3em of leaf2b, label={[tdlabel, xshift=0.1em,yshift=-0.25em]right:$\pi_{9}$}]  {%
	\begin{tabular}{l@{\hspace{0.0em}}l@{\hspace{0.0em}}r}%
		\multicolumn{1}{l}{$\langle \tuplecolor{black}{\nu_{9.i}}, $}&\multicolumn{1}{r}{$\tuplecolor{\specialPredColor}{c_{9.i}} \rangle$}\\
		\hline\hline
		$\langle\tuplecolor{black}{\{\langle\tuplecolor{\inputPredColor}{\{d\}}, \tuplecolor{\outputPredColor}{\emptyset}, \tuplecolor{\statePredColor}{\langle \rangle}\rangle\}}, $&$\tuplecolor{\specialPredColor}{1}\rangle$ \\\specialrule{.1em}{.05em}{.05em}	
		$\langle\tuplecolor{black}{\{\langle\tuplecolor{\inputPredColor}{\{e\}}, \tuplecolor{\outputPredColor}{\{e\}}, \tuplecolor{\statePredColor}{\langle e\rangle}\rangle\}}, $&$\tuplecolor{\specialPredColor}{1}\rangle$ \\\specialrule{.1em}{.05em}{.05em}	
		$\langle\tuplecolor{black}{\{\langle\tuplecolor{\inputPredColor}{\{d,e\}}, \tuplecolor{\outputPredColor}{\{e\}}, \tuplecolor{\statePredColor}{\langle d,e\rangle}\rangle\}}, $&$\tuplecolor{\specialPredColor}{1}\rangle$ \\\hline
		$\langle\tuplecolor{black}{\{\langle\tuplecolor{\inputPredColor}{\{d,e\}}, \tuplecolor{\outputPredColor}{\{e\}}, \tuplecolor{\statePredColor}{\langle e,d\rangle}\rangle\}}, $&$\tuplecolor{\specialPredColor}{1}\rangle$ \\\hline
		$\langle\tuplecolor{black}{\{\langle\tuplecolor{\inputPredColor}{\{d,e\}}, \tuplecolor{\outputPredColor}{\{e\}}, \tuplecolor{\statePredColor}{\langle d,e\rangle}\rangle}, $& \multirow{2}{*}{$\tuplecolor{\specialPredColor}{1}\rangle$} \\ 
		$\hspace{1.05em}\tuplecolor{black}{\langle\tuplecolor{\inputPredColor}{\{d,e\}}, \tuplecolor{\outputPredColor}{\{e\}}, \tuplecolor{\statePredColor}{\langle e,d\rangle}\rangle\},}$ 
	\end{tabular}%
};
\node (joinrb2) [stdnodenum,left=-0.45em of leaf2] {%
	\begin{tabular}{c}
		\multirow{1}{*}{$i$}\\
		\hline\hline
		$1$ \\\specialrule{.1em}{.05em}{.05em}	
		$2$ \\\specialrule{.1em}{.05em}{.05em}	
		$3$ \\\hline
		$4$ \\\hline
		\multirow{2}{*}{$5$} \\\\
	\end{tabular}%
};
\coordinate (middle) at ($ (leaf1.north east)!.5!(leaf2.north west) $);
\node (join) [stdnodetable,left=-0.2em of i13, yshift=4.5em, label={[tdlabel, xshift=0.1em,yshift=-0.15em]below right:$\pi_{{13}}$}] {%
	\begin{tabular}{l@{\hspace{0.0em}}l@{\hspace{0.0em}}r}%
		\multicolumn{1}{l}{$\langle \tuplecolor{black}{\nu_{13.i}}, $}&\multicolumn{1}{r}{$\tuplecolor{\specialPredColor}{c_{13.i}} \rangle$}\\
		\hline\hline
		$\langle\tuplecolor{black}{\{\langle\tuplecolor{\inputPredColor}{\emptyset}, \tuplecolor{\outputPredColor}{\emptyset}, \tuplecolor{\statePredColor}{\langle \rangle}\rangle\}}, $&$\tuplecolor{\specialPredColor}{2}\rangle$ \\\hline
		$\langle\tuplecolor{black}{\{\langle\tuplecolor{\inputPredColor}{\{b\}}, \tuplecolor{\outputPredColor}{\{b\}}, \tuplecolor{\statePredColor}{\langle b\rangle}\rangle\}}, $&$\tuplecolor{\specialPredColor}{2}\rangle$ \\\hline
		$\langle\tuplecolor{black}{\{\langle\tuplecolor{\inputPredColor}{\emptyset}, \tuplecolor{\outputPredColor}{\emptyset}, \tuplecolor{\statePredColor}{\langle \rangle}\rangle}, $& \multirow{2}{*}{$\tuplecolor{\specialPredColor}{1}\rangle$} \\ 
		$\hspace{1.05em}\tuplecolor{black}{\langle\tuplecolor{\inputPredColor}{\{b\}}, \tuplecolor{\outputPredColor}{\{b\}}, \tuplecolor{\statePredColor}{\langle b\rangle}\rangle\},}$ 
	\end{tabular}%
};
\node (joinb) [stdnodenum,left=-0.45em of join] {%
	\begin{tabular}{c}
		\multirow{1}{*}{$i$}\\
		\hline\hline
		$1$ \\\hline
		$2$ \\\hline
		\multirow{2}{*}{$3$}\\\\
	\end{tabular}%
};
\node (leaf0) [stdnodetable,below=-1.5em of l1, xshift=-5em, label={[tdlabel, xshift=0.1em,yshift=0.15em]right:$\pi_{1}$}] {%
	\begin{tabular}{l@{\hspace{0.0em}}l@{\hspace{0.0em}}r}%
		\multicolumn{1}{l}{$\langle \tuplecolor{black}{\nu_{1.i}}, $}&\multicolumn{1}{r}{$\tuplecolor{\specialPredColor}{c_{1.i}} \rangle$}\\
		\hline\hline
		$\langle\tuplecolor{black}{\{\langle\tuplecolor{\inputPredColor}{\emptyset}, \tuplecolor{\outputPredColor}{\emptyset}, \tuplecolor{\statePredColor}{\langle\rangle}\rangle\}}, $&$\tuplecolor{\specialPredColor}{1}\rangle$
	\end{tabular}%
};
\node (leaf0r) [stdnodetable,above=0.5em of rt, xshift=1.1em, label={[tdlabel, xshift=0.1em,yshift=-0.1em]below right:$\pi_{14}$}] {%
	\begin{tabular}{l@{\hspace{0.0em}}l@{\hspace{0.0em}}r}%
		\multicolumn{1}{l}{$\langle \tuplecolor{black}{\nu_{14.i}}, $}&\multicolumn{1}{r}{$\tuplecolor{\specialPredColor}{c_{14.i}} \rangle$}\\
		\hline\hline
		$\langle\tuplecolor{black}{\{\langle\tuplecolor{\inputPredColor}{\emptyset}, \tuplecolor{\outputPredColor}{\emptyset}, \tuplecolor{\statePredColor}{\langle\rangle}\rangle\}}, $&$\tuplecolor{\specialPredColor}{3}\rangle$
	\end{tabular}%
};\node (leaf0nr) [stdnodenum,yshift=0.0em, left=-0.5em of leaf0r] {%
	\begin{tabular}{c}%
		\multirow{1}{*}{$i$}\\ 
		\hline\hline
		$1$
	\end{tabular}%
};
\node (leaf0n) [stdnodenum,yshift=0.0em, left=-0.5em of leaf0] {%
	\begin{tabular}{c}%
		\multirow{1}{*}{$i$}\\ 
		\hline\hline
		$1$
	\end{tabular}%
};
\node (joinrrt) [stdnodetable,right=6.5em of i13, yshift=3em, label={[tdlabel, xshift=0.1em,yshift=+0.25em]right:$\pi_{{12}}$}] {%
		\begin{tabular}{l@{\hspace{0.0em}}l@{\hspace{0.0em}}r}%
		\multicolumn{1}{l}{$\langle \tuplecolor{black}{\nu_{12.i}}, $}&\multicolumn{1}{r}{$\tuplecolor{\specialPredColor}{c_{12.i}} \rangle$}\\
		\hline\hline
		$\langle\tuplecolor{black}{\{\langle\tuplecolor{\inputPredColor}{\emptyset}, \tuplecolor{\outputPredColor}{\emptyset}, \tuplecolor{\statePredColor}{\langle \rangle}\rangle\}}, $&$\tuplecolor{\specialPredColor}{2}\rangle$ \\\hline
		$\langle\tuplecolor{black}{\{\langle\tuplecolor{\inputPredColor}{\{b\}}, \tuplecolor{\outputPredColor}{\emptyset}, \tuplecolor{\statePredColor}{\langle b \rangle}\rangle\}}, $&$\tuplecolor{\specialPredColor}{2}\rangle$ \\\hline
		$\langle\tuplecolor{black}{\{\langle\tuplecolor{\inputPredColor}{\{b\}}, \tuplecolor{\outputPredColor}{\{b\}}, \tuplecolor{\statePredColor}{\langle b \rangle}\rangle\}}, $&$\tuplecolor{\specialPredColor}{1}\rangle$ \\\hline
		$\langle\tuplecolor{black}{\{\langle\tuplecolor{\inputPredColor}{\emptyset}, \tuplecolor{\outputPredColor}{\emptyset}, \tuplecolor{\statePredColor}{\langle \rangle}\rangle}, \langle\tuplecolor{\inputPredColor}{\{b\}}, \tuplecolor{\outputPredColor}{\emptyset}, \tuplecolor{\statePredColor}{\langle b\rangle}\rangle\},$&{$\tuplecolor{\specialPredColor}{1}\rangle$} \\\hline
		$\langle\tuplecolor{black}{\{\langle\tuplecolor{\inputPredColor}{\emptyset}, \tuplecolor{\outputPredColor}{\emptyset}, \tuplecolor{\statePredColor}{\langle \rangle}\rangle}, \langle\tuplecolor{\inputPredColor}{\{b\}}, \tuplecolor{\outputPredColor}{\{b\}}, \tuplecolor{\statePredColor}{\langle b\rangle}\rangle\}, $&{$ \tuplecolor{\specialPredColor}{0}\rangle$} \\\hline
		$\langle\tuplecolor{black}{\{\langle\tuplecolor{\inputPredColor}{\{b\}}, \tuplecolor{\outputPredColor}{\emptyset}, \tuplecolor{\statePredColor}{\langle b \rangle}\rangle}, \langle\tuplecolor{\inputPredColor}{\{b\}}, \tuplecolor{\outputPredColor}{\{b\}}, \tuplecolor{\statePredColor}{\langle b\rangle}\rangle\}, $& {$\tuplecolor{\specialPredColor}{1}\rangle$} \\\hline
		$\langle\tuplecolor{black}{\{\langle\tuplecolor{\inputPredColor}{\emptyset}, \tuplecolor{\outputPredColor}{\emptyset}, \tuplecolor{\statePredColor}{\langle \rangle}\rangle, \langle\tuplecolor{\inputPredColor}{\{b\}}, \tuplecolor{\outputPredColor}{\emptyset}, \tuplecolor{\statePredColor}{\langle b \rangle}\rangle}, $&\multirow{2}{*}{$\tuplecolor{\specialPredColor}{0}\rangle$} \\
		\hspace{1.05em}$\tuplecolor{black}{\langle\tuplecolor{\inputPredColor}{\{b\}}, \tuplecolor{\outputPredColor}{\{b\}}, \tuplecolor{\statePredColor}{\langle b\rangle}\rangle\}}, $ 
	\end{tabular}%
};
\node (joinrbrt) [stdnodenum,left=-0.45em of joinrrt] {%
	\begin{tabular}{c}
		\multirow{1}{*}{$i$}\\
		\hline\hline
		$1$ \\\hline
		$2$ \\\hline
		$3$ \\\hline
		{$4$}\\\hline
		{$5$}\\\hline
		{$6$}\\\hline
		\multirow{2}{*}{$7$}\\\\
	\end{tabular}%
};
\node (joinr) [stdnodetable,right=6.5em of i13, yshift=-7em, label={[tdlabel, xshift=-1.1em,yshift=+0.1em]above right:$\pi_{{10}}$}] {%
	\begin{tabular}{l@{\hspace{0.0em}}l@{\hspace{0.0em}}r}%
		\multicolumn{1}{l}{$\langle \tuplecolor{black}{\nu_{10.i}}, $}&\multicolumn{1}{r}{$\tuplecolor{\specialPredColor}{c_{10.i}} \rangle$}\\
		\hline\hline
		$\langle\tuplecolor{black}{\{\langle\tuplecolor{\inputPredColor}{\{d\}}, \tuplecolor{\outputPredColor}{\{d\}}, \tuplecolor{\statePredColor}{\langle d \rangle}\rangle\}}, $&$\tuplecolor{\specialPredColor}{1}\rangle$ \\\hline
		$\langle\tuplecolor{black}{\{\langle\tuplecolor{\inputPredColor}{\{b,d\}}, \tuplecolor{\outputPredColor}{\{d\}}, \tuplecolor{\statePredColor}{\langle b,d \rangle}\rangle\}}, $&$\tuplecolor{\specialPredColor}{1}\rangle$ \\\hline
		$\langle\tuplecolor{black}{\{\langle\tuplecolor{\inputPredColor}{\{d\}}, \tuplecolor{\outputPredColor}{\{d\}}, \tuplecolor{\statePredColor}{\langle d\rangle}\rangle}, $& \multirow{2}{*}{$\tuplecolor{\specialPredColor}{1}\rangle$} \\ 
		$\hspace{1.05em}\tuplecolor{black}{\langle\tuplecolor{\inputPredColor}{\{b,d\}}, \tuplecolor{\outputPredColor}{\{d\}}, \tuplecolor{\statePredColor}{\langle b,d\rangle}\rangle\},}$\\\specialrule{.1em}{.05em}{.05em}	
		$\langle\tuplecolor{black}{\{\langle\tuplecolor{\inputPredColor}{\{b,e\}}, \tuplecolor{\outputPredColor}{\{e\}}, \tuplecolor{\statePredColor}{\langle b,e\rangle}\rangle\}}, $&$\tuplecolor{\specialPredColor}{1}\rangle$ \\\hline
		$\langle\tuplecolor{black}{\{\langle\tuplecolor{\inputPredColor}{\{b,e\}}, \tuplecolor{\outputPredColor}{\{b,e\}}, \tuplecolor{\statePredColor}{\langle e,b\rangle}\rangle\}}, $&$\tuplecolor{\specialPredColor}{1}\rangle$ \\\hline
		$\langle\tuplecolor{black}{\{\langle\tuplecolor{\inputPredColor}{\{b,e\}}, \tuplecolor{\outputPredColor}{\{e\}}, \tuplecolor{\statePredColor}{\langle b,e\rangle}\rangle}, $& \multirow{2}{*}{$\tuplecolor{\specialPredColor}{1}\rangle$} \\ 
		$\hspace{1.05em}\tuplecolor{black}{\langle\tuplecolor{\inputPredColor}{\{b,e\}}, \tuplecolor{\outputPredColor}{\{b,e\}}, \tuplecolor{\statePredColor}{\langle e,b\rangle}\rangle\},}$\\\specialrule{.1em}{.05em}{.05em}	
		$\langle\tuplecolor{black}{\{\langle\tuplecolor{\inputPredColor}{\{d,e\}}, \tuplecolor{\outputPredColor}{\{d,e\}}, \tuplecolor{\statePredColor}{\langle d,e\rangle}\rangle\}}, $&$\tuplecolor{\specialPredColor}{1}\rangle$ \\\hline
		$\langle\tuplecolor{black}{\{\langle\tuplecolor{\inputPredColor}{\{d,e\}}, \tuplecolor{\outputPredColor}{\{d,e\}}, \tuplecolor{\statePredColor}{\langle e,d\rangle}\rangle\}}, $&$\tuplecolor{\specialPredColor}{1}\rangle$ \\\hline
		$\langle\tuplecolor{black}{\{\langle\tuplecolor{\inputPredColor}{\{d,e\}}, \tuplecolor{\outputPredColor}{\{d,e\}}, \tuplecolor{\statePredColor}{\langle d,e\rangle}\rangle}, $& \multirow{2}{*}{$\tuplecolor{\specialPredColor}{1}\rangle$} \\ 
		$\hspace{1.05em}\tuplecolor{black}{\langle\tuplecolor{\inputPredColor}{\{d,e\}}, \tuplecolor{\outputPredColor}{\{d,e\}}, \tuplecolor{\statePredColor}{\langle e,d\rangle}\rangle\},}$ 
	\end{tabular}%
};
\node (joinr2) [stdnodetable,right=0.5em of joinr, yshift=0em, label={[tdlabel, xshift=-1.1em,yshift=+0.1em]above right:$\pi_{{11}}$}] {%
	\begin{tabular}{l@{\hspace{0.0em}}l@{\hspace{0.0em}}r}%
		\multicolumn{1}{l}{$\langle \tuplecolor{black}{\nu_{11.i}}, $}&\multicolumn{1}{r}{$\tuplecolor{\specialPredColor}{c_{11.i}} \rangle$}\\
		\hline\hline
		$\langle\tuplecolor{black}{\{\langle\tuplecolor{\inputPredColor}{\{d\}}, \tuplecolor{\outputPredColor}{\{d\}}, \tuplecolor{\statePredColor}{\langle d \rangle}\rangle\}}, $&$\tuplecolor{\specialPredColor}{2}\rangle$ \\\hline
		$\langle\tuplecolor{black}{\{\langle\tuplecolor{\inputPredColor}{\{b,d\}}, \tuplecolor{\outputPredColor}{\{d\}}, \tuplecolor{\statePredColor}{\langle b,d \rangle}\rangle\}}, $&$\tuplecolor{\specialPredColor}{1}\rangle$ \\\hline
		$\langle\tuplecolor{black}{\{\langle\tuplecolor{\inputPredColor}{\{d\}}, \tuplecolor{\outputPredColor}{\{d\}}, \tuplecolor{\statePredColor}{\langle d\rangle}\rangle}, $& \multirow{2}{*}{$\tuplecolor{\specialPredColor}{1}\rangle$} \\ 
		$\hspace{1.05em}\tuplecolor{black}{\langle\tuplecolor{\inputPredColor}{\{b,d\}}, \tuplecolor{\outputPredColor}{\{d\}}, \tuplecolor{\statePredColor}{\langle b,d\rangle}\rangle\},}$\\\specialrule{.1em}{.05em}{.05em}	
		$\langle\tuplecolor{black}{\{\langle\tuplecolor{\inputPredColor}{\{b\}}, \tuplecolor{\outputPredColor}{\emptyset}, \tuplecolor{\statePredColor}{\langle b\rangle}\rangle\}}, $&$\tuplecolor{\specialPredColor}{1}\rangle$ \\\hline
		$\langle\tuplecolor{black}{\{\langle\tuplecolor{\inputPredColor}{\{b\}}, \tuplecolor{\outputPredColor}{\{b\}}, \tuplecolor{\statePredColor}{\langle b\rangle}\rangle\}}, $&$\tuplecolor{\specialPredColor}{1}\rangle$ \\\hline
		$\langle\tuplecolor{black}{\{\langle\tuplecolor{\inputPredColor}{\{b\}}, \tuplecolor{\outputPredColor}{\emptyset}, \tuplecolor{\statePredColor}{\langle b\rangle}\rangle}, $& \multirow{2}{*}{$\tuplecolor{\specialPredColor}{1}\rangle$} \\ 
		$\hspace{1.05em}\tuplecolor{black}{\langle\tuplecolor{\inputPredColor}{\{b\}}, \tuplecolor{\outputPredColor}{\{b\}}, \tuplecolor{\statePredColor}{\langle b\rangle}\rangle\},}$\\\specialrule{.1em}{.05em}{.05em}	
		$\langle\tuplecolor{black}{\{\langle\tuplecolor{\inputPredColor}{\{d,e\}}, \tuplecolor{\outputPredColor}{\{d,e\}}, \tuplecolor{\statePredColor}{\langle d,e\rangle}\rangle\}}, $&$\tuplecolor{\specialPredColor}{1}\rangle$
	\end{tabular}%
};
\node (joinrb2) [stdnodenum,right=-0.45em of joinr2] {%
	\begin{tabular}{c}
		\multirow{1}{*}{$i$}\\
		\hline\hline
		$1$ \\\hline
		$2$ \\\hline
		\multirow{2}{*}{$3$} \\\\\specialrule{.1em}{.05em}{.05em}	
		$4$ \\\hline
		$5$ \\\hline
		\multirow{2}{*}{$6$} \\\\\specialrule{.1em}{.05em}{.05em}	
		$7$
	\end{tabular}%
};
\node (joinrb) [stdnodenum,left=-0.45em of joinr] {%
	\begin{tabular}{c}
		\multirow{1}{*}{$i$}\\
		\hline\hline
		$1$ \\\hline
		$2$ \\\hline
		\multirow{2}{*}{$3$} \\\\\specialrule{.1em}{.05em}{.05em}	
		$4$ \\\hline
		$5$ \\\hline
		\multirow{2}{*}{$6$} \\\\\specialrule{.1em}{.05em}{.05em}	
		$7$ \\\hline
		$8$ \\\hline
		\multirow{2}{*}{$9$} \\\\
	\end{tabular}%
};
\coordinate (top) at ($ (leaf2.north east)+(0.6em,-0.5em) $);
\coordinate (bot) at ($ (top)+(0,-12.9em) $);

\draw [<-] (j) to (rt);
\draw [->] (j) to ($ (i13.north)$);
\draw [->] (j) to ($ (r22.north)$);
\draw [->](r2) to (i22);
\draw [->](r22) to (r2);
\draw [<-](i2) to (i22);
\draw [<-](l2) to (i2);
\draw [<-](l1) to (i1);
\draw [->](i12) to (i1);
\draw [->](i13) to (i12);
\draw [<-](bc) to (l2);
\draw [<-](cc) to (bc);
\draw [<-](lrx) to (cc);

\draw [dashed, bend right=40] (leaf0r) to (rt);
\draw [dashed, bend left=15] (joinrrt) to (r22);
\draw [dashed] (j) to (join);
\draw [dashed, bend right=15] (i2) to (leaf2);
\draw [dashed, bend left=5] (i12) to (leaf1);
\draw [dashed, bend left=22] (leaf0) to (l1);
\draw [dashed, bend right=1] (leaf0x) to (i13);
\draw [dashed, bend right=45] (joinr) to (i22);
\end{tikzpicture}
\caption{Selected tables of~$\pi$ obtained by~$\dpa_{\algo{PROJ}}$ on
  TD~${\cal T}$ and purged table mapping~$\nu$ (obtained by purging on~$\tau$, c.f, 
  Figure~\ref{fig:running2}).} 
\label{fig:running3}
\end{figure*}

\begin{example} 
  Recall instance~$(\prog,P)$, TD~$\TTT$, and tables~$\tab{1}$,
  $\ldots$, $\tab{14}$ from Examples~\ref{ex:running0}, \ref{ex:sat},
  and Figure~\ref{fig:running2}. Figure~\ref{fig:running3} depicts
  selected tables of~$\pi_1, \ldots, \pi_{14}$ obtained after
  running~$\dpa_\PROJ$ for counting projected answer sets.
  We assume that row $i$ in table $\pi_t$ corresponds to
  $\vec{v_{t.i}} = \langle \rho_{t.i}, c_{t.i} \rangle$
  where~$\rho_{t.i}\subseteq\nu(t)$.
  Recall that for some nodes~$t$, there are rows among
  different~$\PRIM$-tables that are removed (highlighted gray in Figure~\ref{fig:running2}) during purging. 
  By purging we avoid to correct stored counters (backtracking)
  whenever a row has no ``succeeding'' row in the
  parent table.
  
  Next, we discuss selected rows obtained by
  $\dpa_\PROJ((\prog,P),(T,\chi,\nu))$. Tables $\pi_1$, $\ldots$,
  $\pi_{14}$ are shown in Figure~\ref{fig:running3}.
  Since~$\type(t_1)= \leaf$, we have
  $\pi_1=\langle\{\langle \emptyset , \emptyset, \langle \rangle
  \rangle \}, 1\rangle$.  Intuitively, at~$t_1$ the
  row~$\langle\emptyset, \emptyset, \langle\rangle\rangle$ belongs to~$1$ bucket.
  Node~$t_2$ introduces atom~$a$, which results in
  table~$\pi_2\eqdef\big\SB\langle \{\vec{u_{2.1}}\},
  1\rangle, \langle \{\vec{u_{2.2}}\},
  1\rangle, \langle \{\vec{u_{2.1}}, \vec{u_{2.2}}\},
  1\rangle\big\SE$, where~$\vec{u_{2.1}}=\langle \emptyset, \emptyset, \langle \rangle\rangle$ and~$\vec{u_{2.2}}=\langle \{a\}, \emptyset, \langle a\rangle \rangle$ 
  (derived similarly to table~$\pi_{4}$ as in Example~\ref{ex:pcnt}). 
  %
  Node~$t_{10}$ introduces projected atom~$e$, and
  node~$t_{11}$ removes~$e$.  
  %
  %
  For row~$\vec{v_{11.1}}$ we compute the
  count~$\ipmc(t_{11},\{\vec{\tabval_{11.1}}\},
  \langle\pi_{10}\rangle)$ by means of~$\pcnt$. Therefore, take
  for~$\varphi$ the singleton set~$\{\vec{\tabval_{11.1}}\}$.
  We simply have
  $\ipmc(t_{11},\{\vec{\tabval_{11.1}}\}, \langle\pi_{10}\rangle) =
  \pmc(t_{11},\{\vec{\tabval_{11.1}}\}, \langle\pi_{10}\rangle)$.  To
  compute
  $\pmc(t_{11},\{\vec{\tabval_{11.1}}\}, \langle\pi_{10}\rangle)$, we
  take for~$O$ the sets~$\{\vec{u_{10.1}}\}$, $\{\vec{u_{10.6}}\}$,
  $\{\vec{u_{10.7}}\}$, and~$\{\vec{u_{10.6}}, \vec{u_{10.7}}\}$ into
  account, since all other non-empty subsets of origins
  of~$\vec{\tabval_{11.1}}$ in~$\nu_{10}$ do not occur in~$\pi_{10}$.
  Then, we take the sum over the values
  $\sipmc(\langle \pi_{10}\rangle,\{\vec{\tabval_{10.1}}\})=1$,
  $\sipmc(\langle \pi_{10}\rangle,\{\vec{\tabval_{10.6}}\})=1$,
  $\sipmc(\langle \pi_{10}\rangle,\{\vec{\tabval_{10.7}}\})=1$ and
  subtract
  $\sipmc(\langle \pi_{10}\rangle,\{\vec{\tabval_{10.6}},
  \vec{\tabval_{10.7}}\})=1$. This results
  in~$\pmc(t_{11},\{\vec{\tabval_{11.1}}\}, \langle\pi_{10}\rangle) =
  c_{10.1} + c_{10.7}\; + $ $ c_{10.8} - c_{10.9} = 2$. We proceed similarly
  for row~$v_{11.2}$, resulting in~$c_{11.2}=1$.
  Then for row~$v_{11.3}$,
  $\ipmc(t_{11},\{\vec{\tabval_{11.1}},\vec{\tabval_{11.6}}\}, \langle\pi_{10}\rangle) = | %
  \pmc(t_{11},\{\vec{\tabval_{11.1}},\vec{\tabval_{11.6}}\}, \langle\pi_{10}\rangle) - \ipmc(t_{11},\{\vec{\tabval_{11.1}}\}, \langle\pi_{10}\rangle)$ $- \ipmc(t_{11},\{\vec{\tabval_{11.6}}\}, \langle\pi_{10}\rangle) | = \Card{2-c_{11.1}-c_{11.2}}= \Card{2 -2 - 1} =\Card{-1} = 1 = c_{11.3}$.
  Hence, $c_{11.3} = 1$ represents the number of projected answer sets,
  both rows~$\vec{u_{11.1}}$ and~$\vec{u_{11.6}}$ have in common. We
  then use it for table~$t_{12}$.  Node~$t_{12}$ removes projection
  atom~$d$.  For node~$t_{13}$ where $\type(t_{13}) = \join$ one
  multiplies stored $\sipmc$ values for \AlgA-rows in the two children
  of~$t_{13}$ accordingly.  In the end, the projected answer sets count
  of~$\prog$ corresponds to~$\sipmc(\langle\pi_{14}\rangle,\vec{u_{14.1}})=3$.
\end{example}

\subsection{Runtime Analysis and Correctness}
Next, we present asymptotic upper bounds on the runtime of our
Algorithm~$\dpa_{\PROJ}$.  
%
%
%
%
We assume~$\gamma(n)$ to be the number of operations that are required
to multiply two~$n$-bit integers, which can be achieved in time
$n\cdot log\, n \cdot log\, log\,n$~\cite{Knuth1998,Harvey2016}.  
Often even constant-time multiplication is assumed.




\begin{theorem}
  \label{thm:runtime}
  Given a \PASP instance~$(\prog,P)$ and a tabled tree decomposition
  $\TTT_{\text{purged}} = (T,\chi,\nu)$ of~$\mathcal{G}_\prog$ of width~$k$ with $g$
  nodes. Then, $\dpa_{\PROJ}$ runs in time
  $\mathcal{O}(2^{4m}\cdot g \cdot \gamma(\CCard{\prog}))$
  where~$m\eqdef \max(\{\nu(t) \mid t\in N\})$.
\end{theorem}
\begin{proof}
  Let~$d = k+1$ be maximum bag size of the TD~$\TTT$. For each
  node~$t$ of $T$, we consider the table $\nu(t)$ of $\TTT_{\text{purged}}$.
  %
  Let TDD~$(T,\chi,\pi)$ be the output of~$\dpa_\PROJ$. In worst case,
  we store in~$\pi(t)$ each subset~$\rho \subseteq \nu(t)$ together
  with exactly one counter. Hence, we have at most $2^{m}$ many rows
  in $\rho$.
  %
  %
  In order to compute $\ipmc$ for~$\rho$, we consider every
  subset~$\varphi \subseteq \rho$ and compute~$\pcnt$. Since
  $\Card{\rho}\leq m$, we have at most~$2^{m}$ many subsets $\varphi$
  of $\rho$. Finally, for computing $\pcnt$, we consider in the worst
  case each subset of the origins of~$\varphi$ for each child table,
  which are at most~$2^{m}\cdot 2^{m}$ because of nodes~$t$
  with~$\type(t)=\join$.
  %
  %
  In total, we obtain a runtime bound
  of~$\bigO{2^{m} \cdot 2^{m} \cdot 2^{m}\cdot 2^{m} \cdot
    \gamma(\CCard{\prog})} \subseteq \bigO{2^{4m} \cdot
    \gamma(\CCard{\prog}})$ due to multiplication of two $n$-bit
  integers for nodes~$t$ with~$\type(t)=\join$ at costs~$\gamma(n)$.
  %
  %
  Then, we apply this to every node of~$T$ 
  resulting in
  runtime~$\bigO{2^{4m} \cdot g \cdot \gamma(\CCard{\prog})}$.
  %
\end{proof}

\begin{corollary}\label{cor:runtime}
  Given an instance $(\prog,P)$ of \PASP where $\prog$ is
  head-cycle-free and has treewidth~$k$. Then, $\mdpa{\PRIM}$ runs in
  time~$\mathcal{O}(2^{3^{k+1.27}\cdot k!}\cdot \CCard{\prog}\cdot
  \gamma(\CCard{\prog}))$.
\end{corollary}
\begin{proof}
  We can compute in time~$2^{\mathcal{O}(k^3)}\cdot\CCard{\mathcal{G}_\prog}$ a
  TD~${\cal T'}$ with~$g\leq \CCard{\prog}$ nodes of width at
  most~$k$~\cite{Bodlaender96}. Then, we can simply
  run~$\dpa_{\PRIM}$, which runs in
  time~$\mathcal{O}({3^{k}\cdot k!}\cdot \CCard{\prog})$ by
  Theorem~\ref{thm:primruntime} and since the number of nodes of a
  tree decomposition is linear in the size of the input
  instance~\cite{Bodlaender96}. 
  Then, we again traverse the TD for purging and output
  $\TTT_{\text{purged}}$, which runs in time single exponential of the
  treewidth and linear of the instance size. Finally, we run
  $\dpa_{\PROJ}$ and obtain by Theorem~\ref{thm:runtime} that the
  runtime bound
  $\mathcal{O}(2^{4\cdot3^{k}\cdot k!}\cdot \CCard{\prog}\cdot
  \gamma(\CCard{\prog})) \subseteq $
  $\mathcal{O}(2^{3^{k + 1.27}\cdot k!}\cdot \CCard{\prog}\cdot
  \gamma(\CCard{\prog}))$.  %
  Hence, the corollary holds.
\end{proof}


The next result establishes lower bounds. 

\begin{theorem}
  Unless ETH fails, $\PASP$ cannot be solved in
  time~$2^{2^{o(k)}}\cdot \CCard{\prog}^{o(k)}$ for a given instance
  $(\prog,P)$ where~$k$ is the treewidth of the primal graph of~$\prog$.
\end{theorem}
\begin{proof}
  Assume for proof by contradiction that there is such an algorithm.
  %
  %
  We show that this contradicts a very recent
  result~\cite{LampisMitsou17,FichteEtAl18}, which states that one
  cannot decide the validity of a QBF
  $\forall{V_1}.\exists V_2.E$ in
  time~$2^{2^{o(k)}}\cdot \CCard{E}^{o(k)}$,
  where 
  $E$ is in CNF.
  %
  Let $(\forall{V_1}.\exists V_2.E,k)$ be an instance
  of~$\forall\exists$-\SAT parameterized by the treewidth~$k$. Then,
  we reduce to an instance~$((\prog,P),2k)$ of the decision
  version~$\PASP$-exactly-$2^{\Card{V_1}}$ when parameterized by
  treewidth of~$\mathcal{G}_\prog$ such that $P=V_1$, the number of solutions is
  exactly~$2^{\Card{V_1}}$, and~$\prog$ is as follows.  For
  each~$v\in V_1 \cup V_2$, program~$\prog$ contains 
  rule~$v \lor nv \hsep$.
  Each clause~$x_1, \ldots, x_i, \neg x_{i+1}, \ldots, \neg x_j$
  results in one additional
  rule~$\hsep \neg x_1,\ldots, \neg x_i, x_{i+1}, \ldots, x_{j}$.
  It is easy to see that the reduction is correct
  and therefore instance~$((\prog,P), 2k)$ is
  a yes instance of 
  $\PASP$-exactly-$2^{\Card{V_1}}$ 
  if and only if~$(\forall{V_1}.\exists V_2.E,k)$
  is a yes instance of 
  problem~$\forall\exists$-\SAT. 
  %
  %
  %
  In fact, the reduction is also an fpl-reduction, since the treewidth
  of $\prog$ at most doubles due to duplication of atoms.
  Note that we require an \emph{fpl}-reduction here, as results do not
  carry over from simple fpt-reductions.
  This concludes the proof and establishes the theorem.
\end{proof}

\longversion{
\begin{corollary}
  Unless ETH fails, $\PASP$ cannot be solved in
  time~$2^{2^{o(k)}}\cdot \CCard{\prog}^{o(k)}$ for a given instance
  $(\prog,P)$ where~$k$ is the treewidth of the incidence graph of~$\prog$.
\end{corollary}
\begin{proof}
  Let $w_i$ and $w_p$ be the treewidth of the incidence graph and
  primal graph of~$\prog$, respectively. Then,
  $w_i \leq w_p +1$~\cite{SamerSzeider10b}, which establishes the
  claim.
\end{proof}

\begin{corollary}
  Given an instance $(\prog,P)$ of \PASP where $\prog$ has treewidth~$k$. Then,
  Algorithm~$\mdpa{\AlgA}$ runs in
  time~$2^{2^{\Theta(k)}} \cdot \CCard{\prog}^c$ for some positive
  integer~$c$.
\end{corollary}}

Finally, we state that indeed~$\mdpa{\PRIM}$ gives the projected answer sets count of a given head-cycle-free program~$\prog$.

\begin{proposition}[$\star$]\label{prop:phcworks}
  Algorithm $\mdpa{\PRIM}$ is correct and outputs for any instance
  of \PASP its projected answer sets count.
\end{proposition}
\begin{proof}
Soundness follows by establishing an invariant for any row of~$\pi(t)$ guaranteeing that the values of~$\ipmc$ indeed capture ``all-overlapping'' counts of~$\progt{t}$. One can show that the invariant is a consequence of the properties of~\PRIM and the additional ``purging'' step, which neither destroys soundness nor completeness of~$\dpa_\PRIM$. Further, completeness guarantees that indeed all the required rows are computed.
\end{proof}

\subsection{Solving \PDASP for Disjunctive Programs}

In this section, we extend our algorithm to solve the projected answer
set counting problem (\PDASP) for disjunctive programs. Therefore, we
simply use a local algorithm \algo{PRIM} for disjunctive ASP that was
introduced in the
literature~\cite{FichteEtAl17a,JaklPichlerWoltran09}.
%
%
%
Recall algorithm~\mdpa{\AlgA} illustrated in
Figure~\ref{fig:multiarch}.
First, we 
construct a graph representation and heuristically compute a tree
decomposition of this graph. Then, we run $\dpa_\algo{PRIM}$ as first
traversal resulting in TTD~$(T,\chi,\tau)$. Next, we purge rows
of~$\tau$, which can not be extended to an answer set resulting in
TTD~$(T,\chi,\nu)$. Finally, we compute the projected answer sets count
by~$\dpa_{\PROJ}$ and obtain TTD~$(T,\chi,\pi)$.

\begin{proposition}[$\star$]\label{prop:disjworks}
  $\mdpa{\algo{PRIM}}$ is correct, i.e., it outputs the projected answer sets count for any instance
  of \PDASP.
\end{proposition}

The following corollary states the runtime results.

\begin{corollary}\label{cor:disjruntime}
  Given an instance $(\prog,P)$ of \PDASP where $\prog$ is a
  disjunctive program of treewidth~$k$. Then, $\mdpa{\algo{PRIM}}$
  runs in
  time~$\mathcal{O}(2^{2^{2^{k+3}}}\cdot \CCard{\prog}\cdot
  \gamma(\CCard{\prog}))$.
\end{corollary}
\begin{proof}
  The first two steps follow the proof of Corollary~\ref{cor:runtime}.
  However, $\dpa_{\algo{PRIM}}$ runs in
  time~$\mathcal{O}(2^{2^{k+2}}\cdot
  \CCard{\prog})$~\cite{FichteEtAl17a}. Finally, we run $\dpa_{\PROJ}$
  and obtain by Theorem~\ref{thm:runtime} that
  $\mathcal{O}(2^{4\cdot2^{2^{k+2}}}\cdot \CCard{\prog}\cdot
  \gamma(\CCard{\prog})) \subseteq $
  $\mathcal{O}(2^{2^{2^{k+3}}}\cdot \CCard{\prog}\cdot
  \gamma(\CCard{\prog}))$.  
  %
  %
\end{proof}

%
%





Again, we are interested in whether we can improve the algorithm
significantly. While we obtain lower bounds from the ETH for~$\SAT$
(single-exponential) and
for~$\forall\exists$-\SAT/$\exists\forall$-\SAT (double-exponential),
to our knowledge it is unproven whether this extends to
$\forall\exists\forall$-\SAT and~$\exists\forall\exists$-\SAT
(triple-exponential).
Since it was anticipated by~\citex{MarxMitsou16} that it follows
just by assuming ETH, we state this as hypothesis. In particular, they
claimed that alternating quantifier alternations are the reason for
large dependence on treewidth. However, the proofs can be quite
involved, trading an additional alternation for exponential
compression.

\begin{hypothesis}\label{hyp:lampis3}
  The $\forall\exists\forall$-\SAT problem for a QBF~$Q$ in DNF of
  treewidth~$k$ can not be decided in
  time~${2^{2^{2^{o(k)}}}}\cdot \CCard{Q}^{o(k)}$.
\end{hypothesis}



%

\longversion{\begin{proposition}
  Unless ETH fails, QBFs of the form $\exists V_1.\forall V_2.\cdots\forall V_\ell. E$
where~$k$ is the treewidth of the primal graph of DNF formula~$E$, can not be solved in time~$2^{2^{\dots^{2^{o(k)}}}}\cdot \CCard{\prog}^{o(k)}$,
where the height of the tower is~$\ell$.
\end{proposition}

\begin{proof}[Idea]
	Follows by construction defined in the proof of~\cite{LampisMitsou17} for QBFs of the form~$\exists V_1.\forall V_2. E$. \todo{provide rigorous proof?}
\end{proof}

\begin{corollary}
Unless ETH fails, QBFs of the form $\forall V_1.\exists V_2.\cdots\forall V_\ell. E$
where~$k$ is the treewidth of the primal graph of CNF formula~$E$, can not be solved in time~$2^{2^{\dots^{2^{o(k)}}}}\cdot \CCard{\prog}^{o(k)}$,
where the height of the tower is~$\ell$.
\end{corollary}}

\begin{theorem}\label{thm:lowerbound_disj}
  Unless Hypothesis~\ref{hyp:lampis3} fails, \PDASP for disjunctive programs~$\prog$ cannot be
  solved in time~$2^{2^{2^{o(k)}}} \cdot \CCard{\prog}^{o(k)}$ for
  given instance~$(\prog, P)$ of treewidth~$k$.
\end{theorem}
\begin{proof}
  Assume for proof by contradiction that there is such an algorithm.
  %
  %
  We show that this contradicts Hypothesis~\ref{hyp:lampis3},~i.e.,
  we cannot decide the validity of a QBF 
  %
  $Q=\forall{V_1}.\exists V_2.\forall V_3.E$ in
  time~$2^{2^{2^{o(k)}}}\cdot \CCard{E}^{o(k)}$
  %
  where 
  $E$ is in DNF. 
  %
  %
  %
  Assume we have
  %
  given such an instance when parameterized by the treewidth~$k$.
  In the following, we employ a well-known
  reduction~$R$~\cite{EiterGottlob95}, which
  transforms~$\exists V_2.\forall V_3. E$
  into~$\prog=R(\exists V_2.\forall V_3. E)$ and gives a yes
  instance~$\prog$ of consistency if and only
  if~$\exists V_2. \forall V_3. E$ is a yes instance of
  $\exists\forall$-\SAT.
  Then, we reduce instance~$(Q,k)$ via a reduction~$S$ to an
  instance~$((\prog',V_1),2k+2)$, where $\prog'=R(\exists V_2'.\forall V_3. E)$, $V_2'\eqdef V_1\cup V_2$, of the decision
  version~$\PDASP$-exactly-$2^{\Card{V_1}}$ of~$\PDASP$ when
  parameterized by treewidth such that the number of projected answer
  sets is
  exactly~$2^{\Card{V_1}}$.
  %
  %
  It is easy to see that reduction~$S$
  %
  %
  gives a yes instance~$(\prog',V_1)$
  of~$\PDASP$-exactly-$2^{\Card{V_1}}$ if and only
  if~$\forall V_1.\exists V_2. \forall V_3. E$ is a yes instance
  of~$\forall\exists\forall$-\SAT.
  However, it remains to show that the reduction~$S$ indeed increases
  the treewidth only linearly.
  %
  %
  %
  %
  Therefore, let $\TTT=(T,\chi)$ be TD of~$E$. We transform~$\TTT$
  into a TD~$\TTT'=(T,\chi')$ of~$\mathcal{G}_{\prog'}$ as follows.  For each
  bag~$\chi(t)$ of~$\TTT$, we add vertices for the atoms~$w$ and $w'$
  (two additional atoms introduced in reduction~$R$) and in addition
  we duplicate each vertex~$v$ in~$\chi(t)$ (due to corresponding duplicate
  atoms introduced in reduction~$R$). Observe
  that~$\width(\TTT') \leq 2\cdot \width(\TTT) + 2$. By construction
  of~$R$, $\TTT'$ is then a TD of~$\mathcal{G}_{\prog'}$.
  Hence, $S$ is also an fpl-reduction.
\end{proof}

\longversion{
\begin{corollary}
Unless ETH fails, \PDASP for disjunctive programs~$\prog$ cannot be solved in time~$2^{2^{2^{o(k)}}} \cdot \CCard{\prog}^{o(k)}$ for given instance
~$(\prog, P)$ where~$k$ is the treewidth of the incidence graph of~$\prog$.
\end{corollary}}

Then, the runtime of algorithm~$\mdpa{\algo{PRIM}}$ is asymptotically
worst-case optimal, depending on multiplication costs~$\gamma(n)$. 
%
\longversion{
In fact, we can conclude the following
corollary, which renders algorithm~$\mdpa{\algo{PRIM}}$ asymptotically
worst-case optimal, depending on the costs~$\gamma(n)$ for multiplying
two~$n$-bit numbers.

\begin{corollary}
Unless ETH fails, \PDASP for disjunctive programs~$\prog$ runs in time~$2^{2^{2^{\Theta(k)}}} \cdot \CCard{\prog} \cdot \gamma(\CCard{\prog})$ for given instance
~$(\prog, P)$ where~$k$ is the treewidth of the primal graph of~$\prog$.
\end{corollary}
\begin{proof}
Lower bounds by Theorem~\ref{thm:lowerbound_disj}. Upper bound by Algorithm~$\dpa_{\algo{PRIM}}$ in the first pass (c.f., Corollary~\ref{cor:disjruntime}),
followed by purging and the projection algorithm~$\dpa_\PROJ$ in the second pass.
\end{proof}
}
%

%
%
%

\section{Conclusions}\label{sec:conclusions}

In the light of very recent works~\cite{EibenEtAl19,ijcai,HecherMorakWoltran20},
which provide methods on dealing with high treewidth for \SAT, QBF and other formalisms,
it seems that the full potential of treewidth has not been unleashed for ASP yet.
Towards making these advancements accessible for ASP, we present novel reductions for the important fragments of normal and HCF programs to \SAT.
We introduced novel algorithms to count the projected answer sets
(\PASP) of head-cycle-free or disjunctive programs. Our algorithms
employ dynamic programming and exploit small treewidth of the
primal graph of the input program. The second algorithm, which solves
arbitrary disjunctive programs, is expected asymptotically optimal
assuming the exponential time hypothesis (ETH).
More precisely, runtime is triple exponential in the treewidth and
polynomial in the size of the input instance. When we restrict the
input to head-cycle-free programs, the runtime drops to double
exponential.

Our results extend previous work to
answer set programming and we believe that it can be applicable to
other hard combinatorial problems, such as
circumscription~\cite{DurandHermannKolaitis05}, quantified Boolean
formulas (QBF)~\cite{CharwatWoltran16a}, or default
logic~\cite{FichteHecherSchindler18a}.
%

}
\appendix
\clearpage



\bibliography{references}

\futuresketch{
\clearpage
\section{Additional Resources}

\appendix

\subsection{Additional Examples}
\begin{example}[c.f.,\citey{FichteEtAl17a}]\label{ex:bagprog} 
  %
  Intuitively, the tree decomposition of Figure~\ref{fig:graph-td}
  enables us to evaluate program $\prog$ by analyzing sub-programs
  $\{r_2\}$ and $\{r_3,r_4, r_5\}$, and combining results agreeing on
  $e$ followed by analyzing~$\{r_1\}$.  Indeed, for the given tree
  decomposition of Figure~\ref{fig:graph-td}, $\progt{t_1}=\{r_2\}$,
  $\progt{t_2}=\{r_3,r_4, r_5\}$ and
  $\prog=\progt{t_3}=\{r_1\} \cup \progtneq{t_3}$. Note that
  here~$\prog=\progt{t_3} \neq \progtneq{t_3}$ and the tree
  decomposition is not nice.  \longversion{For the tree decomposition
    of Figure~\ref{fig:graph-td2}, we have
    $\progt{t_1} = \{r_1,r_2\}$, 
    as well as $\progt{t_3} = \{r_3\}$.} 
\end{example}%

%

\subsection{Worst-Case Analysis of $\dpa_{\PRIM}$: Omitted proofs}

\begin{restateproposition}[prop:kfact]
\begin{proposition}
Given any positive integer~$i \geq 1$ and functions~$f(k)\eqdef k!, g(k) \eqdef 2^{k^{(i+1)/i}}$. Then, $f \in O(g)$.
\end{proposition}
\end{restateproposition}
\begin{proof}
We proceed by simultaneous induction.\\
Base case ($k=i=1$): Obviously, $1^{2} \geq 1!$.\\
Induction hypothesis: $k! \in O(2^{k^{(i+1)/i}})$\\
Induction step ($k \rightarrow k+1$): \\We have to show that for $k\geq k_0$ for some fixed $k_0$, the following equation holds.
\begin{align*}
  2^{(k+1)^{(i+1)/i}} \geq (k+1)\cdot k!\\
  2^{(k+1)^{1/i}\cdot(k+1)} \geq (k+1)\cdot k!\\
  2^{(k+1)^{1/i}+k\cdot(k+1)^{1/i}} \geq (k+1)\cdot k!\\
  2^{(k+1)^{1/i}}\cdot 2^{k\cdot(k+1)^{1/i}} \geq (k+1)\cdot k!\\
  2^{(k+1)^{1/i}} \cdot k! \geq^{IH} (k+1)\cdot k!\\
  2^{(k+1)^{1/i}} \geq (k+1)\\
  2^{(k+1)^{1/i}} \geq 2^{\text{log}_2(k+1)}\geq (k+1)\\
  \text{ where } k\geq k_0 \text{ for some fixed } k_0 \text{ since } \text{log}_2\in O(\text{exp}(1/i))
\end{align*}
Induction step ($k \rightarrow k+1, i \rightarrow i+1$): Analogous, previous step works for any~$i$.\\
Induction step ($i \rightarrow i+1$): Analogous.
\end{proof}

\subsection{Characterizing Extensions}

In the following, we assume~$(\prog,P)$ to be an instance of~$\PASP$. 
Further, let~$\mathcal{T}=(T,\chi,\tau)$
be an~$\AlgA$-TTD of~$\mathcal{G}_\prog$ where~$T=(N,\cdot,n)$, node~$t\in N$, and~$\rho\subseteq\tau(t)$.

\begin{definition}\label{def:extensions}
  Let $\vec u$ be a row of $\rho$.

  An \emph{extension below~$t$} is a set of pairs where a pair consist
  of a node~$t'$ of the \emph{induced sub-tree~$T[t]$ rooted at~$t$} and a row~$\vec v$ of $\tau(t')$
  and the cardinality of the set equals the number of nodes in the
  sub-tree~$T[t]$. 
  
  We define the family of \emph{extensions below~$t$}
  recursively as follows.  If $t$ is of type~\leaf, then
  $\Ext_{\leq t}(\vec u) \eqdef \{\{\langle t,\vec u\rangle\}\}$;
  otherwise
  $\Ext_{\leq t}(\vec u) \eqdef \bigcup_{\vec v \in \origs(t,\vec u)}
  \big\SB\{\langle t,\vec u\rangle\}\cup X_1 \cup \ldots \cup X_\ell
  \SM X_i\in\Ext_{\leq t_i}({\vec v}_{(i)})\big\SE$ 
  for the~$\ell$ children~$t_1, \ldots, t_\ell$ of~$t$.
  %
  We extend this notation for an $\AlgS$-table~$\rho$ by
  $\Ext_{\leq t}(\rho)\eqdef \bigcup_{\vec u\in\rho} \Ext_{\leq
    t}(\vec u)$.  Further, we
  let~$\Exts \eqdef \Ext_{\leq n}(\tau(n))$ be the
  \emph{family of all extensions}. 
  
  Further, we define \emph{the local table for node}~$t$ and family~$E$ of extensions (below some node) as
  $\local(t,E)\eqdef \bigcup_{\hat\rho \in E}\{ \langle \vec{\tabval}\rangle \mid
  \langle t, \vec{\tabval}\rangle \in \hat{\rho}\}$.

\end{definition}

If we would construct all extensions below the root~$n$, it allows us
to also obtain all models of program~$\prog$.  To this end, we state the following definition.

\begin{definition}\label{def:satext}
  %
  We define 
  the \emph{satisfiable
    extensions below~$t$} for~$\rho$ by
  \[\PExt_{\leq t}(\rho)\eqdef \bigcup_{\vec u\in\rho} \SB X \SM X
    \in \Ext_{\leq t}(\vec u), X \subseteq Y, Y \in \Exts\SE.\]
\end{definition}

\begin{observation}
$\PExt_{\leq n}(\tau(n)) = \Exts$.
\end{observation}

\begin{definition}
We define the \emph{purged table mapping~$\nu$ of~$\tau$} by
$\nu(t)\eqdef \local(t,\PExt_{\leq t}[\tau(t)])$ for every~$t\in N$.
\end{definition}

Next, we define an auxiliary notation that gives us a way to
reconstruct interpretations from families of extensions.


\begin{definition}\label{def:iextensions}
  Let $E$ be a family of extensions
  below~$t$. 
  We define the \emph{set~$\mathcal{I}(E)$ of interpretations} of~$E$
  by
  $\mathcal{I}(E) \eqdef \big\SB \bigcup_{\langle \cdot, \vec u
    \rangle \in X} \mathcal{I}(\vec u) \mid X \in E \big\SE$
  and the set~$\mathcal{I}_P(E)$ of \emph{projected interpretations} by
  $\mathcal{I}_P(E) \eqdef \big\SB \bigcup_{\langle \cdot, \vec u \rangle \in X}
  \mathcal{I}(\vec u) \cap P \mid X \in E \big\SE$.

\end{definition}

\begin{example} 
  Consider again program~$\prog$ and TTD~$(T,\chi,\tau)$ of~$\mathcal{G}_\prog$,
  where~$t_{14}$ is the root of~$T$, from Example~\ref{ex:sat}.
  Let~$X=\{\langle t_{13}, \langle\{b\}, \{b\}, \langle b\rangle\rangle\rangle, \langle t_{12},
  \langle\{b\}, \emptyset, \langle b\rangle\rangle\rangle,
  \langle t_{11},
  \langle\{b\}, \emptyset, \langle b\rangle\rangle\rangle,$
  $\langle t_{10},
  \langle\{b,e\}, \{e\}, \langle b,e\rangle\rangle\rangle,
  \langle t_{9},
  \langle\{e\}, \{e\}, \langle e\rangle\rangle\rangle,
  \langle t_{4},
  \langle\{b\}, \{b\},$ 
  $\langle b\rangle\rangle\rangle,
  \langle t_{3},
  \langle\{b\}, \{b\}, \langle b\rangle\rangle\rangle,
  \langle t_{1},
  \langle\emptyset, \emptyset, \langle \rangle\rangle\rangle\}$
  be an extension
  below~$t_{14}$.  Observe that~$X\in\Exts$ and that
  Figure~\ref{fig:running2} highlights those rows of tables for
  nodes~$t_{13}, t_{12}, t_{11}, t_{10}, t_{9}, t_4, t_3$ and~$t_1$ that also occur in~$X$
  (in yellow). Further, $\mathcal{I}(\{X\})=\{b,e\}$ computes the
  corresponding answer set of~$X$, and $\mathcal{I}_P(\{X\}) = \{e\}$ derives
  the projected answer sets of~$X$.  $\mathcal{I}(\Exts)$ refers to the set
  of answer sets of~$\prog$, whereas~$\mathcal{I}_P(\Exts)$ is the set
  of projected answer sets of~$\prog$.
\end{example}

\subsection{Correctness of~$\dpa_{\PRIM}$: Omitted proofs}

In the following, we assume~$\prog$ to be a head-cycle-free program. Further, let~$\mathcal{T}=(T,\chi,\tau)$
be an~$\AlgA$-TTD of~$\mathcal{G}_\prog$ where~$T=(N,\cdot,n)$ and~$t\in N$ is a node.

We state definitions required for the correctness
proofs of our algorithm \PRIM. In the end, we only store rows that
are restricted to the bag content to maintain runtime bounds. 
Similar to related work~\cite{FichteEtAl17a}, we define the
content of our tables in two steps. First, we define the properties of
so-called \emph{$\PRIM$-solutions up to~$t$}. Second, we restrict
these solutions to~\emph{$\PRIM$-row solutions} at~$t$.

\begin{definition}\label{def:globalhcf}
Let~$\hat I\subseteq\att{t}$ be an interpretation,
$\hat{\mathcal{P}}\subseteq \hat I$ be a set of atoms and~$\hat\sigma$ be an ordering over atoms~$\hat I$.
Then, $\langle \hat I, \hat{\mathcal{P}}, \hat\sigma\rangle$ is referred to as~\emph{$\PRIM$-solution up to~$t$} if the following holds.
  \begin{enumerate}
    \item~$\hat I\models\progt{t}$,
    \item for each~$a\in\hat I\cap\attneq{t}$, we have~$a\in\hat{\mathcal{P}}$, and
    \item $a\in\hat{\mathcal{P}}$ if and only if~$a$ is proven using program~$\progt{t}$ and ordering~$\hat\sigma$.
  \end{enumerate}
\end{definition}

Next, we observe that the $\PRIM$-solutions up to~$n$ suffice to capture all the answer sets.

\begin{proposition}\label{prop:hcfglobal}
The set of~$\PRIM$-solutions up to~$n$ characterizes the set of answer sets of~$\prog$.
In particular: $\{\hat I \mid  \langle \hat I, \hat{\mathcal{P}}, \hat\sigma \rangle \text{ is a } \PRIM\text{-solution up to }n\} = \{I \mid I \text{ is an answer set of }\prog\}$.
\end{proposition}
\begin{proof}
Observe that Definition~\ref{def:globalhcf} for root node~$t=n$ indeed suffices for~$\hat I$ to be a model of~$\progt{n}=\prog$,
and, moreover, every atom in~$\hat I=\hat P$ is proven in~$\prog$ by ordering~$\hat\sigma$.
\end{proof}

\begin{definition}\label{def:localhcf}
Let~$\langle \hat I, \hat{\mathcal{P}}, \hat\sigma\rangle$ be a~$\PRIM$-solution up to~$t$. Then, $\langle \hat I \cap \chi(t), \hat{\mathcal{P}} \cap \chi(t), \sigma \rangle$, where~$\sigma$ is the partial ordering of~$\hat\sigma$ only containing~$\chi(t)$, is referred to as~\emph{$\PRIM$-row solution at node~$t$}.
\end{definition}

Given a~$\PRIM$-solution~$\vec{\hat\tabval}$ up to~$t$ and a~$\PRIM$-row solution~$\vec\tabval$ at~$t$.
We say~$\vec{\hat\tabval}$ is a \emph{corresponding} $\PRIM$-solution up to~$t$ of~$\PRIM$-row solution at~$t$ if~$\vec{\hat\tabval}$ can be used
to construct~$\vec\tabval$ according to Definition~\ref{def:localhcf}.

In fact,~\emph{$\PRIM$-row solutions} at~$t$ suffice to capture all the answer sets of~$\prog$.
Before we show that, we need the following definition.

\begin{definition}
Let $t\in N$ be a node of~$\TTT$
  with~$\children(t) = \langle t_1, \ldots, t_\ell \rangle$.
Further, let~$\vec{\hat\tabval}=\langle \hat I, \hat{\mathcal{P}}, \hat\sigma\rangle$ be a~$\PRIM$-solution up to~$t$ 
and~$\vec{\hat v}=\langle \hat{I'}, \hat{\mathcal{P}'}, \hat{\sigma'} \rangle$ be a~$\PRIM$-solution up to
$t_i$. Then,~$\vec\tabval$ is \emph{compatible with~$\vec v$} (and vice-versa) if
	\begin{enumerate}
		\item $\hat{I'} = \hat{I}\cap \att{t_i}$
		\item $\hat{\mathcal{P}'} = \hat{\mathcal{P}}\cap \att{t_i}$
		\item $\hat{\sigma'}$ is a sub-sequence of~$\hat\sigma$ such that~$\hat\sigma$ may additionally contain atoms in~$\att{t}\setminus\att{t_i}$
	\end{enumerate}
\end{definition}

\begin{lemma}[Soundness]\label{lem:paspcorrect}
  Let $t\in N$ be a node of~$\TTT$
  with~$\children(t,T) = \langle t_1, \ldots, t_\ell \rangle$.
  Further, let $\vec v_i$ be a~$\PRIM$-row solution at~$t_i$ for~$1\leq i\leq \ell$.
  Then, each row~$\vec\tabval = \langle I, \mathcal{P}, \sigma \rangle$ in~$\tau(t)$ 
  with~$\langle \vec v_1, \ldots, \vec v_\ell \rangle \in \origa{\PRIM}(t, \vec\tabval)$ is also a~$\PRIM$-row solution at
  node~$t$. Moreover, for any corresponding~$\PRIM$-solution~${\vec{\hat\tabval}}$ up to~$t$ (of~$\vec\tabval$)
  there are corresponding \emph{compatible}~$\PRIM$-solutions~$\vec{\hat{v_i}}$ up to~$t_i$ (for~$\vec v_i$). 
\end{lemma}
\begin{proof}[Proof (Sketch)]
We proceed by case distinctions.
Assume case(i):~$\type(t)=\leaf$. Then, $\langle \emptyset, \emptyset, \langle \rangle \rangle$ is a~\PRIM-row solution at~$t$. This concludes case(i).

Assume case(ii):~$\type(t)=\intr$ and~$\chi(t)\setminus\chi(t')=\{a\}$. Let~$\vec v_1=\langle I, \mathcal{P}, \sigma\rangle$ be any \PRIM-row solution at child node~$t_1$,
and~$\vec{\hat{v_1}}=\langle \hat I, \hat{\mathcal{P}}, \hat\sigma\rangle$ be any corresponding \PRIM-solution up to~$t_1$, which exists by Definition~\ref{def:localhcf}. In the following, we show that the way~\PRIM transforms
\PRIM-row solution~$\vec v_1$ at~$t_1$ to a \PRIM-row solution~$\vec\tabval=\langle I', \mathcal{P}', \sigma'\rangle$ at~$t$ is sound.
We identify several sub-cases.

Case (a): Atom~$a\not\in I'$ is set to false. Then, \PRIM constructs~$\vec\tabval$ where~$I'=I, \sigma'=\sigma$ and~$\mathcal{P}'=\mathcal{P}\cup \gatherproof(I',\sigma', \prog_t)$. Note that by construction~$I'\models \prog_t$.
Towards showing soundness, we define how to transform~$\vec{\hat{v_1}}$ into~$\vec{\hat\tabval}$ such that~$\vec{\hat\tabval}$ is indeed the corresponding~$\PRIM$-solution up to~$t$ of row~$\vec{\tabval}$ constructed by~\PRIM. To this end, we define~$\vec{\hat\tabval}$ as follows: $\vec{\hat\tabval} = \langle \hat I, \hat{\mathcal{P}} \cup  \gatherproof(I',\sigma', \prog_t), \hat\sigma\rangle$. Observe that~$\vec{\hat\tabval}$ is a~\PRIM-solution up to~$t$ according to Definition~\ref{def:globalhcf}.
Moreover, by construction and Definition~\ref{def:localhcf}, $\vec{\hat\tabval}$ is a corresponding~$\PRIM$-solution up to~$t$ of~$\hat\tabval$. 
It remains to show, that indeed for any
corresponding~$\PRIM$-solution~${\vec{\hat\tabval}}=\langle \hat {I'},
\hat{\mathcal{P}'}, \hat{\sigma'} \rangle$ up to~$t$
(of~$\vec\tabval$, there is a
corresponding~$\PRIM$-solution~$\vec{\hat{\zeta_1}}$ up to~$t_1$
(of~$\vec{{v_1}}$). 
To this end, we
define~$\vec{\hat{\zeta_1}}=\langle \hat{I'}, \hat{\mathcal{P}'}
\setminus (\mathcal{P}'\setminus\mathcal{P}), \hat{\sigma'}\rangle$
that is by construction according to Definition~\ref{def:globalhcf}
indeed a corresponding~$\PRIM$-solution up to~$t_1$
of~$\vec{\hat{v_1}}$.  This concludes case (a).

Case (b): Atom~$a\in I'$ is set to true. Conceptually, the case works analogously. 
This concludes cases (b) and (ii). 

The remaining cases for nodes~$t$ with~$\type(t)=\rem$ (slightly easier) and nodes~$t$ with~$\type(t)=\join$, 
where we need to consider \PRIM-row solutions at two different child nodes of~$t$, go through similarly.
\end{proof}


\begin{lemma}[Completeness]\label{lem:primcomplete}
  Let~$t\in N$ be node of~$\TTT$ where
  $\type(t) \neq \leaf$ and~$\children(t,T) = \langle t_1, \ldots, t_\ell \rangle$. Given a
  $\PRIM$-row solution~$\vec\tabval=\langle I, \mathcal{P}, \sigma \rangle$ at node~$t$,
  and any corresponding~$\PRIM$-solution~$\vec{\hat\tabval}$ up to~$t$ (of~$\vec\tabval$).
  Then, there exists $\vec s=\langle {v_1}, \ldots, {v_\ell}\rangle$ where ${v_i}$ is a
  $\PRIM$-row solution at~$t_i$ 
  such that~$\vec s\in\origa{\PRIM}(t,\vec\tabval)$,
  and corresponding~$\PRIM$-solution~$\vec{\hat{v_i}}$ up to~$t_i$ (of~$v_i$) that is
  compatible with~$\vec{\hat\tabval}$.
\end{lemma}
\begin{proof}[Proof (Idea)]
Since~$\vec\tabval$ is a~\PRIM-row solution at~$t$, there is by Definition~\ref{def:localhcf} a corresponding~\PRIM-solution~$\vec{\hat\tabval}=\langle \hat I, \hat{\mathcal{P}}, \hat\sigma\rangle$ up to~$t$. 

We proceed again by case distinction. Assume that~$\type(t)=\intr$. Then we define~$\vec{\hat{v_1}}\eqdef \langle \hat I \setminus \{a\}, \hat{\mathcal{P}'}, \hat{\sigma'}\rangle$,
where~$\hat{\sigma'}$ is a sub-sequence of~$\hat\sigma$ that does not contain~$a$ and~$\hat{\mathcal P}'=\gatherproof(\hat I \setminus \{a\}, t_1, \progt{t_1})$. 
Observe that all the conditions of Definition~\ref{def:globalhcf} are met and that~$\hat{\mathcal P}'\subseteq \hat{\mathcal{P}'}$. Then, we can easily define \PRIM-row solution~$\vec{v_1}$ at~$t_1$ according to Definition~\ref{def:localhcf} by using~$\vec{\hat{v_1}}$. By construction of~$\vec{\hat{v_1}}$ and by the definition of~$\gatherproof$, we conclude that~$\vec\tabval$ can be constructed with~$\PRIM$
using~$\vec{v_1}$. Moreover, \PRIM-solution~$\vec{\hat{v_1}}$ up to~$t_1$ is indeed compatible with~$\vec{\hat\tabval}$.

Assume that~$\type(t)=\rem$. The case is slightly easier as the one above, and the remainder works similar.

Similarly, one can show the result for the remaining node with~$\type(t)=\join$, but define \PRIM-row solutions for two preceding child nodes of~$t$.
\end{proof}

We are now in the position to proof our theorem.

\begin{restatetheorem}[thm:primcorrectness]%
\begin{theorem}
  The algorithm~$\dpa_\PRIM$ is correct. \\
  More precisely, 
  %
  the algorithm~$\dpa_\PRIM((\prog,\cdot),\TTT)$ returns
  $\PRIM$-TTD~$(T,\chi,\tau)$ such that we can decide consistency of~$\prog$ and even reconstruct the answer sets of~$\prog$:
  \begin{align*}
	&\mathcal{I}(\Ext_{\leq n}[\tau(n)])=
	\{\hat I \mid  \langle \hat I, \hat{\mathcal{P}}, \hat\sigma \rangle \text{ is a } \PRIM\text{-solution up to }n\}\\
	&=\{I \mid I \in \ta{\at(\prog)}, I \text{ is an answer set of }\prog\}.
  \end{align*}
\end{theorem}
\end{restatetheorem}
\begin{proof}[Proof (Idea).]
  %
  By Lemma~\ref{lem:paspcorrect} we have soundness for every
  node~$t \in N$ and hence only valid rows as output of table
  algorithm~$\PRIM$ when traversing the tree decomposition in
  post-order up to the root~$n$.
  By Proposition~\ref{prop:hcfglobal} we then know that we can reconstruct answer sets
  given~\PRIM-solutions up to~$n$.
  In more detail, we proceed by means of induction. 
  For the induction base we only store~\PRIM-row solutions~$\vec\tabval\in\tau(t)$ at a certain node~$t$ starting at the leaves.
  For nodes~$t$ with~$\type(t)=\leaf$, obviously there is only the following (one)~\PRIM-row solution at~$t$: $\vec\tabval=\langle \emptyset, \emptyset, \langle \rangle\rangle$.
  
  Then, by Lemma~\ref{lem:paspcorrect} we establish the induction step, since algorithm~\PRIM only creates~\PRIM-row solutions at every node~$t$,
  assuming that it gets~\PRIM-row solutions at~$t_i$ for every child node~$t_i$ of~$t$.
  As a result, if there is no answer set of~$\prog$, the table~$\tau(n)$ is empty.
  On the other hand, if there is an answer set of~$\prog$, we obtain a~\PRIM-row solution~$\vec\tabval$ at root node~$n$, 
  for which by Definition~\ref{def:localhcf} a corresponding~\PRIM-solution~$\vec{\hat\tabval}$ up to~$n$ exists.
  Further, in the induction step we ensured that~\PRIM-solutions up to~$t$ for every~\PRIM-row solution at~$t$ for every node~$t\in N$ can be found that are compatible to~$\vec{\hat\tabval}$. In other words, by keeping track of corresponding origin~\PRIM-row solutions of~$\vec\tabval$ we can combine interpretation positions~$\mathcal{I}(\cdot)$ of rows by following origin rows top-down in order to reconstruct only valid answer set.
  %
  %
  
  %
  %

  %
  %

  Next, we establish completeness by induction starting from the
  root~$n$. Let therefore,
  $\hat\rho=\langle \hat I, \hat{\mathcal{P}}, \hat\sigma \rangle$ be
  the~\PRIM-solution up to node~$n$. If~$\hat\rho$ does not exist for
  node~$n$, there is by definition no answer set
  of~$\prog$. Otherwise, by Definition~\ref{def:localhcf}, we know
  that for the root~$n$ we can construct \PROJ-row solutions at~$n$ of
  the form~$\rho=\langle\emptyset, \emptyset, \langle \rangle\rangle$
  for~$\hat\rho$.  We already established the induction step in
  Lemma~\ref{lem:primcomplete} using~$\rho$ and~$\hat\rho$. As a
  consequence, we can reconstruct exactly \emph{all the answer sets}
  of~$\prog$ by following origin rows (see
  Definition of~$\orig$) back to the leaves and combining
  interpretation parts~$\mathcal{I}(\cdot)$, accordingly.
  Hence, we obtain some (corresponding) rows for every
  node~$t$. Finally, we stop at the leaves.

  In consequence, we have shown both soundness and completeness. As a
  result, Theorem~\ref{thm:primcorrectness} is sustains.
\end{proof}

\begin{corollary}\label{cor:primcorrectness}
  %
  %
  Algorithm~$\dpa_\PRIM((\prog,\cdot),\TTT)$ returns
  $\PRIM$-TTD~$(T,\chi,\tau)$ such that:
  \begin{align*}
    &\mathcal{I}(\PExt_{\leq t}[\tau(t)])\\
    &=\{\hat I \mid  \langle \hat I, \hat{\mathcal{P}}, \hat\sigma \rangle \text{ is a } \PRIM\text{-solution up to }t, \text{ there is answer set }\\
    &\quad\;\;\, I' \supseteq \hat I \text{ of } \prog \text{ such that }
      I' \subseteq I \cup (\at(\prog) \setminus \att{t})\}\\
    &=\{I \mid I \in \ta{\att{t}}, I \models \prog_{\leq t}, \text{ there is an answer set }\\
    &\quad\;\;\, I' \supseteq I \text{ of } \prog \text{ such that } I'\subseteq I \cup (\at(\prog) \setminus \att{t})\}.
  \end{align*}
\end{corollary}
\begin{proof}
The corollary follows from the proof of Theorem~\ref{thm:primcorrectness} applied up to node~$t$ and by considering only rows that are involved in reconstructing answer sets (see Definition~\ref{def:satext}).
\end{proof}


\subsection{Correctness of~$\mdpa{\AlgA}$: Omitted proofs}

In the following, we assume~$(\prog, P)$ to be an instance of~$\PASP$. Further, let~$\mathcal{T}=(T,\chi,\tau)$
be an~$\AlgA$-TTD of~$\mathcal{G}_\prog$ where~$T=(N,\cdot,n)$, node~$t\in N$, and~$\rho\subseteq\tau(t)$.

\begin{definition}\label{def:asplocalsol}
Table algorithm $\AlgA$ is referred to as \emph{admissible}, if for each row $\vec{u_{t.i}}\in\tau(t)$ of any node~$t\in T$ the following holds:
  \begin{enumerate}
    \item $\mathcal{I}(\vec{\tabval_{t.i}}) \subseteq \chi(t)$
    \item For any $\vec v \in \tau(t')$, $\vec w \in \tau(t'')$ we have $\mathcal{I}(\vec v) \cap \chi(t') \cap \chi(t'') = \mathcal{I}(\vec w) \cap \chi(t') \cap \chi(t'')$
    \item $\mathcal{I}(\PExt_{\leq t}[\tau(t)]) = \{I \mid I \in
    \ta{\att{t}}, I \models \prog_{\leq t}, \text{ there is an answer set } I \cup (\at(\prog) \setminus \att{t}) \supseteq  I' \supseteq I \text{ of } \prog\}$
    \item If~$t=n$ or~$\type(t)=\leaf$: $\Card{\local(t,\PExt_{\leq t}[\tau(t)])} \leq 1$
  \end{enumerate}
\end{definition}

Note that the last condition is not a hard restriction, since the bags of the leaf and root nodes of a tree decomposition are defined to be empty anyway. However, it rather serves as technical trick simplifying proofs.

\begin{observation}
Table algorithms~$\PRIM$ and~$\algo{PRIM}$ are admissible.
\end{observation}
\begin{proof}
  Obviously, Conditions 1, 2, and 4 hold by construction of the table algorithms and by properties auf tree decompositions. For condition 3, we have to check for correctness and completeness, which has been shown~\cite{FichteEtAl17a} for algorithm~$\algo{PRIM}$. For~$\PRIM$, see Theorem~\ref{thm:primcorrectness} and Corollary~\ref{cor:primcorrectness}.
\end{proof}

In the following, we assume that whenever~$\AlgA$ occurs, $\AlgA$ is an admissible table algorithm.

\begin{proposition}\label{prop:sat}
$\mathcal{I}(\PExt_{\leq n}[\tau(n)]) = \mathcal{I}(\Exts) = \{I \mid I \in
    \ta{\at(\prog)}, I \text{ is an answer set of } \prog\}.$
\end{proposition}
\begin{proof}
  Fill in Definition~\ref{def:asplocalsol} with root~$n$ of $\AlgA$-TTD ${\cal T}$.
\end{proof}

The following definition is key for the correctness proof, since later we show that these are equivalent with the result of~$\dpa_\PROJ$ using purged table mapping~$\nu$.

\begin{definition}\label{def:pmc}
  %
  The \emph{projected answer sets count} $\pmc_{\leq t}(\rho)$ of
  $\rho$ below~$t$ is the size of the union over projected
  interpretations of the satisfiable extensions of~$\rho$ below~$t$,
  formally,
  $\pmc_{\leq t}(\rho) \eqdef \Card{\bigcup_{\vec u\in\rho}
    \mathcal{I}_P(\PExt_{\leq t}(\{\vec u\}))}$.

  The \emph{intersection projected answer sets count}
  $\ipmc_{\leq t}(\rho)$ of $\rho$ below~$t$ is the size of the
  intersection over projected interpretations of the satisfiable
  extensions of~$\rho$ below~$t$,~i.e.,
  $\ipmc_{\leq t}(\rho) \eqdef \Card{\bigcap_{\vec u\in\rho}
    \mathcal{I}_P(\PExt_{\leq t}(\{\vec u\}))}$.
\end{definition}

In the following, we state definitions required for the correctness
proofs of our algorithm \PROJ.  In the end, we only store rows that
are restricted to the bag content to maintain runtime bounds. 
We define the
content of our tables in two steps. First, we define the properties of
so-called \emph{$\PROJ$-solutions up to~$t$}. Second, we restrict
these solutions to~\emph{$\PROJ$-row solutions} at~$t$.

\begin{definition}\label{def:globalsol}
  Let~$\emptyset \subsetneq \rho \subseteq \tau(t)$ be a
  table with $\rho \in \subbuckets_P(\tau(t))$.
  We define a \emph{${\PROJ}$-solution up to~$t$} to be the sequence
  $\langle \hat {\rho}\rangle = \langle\PExt_{\leq t}(\rho)\rangle$.
\end{definition}

Before we present equivalence results between~$\ipmc_{\leq t}(\ldots)$
and the recursive version~$\ipmc(t, \ldots)$
used during the computation of
$\dpa_\PROJ$, recall that~$\ipmc_{\leq t}$ and~$\pmc_{\leq t}$
(Definition~\ref{def:pmc}) are key to compute the projected answer sets
count. The following corollary states that computing $\ipmc_{\leq n}$
at the root~$n$ actually suffices to compute~$\pmc_{\leq n}$, which is
in fact the projected answer sets count of the input program.

\begin{corollary}\label{cor:psat}
  \begin{align*}
    &\ipmc_{\leq n}(\local(n,\PExt_{\leq n}[\tau(n)]))\\
 =& \pmc_{\leq n}(\local(n,\PExt_{\leq n}[\tau(n)]))\\
    =& \Card{\mathcal{I}_P(\PExt_{\leq n}[\tau(n)])}\\
    =& \Card{\mathcal{I}_P(\Exts)}\\
    =& \,|\{J \cap P
       \mid J \in \ta{\at(\prog)},\\
    & J \text{ is an answer set of } \prog\}|
  \end{align*}
\end{corollary}
\begin{proof}
  The corollary immediately follows from Proposition~\ref{prop:sat}
  and since the cardinality of $\local(n,\PExt_{\leq n}[\tau(n)])$ is at
  most one at root~$n$, by Definition~\ref{def:asplocalsol}.
\end{proof}

The following lemma establishes that the \PROJ-solutions up to
root~$n$ of a given tree decomposition solve the \PASP problem.

\begin{lemma}\label{lem:global}
  The
  value~$c = \sum_{\langle\hat{\rho}\rangle\text{ is a \PROJ-solution
      up to } n}\Card{\mathcal{I}_P(\hat{\rho})}$ corresponds to the
  projected answer sets count of~$\prog$ with respect to the set~$P$ of
  projection atoms.
\end{lemma}
\begin{proof}
  (``$\Longrightarrow$''): Assume
  that~$c = \sum_{\langle\hat{\rho}\rangle\text{ is a \PROJ-solution
      up to } n}\Card{\mathcal{I}_P(\hat {\rho})}$. Observe that there can be at
  most one projected solution up to~$n$ by Definition~\ref{def:asplocalsol}. %
  If~$c=0$, then $\tau(n)$ contains no rows. Hence, $\prog$ has no
  answer sets,~c.f., Proposition~\ref{prop:sat}, and obviously also no
  answer sets projected to~$P$. Consequently, $c$ is the projected answer sets
  count of~$\prog$.  
  If~$c>0$ we have by Corollary~\ref{cor:psat} that~$c$ is
  equivalent to the projected answer sets count of~$\prog$ with respect to~$P$.

  (``$\Longleftarrow$''): The proof proceeds similar to the only-if
  direction.
\end{proof}

\medskip %

In the following, we provide for a given node~$t$ and a given \PROJ-solution up to~$t$,
the definition of a \PROJ-row solution at~$t$.

\begin{definition}\label{def:loctab}~
  %
  %
  %
%
  Let 
  $\langle \hat{\rho} \rangle$ be a~$\PROJ$-solution up to~$t$. Then, we
  define the \emph{$\PROJ$-row solution at $t$} by
  $\langle \local(t,\hat{\rho}), \Card{\mathcal{I}_P(\hat{\rho})}\rangle$.
\end{definition}






\begin{observation}\label{obs:unique}
  Let $\langle \hat {\rho}\rangle$ be a \PROJ-solution up to a
  node~$t\in N$.  There is exactly one corresponding \PROJ-row
  solution
  $\langle \local(t,\hat{\rho}), \Card{\mathcal{I}_P(\hat{\rho})}\rangle$ at~$t$.

  Vice versa, let $\langle \rho, c\rangle$ at~$t$ be a \PROJ-row
  solution at~$t$ for some integer~$c$. Then, there is exactly one
  corresponding \PROJ-solution~$\langle\PExt_{\leq t}(\rho)\rangle$
  up to~$t$.
\end{observation}

We need to ensure that storing~$\PROJ$-row solutions at a
node~$t \in N$ suffices to solve the~\PASP problem, which is necessary
to obtain the runtime bounds as presented in
Corollary~\ref{cor:runtime}. For the root node~$n$, this is sufficient, shown in the following.

\begin{lemma}\label{lem:local}
  There is a
  \PROJ-row solution at the root~$n$ if and only if the projected
  answer sets count of~$\prog$ is larger than zero. Further, if there is a \PROJ-row solution~$\langle \rho, c\rangle$ at root~$n$, then~$c$ is the projected answer sets count of~$\prog$.
\end{lemma}
\begin{proof}%

  (``$\Longrightarrow$''): Let $\langle \rho, c\rangle$ be a
  \PROJ-row solution at root~$n$ where $\rho$ is an $\AlgA$-table and
  $c$ is a positive integer. Then, by Definition~\ref{def:loctab}
  there also exists a
  corresponding~$\PROJ$-solution~$\langle \hat{\rho} \rangle$ up
  to~$n$ such that $\rho = \local(n,\hat{\rho})$ and
  $c=\Card{\mathcal{I}_P(\hat{\rho})}$.
  Moreover, by Definition~\ref{def:asplocalsol}, we
  have~$\Card{\local(n,\PExt_{\leq n}[\tau(n)])}=1$.  
  Then, by Definition~\ref{def:globalsol},
  $\hat{\rho} = \PExt_{\leq n}[\tau(n)]$. By Corollary~\ref{cor:psat}, we
  have $c=\Card{\mathcal{I}_P(\PExt_{\leq n}[\tau(n)])}$ equals the projected answer sets count of~$\prog$.
  Finally, the claim follows.
  

  (``$\Longleftarrow$''): The proof proceeds similar to the only-if
  direction.
\end{proof}

Before we show that \PROJ-row solutions suffice, we require the following lemma.

\begin{observation}\label{obs:main_incl_excl}
  Let $n$ be a positive integer, $X = \{1, \ldots, n\}$, and $X_1$,
  $X_2$, $\ldots$, $X_n$ subsets of $X$.
  The number of elements in the intersection over all sets~$A_i$ is
    \[\Card{\bigcap_{i \in X} X_i} 
    = %
       \Bigg|\,\Card{\bigcup^n_{j = 1} X_j} 
                                         + \sum_{\emptyset \subsetneq I \subsetneq X} (-1)^{\Card{I}} 
                                              \Card{\bigcap_{i \in I} X_i}\,\Bigg|.\]
\end{observation}
\begin{proof}
  We take the well-known inclusion-exclusion
  principle~\cite{GrahamGrotschelLovasz95a} and rearrange the
  equation.
\end{proof}

\begin{lemma}\label{lem:main_incl_excl}
  Let $t\in N$ be a node of~$\TTT$
  with~$\children(t,T) = \langle t_1, \ldots, t_\ell \rangle$ and let
  $\langle\rho,\cdot\rangle$ be a~\PROJ-row solution at~$t$. Further, let~$\pi$ be a partial mapping of~$\pi'$ (finally returned by~$\dpa_\PROJ((\prog,P),\TTT)=(T,\chi,\pi')$), which maps nodes of the sub-tree~$T[t]$ rooted at~$t$ (excluding~$t$) to~$\PROJ$-tables.
  Then,
  \begin{enumerate}
  \item %
    $\ipmc(t,\rho,\langle\pi(t_1), \ldots,
    \pi(t_\ell)\rangle) = \ipmc_{\leq t}(\rho)$
  \item \smallskip%
    for $\type(t) \neq \leaf$:\\
    $\pmc(t,\rho,\langle\pi(t_1), \ldots,
    \pi(t_\ell)\rangle) = \pmc_{\leq t}(\rho)$.
  \end{enumerate}
\end{lemma}
\begin{proof}[Sketch]
  We prove the statement by simultaneous induction.
  
  (``Induction Hypothesis''): Lemma~\ref{lem:main_incl_excl} holds for the nodes in~$\children(t,T)$ and also for node~$t$, but on strict subsets~$\varphi\subsetneq\rho$.
  (``Base Cases''): Let $\type(t) = \leaf$.
  Then by definition,
  $\ipmc(t,\{\langle \emptyset, \ldots\rangle\}, \langle \rangle) = \ipmc_{\leq t}(\{\langle\emptyset,\ldots\rangle\}) =
  1$.  
  Recall that for $\pmc$ the equivalence does not hold for leaves, but we use a node
  that has a node~$t'\in N$ with~$\type(t') = \leaf$ as child for the
  base case. Observe that by definition of a tree decomposition
  such a node~$t$ can have exactly one child.
  Then, we have that
  $\pmc(t,\rho,\langle\pi(t')\rangle) = \sum_{\emptyset
    \subsetneq O \subseteq {\origs(t,\rho)}} (-1)^{(\Card{O} - 1)}
  \cdot \sipmc(\langle \tau(t')\rangle, O) =
  \Card{\bigcup_{\vec u\in\rho} \mathcal{I}_P(\PExt_{\leq t}(\{\vec u\}))} =
  \pmc_{\leq t}(\rho) = 1$ where $\langle\rho,\cdot\rangle$ is
  a~\PROJ-row solution at~$t$.

  (``Induction Step''): We proceed by case distinction.

  Assume that $\type(t) = \intr$.
  Let $a \in (\chi(t) \setminus \chi(t'))$ be an introduced
  atom. We have two cases. Case (i) $a$ also belongs to
  $(\at(\prog) \setminus P)$,~i.e., $a$ is not a projection atom; and
  Case (ii) $a$ also belongs to $P$,~i.e., $a$ is a projection
  atom.
  Assume that we have Case~(i).
  Let~$\langle \rho, c \rangle$ be a \PROJ-row solution at~$t$ for
  some integer~$c$. As a consequence of admissible algorithm~$\AlgA$ (see Definition~\ref{def:asplocalsol})
  there can be many rows in the table~$\tau(t)$ for one row in
  the table~$\tau(t')$, more precisely,
  $\Card{\buckets_P(\rho)} = 1$.
  As a result,
  $\pmc_{\leq t}(\rho) = \pmc_{\leq t'}(\orig(t,\rho))$ by
  applying Observation~\ref{obs:unique}.
  We apply the inclusion-exclusion principle on every subset~$\varphi$ of
  the origins of~$\rho$ in the definition of~$\pmc$ and by induction
  hypothesis we know that
  $\ipmc(t',\varphi,\langle\pi(t')\rangle) = \ipmc_{\leq
    t'}(\varphi)$, therefore,
  $\sipmc(\pi(t'), \varphi) = \ipmc_{\leq t'}(\varphi)$.  This
  concludes Case~(i) for $\pmc$. The induction step for $\ipmc$ works
  similar 
  by applying
  Observation~\ref{obs:main_incl_excl} and comparing the corresponding
  \PROJ-solutions up to~$t$ or $t'$, respectively. 
  Further, for showing the lemma for~$\ipmc$, one has to additionally apply the hypothesis for node~$t$, but on strict subsets~$\emptyset\subsetneq\varphi\subsetneq\rho$ of~$\rho$.
  %
  Assume that we have Case~(ii). We proceed similar as in Case~(i),
  since Case~(ii) is just a special case here, more precisely, we also
  have $\Card{\buckets_P(\rho)} = 1$ here.

  Assume that $\type(t) = \rem$. Let
  $a \in (\chi(t') \setminus \chi(t))$ be a removed atom. We have
  two cases. Case (i) $a$ also belongs to
  $(\at(\prog) \setminus P)$,~i.e., $a$ is not a projection atom; and
  Case (ii) $a$ also belongs to $P$,~i.e., $a$ is a projection
  atom.
  Assume that we have Case~(i).  Let~$\langle \rho, c \rangle$ be a
  \PROJ-row solution at~$t$ for some integer~$c$.
  As a consequence of admissible table algorithms~$\AlgA$ (see Definition~\ref{def:asplocalsol}) there can be many rows
  in the table~$\tau(t)$ for one row in the
  table~$\tau(t')$ (and vice-versa). Nonetheless we still have
  $\pmc_{\leq t}(\rho) = \pmc_{\leq t'}(\orig(t,\rho))$, because
  $a \notin P$ by applying Observation~\ref{obs:unique}.
  We apply the inclusion-exclusion principle on every subset~$\varphi$ of
  the origins of~$\rho$ in the definition of~$\pmc$ and by induction
  hypothesis we know that
  $\ipmc(t',\varphi,\langle\pi(t')\rangle) = \ipmc_{\leq
    t'}(\varphi)$, therefore,
  $\sipmc(\pi(t'), \varphi) = \ipmc_{\leq t'}(\varphi)$.  This
  concludes Case~(i) for $\pmc$. Again, the induction step for $\ipmc$
  works similar, but swapped.
  Assume that we have Case~(ii).
  Let~$\langle \rho, c \rangle$ be a \PROJ-row solution at~$t$ for
  some integer~$c$.
  Here we cannot ensure
  $\pmc_{\leq t}(\rho) = \pmc_{\leq t'}(\orig(t,\rho))$, since
  buckets fall together.  However, by applying
  Observation~\ref{obs:unique} we have
  $\pmc_{\leq t}(\rho) = \sum_{\varphi \in
    \buckets_P(\origs(t,\rho)_{(1)})} \pmc(t', \varphi, C) $ where the
  sequence~$C$ consists of the tables~$\pi(t'_i)$ of the children~$t'_i$ of~$t'$.
  For every~$\varphi \in \subbuckets_P(\origs(t,\rho)_{(1)})$ by
  induction hypothesis we know that
  $\ipmc(t',\varphi,\langle\pi(t')\rangle) = \ipmc_{\leq
    t'}(\varphi)$.
  Hence, we apply the inclusion-exclusion principle over all
  subsets~$\zeta$ of~$\varphi$ for all~$\varphi$ independently.  By
  construction
  $\sipmc(\pi(t'), \zeta) = \ipmc_{\leq t'}(\zeta)$.  Then,
  by construction
  $\pcnt(t,\rho, C') = \sum_{\emptyset \subsetneq O \subseteq
    {\origs(t,\rho)}} (-1)^{(\Card{O} - 1)} \cdot \sipmc(C', O) =
  \pmc_{\leq t}(\rho)$ where
  $C' = \langle \pi(t') \rangle$, since for the remaining
  terms $\sipmc(C', O)$ is simply zero, including cases where
  different buckets are involved.
  This concludes Case~(ii) for $\pmc$. Again, the induction step for
  $\ipmc$ works similar, but swapped by again applying
  Observation~\ref{obs:main_incl_excl}.

  Assume that $\type(t) = \join$. We proceed similar to the introduce
  case. However, we have two \PROJ-tables for the children of~$t$.
  Hence, we have to consider both sides when computing $\sipmc$
  (see Definition of~$\sipmc$). 
  There we consider the
  cross-product of two \AlgA-tables and we can also correctly apply
  the inclusion-exclusion principle on subsets of this cross-product,
  which we can do by simply multiplying $\sipmc$-values
  accordingly. The multiplication is closely related to the join case
  in table algorithm~\AlgA. For $\ipmc$ this does not apply, since the
  inclusion-exclusion principle is carried out at the node~$t$ and not
  for its children.

  Since we outlined all cases that can occur for node~$t$, this
  concludes the proof sketch.
\end{proof}


\begin{lemma}[Soundness]\label{lem:correct}
  Let $t\in N$ be a node of~$\TTT$
  with~$\children(t,T) = \langle t_1, \ldots, t_\ell \rangle$.
  Then, each row~$\langle \rho, c \rangle$ at node~$t$ constructed
  by table algorithm~$\PROJ$ is also a~\PROJ-row solution for
  node~$t$.
\end{lemma}
\begin{proof}[Idea]
  Observe that Listing~\ref{fig:dpontd3} computes a row for each
  sub-bucket $\rho \in \subbuckets_P(\local(t,\PExt_{\leq t}[\tau(t)]))$. The
  resulting row~$\langle\rho, c \rangle$ obtained by~$\ipmc$ is
  indeed a \PROJ-row solution for~$t$ according to
  Lemma~\ref{lem:main_incl_excl}.
\end{proof}


\begin{lemma}[Completeness]\label{lem:complete}
  Let~$t\in N$ be node of~$\TTT$ where
  $\type(t) \neq \leaf$ and~$\children(t,T) = \langle t_1, \ldots, t_\ell \rangle$. Given a
  \PROJ-row solution~$\langle \rho, c \rangle$ at node~$t$.
  There exists $\langle C_1, \ldots, C_\ell\rangle$ where $C_i$ is set
  of \PROJ-row solutions at~$t_i$ 
  such that
  $\rho \in \PROJ(t, \cdot, \tau(t), \cdot, P, \langle C_1, \ldots,
  C_\ell\rangle)$.
\end{lemma}
\begin{proof}[Idea]
Since~$\langle\rho,c \rangle$ is a~\PROJ-row solution for~$t$, there is by Definition~\ref{def:loctab} a corresponding ~\PROJ-solution~$\langle\hat\rho\rangle$ up to~$t$ such that~$\local(t,\hat\rho) = \rho$. 

We proceed again by case distinction. Assume that~$\type(t)=\intr$. Then we define~$\hat{\rho'}\eqdef \{(t',\hat\varphi) \mid (t', \hat\varphi)\in \rho, t \neq t'\}$. Then, for each subset~$\emptyset\subsetneq\varphi\subseteq\local(t',\hat{\rho'})$, we define~$\langle \varphi, \Card{\mathcal{I}_P(\PExt_{\leq t}(\varphi))}\rangle$ in accordance with Definition~\ref{def:loctab}. By Observation~\ref{obs:unique}, we have that~$\langle \varphi, \Card{\mathcal{I}_P(\PExt_{\leq t}(\varphi))}\rangle$ is an \AlgA-row solution at node~$t'$. 
Since we defined the~\PROJ-row solutions for~$t'$ for all the respective \PROJ-solutions up to~$t'$, we encountered every~\PROJ-row solution for~$t'$ that is required for deriving~$\langle \rho, c\rangle$ via~\PROJ (c.f., Definitions of~$\ipmc$ and of~$\pmc$). 

Assume that~$\type(t)=\rem$. The case is slightly easier as the one
above. We do not need to define a~\PROJ-row solution for~$t'$ for all
subsets~$\varphi$, since we only have to consider subsets~$\varphi$ here,
with~$\Card{\buckets_P(\varphi)}=1$. The remainder works similar.

Similarly, one can show the result for the remaining node with~$\type(t)=\join$, but define \PROJ-row solutions for two preceding child nodes of~$t$.
\end{proof}

We are now in the position to proof our theorem.

\begin{theorem}\label{thm:correctness}
  The algorithm~$\dpa_\PROJ$ is correct. \\
  More precisely, 
  %
  the algorithm~$\dpa_\PROJ((\prog,P),\TTT)$ returns
  $\PROJ$-TTD~$(T,\chi,\pi$) such that $c=\sipmc(\pi(n), \cdot)$
  is the projected answer sets count of~$\prog$ with respect to the set~$P$ of
  projection atoms.
\end{theorem}
\begin{proof}
  %
  By Lemma~\ref{lem:correct} we have soundness for every
  node~$t \in N$ and hence only valid rows as output of table
  algorithm~$\PROJ$ when traversing the tree decomposition in
  post-order up to the root~$n$.
  By Lemma~\ref{lem:local} we know that the projected answer sets count~$c$
  of~$\prog$ is larger than zero if and only if there exists a
  certain~\PROJ-row solution for~$n$.
  This~\PROJ-row solution at node~$n$ is of the
  form~$\langle \{\langle\emptyset, \ldots\rangle\} ,c\rangle$. If
  there is no \PROJ-row solution at node~$n$,
  then~$\tau(n)=\emptyset$ since the table algorithm~$\AlgA$
  is admissible (c.f., Proposition~\ref{prop:sat}). Consequently, we have
  $c=0$. Therefore, $c=\sipmc(\pi(n), \cdot)$ is the
  projected answer sets count of~$\prog$ with respect to~$P$ in both cases.
  %
  %
  
  %
  %

  %
  %

  Next, we establish completeness by induction starting from the
  root~$n$. Let therefore, $\langle \hat\rho \rangle$ be
  the~\PROJ-solution up to node~$n$, where for each row
  in~$\vec u\in \hat\rho$, $\mathcal{I}(\vec u)$ corresponds to an
  answer set of~$\prog$.  By Definition~\ref{def:loctab}, we know that
  for the root~$n$ we can construct a \PROJ-row solution at~$n$ of the
  form~$\langle \{\langle\emptyset, \ldots\rangle\} ,c\rangle$
  for~$\hat\rho$.  We already established the induction step in
  Lemma~\ref{lem:complete}.
  Hence, we obtain some (corresponding) rows for every
  node~$t$. Finally, we stop at the leaves.

  In consequence, we have shown both soundness and completeness. As a
  result, Theorem~\ref{thm:correctness} is sustains.
\end{proof}

\begin{corollary}\label{cor:correctness}
  The algorithm $\mdpa{\AlgA}$ is correct and outputs for any instance
  of \PASP its projected answer sets count.
\end{corollary}
\begin{proof}
  The result follows immediately, since~$\mdpa{\AlgA}$ consists of two
  dynamic programming passes~$\dpa_\AlgA$, a purging step, and~$\dpa_\PROJ$. For the
  soundness and completeness of~$\dpa_\algo{PRIM}$ we refer to other
  sources~\cite{FichteEtAl17a}. By Proposition~\ref{prop:sat}, the
  ``purging'' step does neither destroy soundness nor completeness
  of~$\dpa_\AlgA$.
\end{proof}

\begin{restateproposition}[prop:phcworks]
\begin{proposition}
  The algorithm $\mdpa{\PRIM}$ is correct and outputs for any instance
  of \PASP its projected answer sets count.
\end{proposition}
\end{restateproposition}
\begin{proof}
This is a direct consequence of Corollary~\ref{cor:correctness}.
\end{proof}

\begin{restateproposition}[prop:disjworks]
\begin{proposition}
  The algorithm $\mdpa{\algo{PRIM}}$ is correct and outputs for any instance
  of \PDASP its projected answer sets count.
\end{proposition}
\end{restateproposition}
\begin{proof}
This is a direct consequence of Corollary~\ref{cor:correctness}.
\end{proof}

\longversion{
\section{Correctness of QBF lower bound}

\begin{definition}[\cite{MarxMitsou16}]
  Given graph~$G=(V,E)$, integers~$i,j,r$ where~$i\leq j$ and total list-capacity function~$f: V \rightarrow \{i,\ldots,j\}$. Then an instance~$(G,r,f)$ of~$(i,j)$-Choosability Deletion asks for the existence of a set of vertices~$V'\subseteq V$ with~$\Card{V'}\leq r$ and~$V_1\eqdef V\setminus V'$, such that for all assignments~$\mathcal{L}: V_1 \rightarrow 2^{\{i,\ldots,j\}}$ with $\Card{\mathcal{L}(v)} = f(v)$ for all~$v\in V_1$, there is a coloring~$c: V_1 \rightarrow \mathcal{L}(v)$ such that for every edge~$(u,v)\in E\setminus (V_1\times V_1): c(u) \neq c(v)$.
\end{definition}

\begin{proposition}[\cite{MarxMitsou16}]
  Instances~$(G,r,f)$ of~$(1,4)$-Choosability Deletion where~$k$
  is the treewidth of~$G$, can not be
  solved in time~${2^{2^{2^{o(k)}}}}\cdot \CCard{G}^{o(k)}$.
\end{proposition}

\begin{restateproposition}[prop:lampis3]
\begin{proposition}
  QBFs of the form $\forall V_1.\exists V_2.\forall V_3. E$ where~$k$
  is the treewidth of the primal graph of DNF formula~$E$, can not be
  solved in time~${2^{2^{2^{o(k)}}}}\cdot \CCard{E}^{o(k)}$.
\end{proposition}
\end{restateproposition}
\begin{proof}
We proof the result by reducing from the problem~$(1,4)$-Choosability Deletion,
\end{proof}}
}

\end{document}